\documentclass{article}
\usepackage{amsmath}
\usepackage{amsthm}
\usepackage{amssymb}
\usepackage{dsfont}
\usepackage{mathrsfs}
\usepackage{mathabx}
\usepackage{mathtools}
\usepackage{stmaryrd}
\usepackage[all]{xy}
\usepackage{tikz}
\usepackage[utf8]{inputenc} 
\usepackage{hyperref}
\usetikzlibrary{arrows.meta}
\usetikzlibrary{bending}

 \newtheorem{theorem}{Theorem} [section]
 \newtheorem{lemma}[theorem]{Lemma}
 \newtheorem{corollary}[theorem]{Corollary}

 \newtheorem{definition}[theorem]{Definition}
 \newtheorem{example}[theorem]{Example}

\newcommand{\set}[1]{\{#1\}}
\newcommand{\setof}[2]{\{#1\mid #2\}}

\newcommand{\card}[1]{|#1|}
\newcommand{\tuple}[1]{(#1)}

\newcommand{\ensvide}{\varnothing}
\newcommand{\empstr}{\varepsilon}

\newcommand{\defeq}{\mathrel{\stackrel{\mathrm{\scriptscriptstyle def}}{=}}}
\newcommand{\Id}[1]{\mathrm{Id}_{#1}}
\newcommand{\restr}[3]{#1\vert_{#2}^{#3}}
\newcommand{\projl}{\pi_1}
\newcommand{\projr}{\pi_2}
\newcommand{\pair}[2]{\langle #1,#2\rangle}

\newcommand{\seqto}[2]{#1^{<#2}}
\newcommand{\edges}[1]{\mathsf{E}{#1}}
\newcommand{\occin}[2]{#1\mid #2}
\newcommand{\repet}[2]{#1 {\uparrow} #2}
\newcommand{\nodes}[1]{\mathrm{N}_{#1}}
\newcommand{\trace}[1]{\mathrm{tr}(#1)}

\newcommand{\invf}[1]{#1^{-1}}
\newcommand{\invrel}[1]{#1^{-1}}
\newcommand{\relim}[2]{#1[#2]}

\newcommand{\Sets}{\mathbf{Sets}}
\newcommand{\Graphs}{\mathbf{Graphs}}
\newcommand{\aCat}{\boldsymbol{A}}

\newcommand{\MonGr}{\mathbf{Monogr}}
\newcommand{\SMonGr}{\mathbf{SMonogr}}
\newcommand{\OMonGr}[1]{#1\mbox{-}\MonGr}
\newcommand{\OSMonGr}[1]{#1\mbox{-}\SMonGr}
\newcommand{\FMonGr}{\mathbf{FMonogr}}
\newcommand{\sliceCat}[2]{#1\setminus #2}
\newcommand{\sliceM}[1]{\sliceCat{\MonGr}{#1}}

\newcommand{\SigCat}{\mathbf{Sig}}
\newcommand{\SigCats}{\sorts{\SigCat}}
\newcommand{\MonSig}{\mathbf{GrStruct}}

\newcommand{\iso}{\simeq}
\newcommand{\equivCat}{\approx}

\newcommand{\Sig}{\varSigma}
\newcommand{\Op}{\varOmega}
\newcommand{\Msig}{\varGamma}
\newcommand{\Rng}[1]{\mathrm{Rng}(#1)}
\newcommand{\Dom}[1]{\mathrm{Dom}(#1)}
\newcommand{\GraphSig}{\Gamma_{\mathrm{g}}}
\newcommand{\EGraphSig}{\Gamma_{\mathrm{e}}}
\newcommand{\opn}[1]{#1_{\mathrm{opn}}}
\newcommand{\sorts}[1]{#1_{\mathrm{srt}}}

\newcommand{\id}[1]{\mathrm{1}_{#1}}
\newcommand{\SigFunc}[1]{\mathsf{S}#1}
\newcommand{\Funcs}[1]{\mathrm{\Omega}_{#1}}
\newcommand{\dotiso}{\mathrel{\dot{\iso}}}
\newcommand{\opname}[2]{\texttt{[}#1{\cdot}#2\texttt{]}}

\newcommand{\Alg}{\mathcal{A}}
\newcommand{\Blg}{\mathcal{B}}
\newcommand{\Clg}{\mathcal{C}}
\newcommand{\Glg}{\mathcal{G}}
\newcommand{\carrier}[2]{#1_{#2}}
\newcommand{\interp}[2]{#1^{#2}}
\newcommand{\AlgCat}[1]{#1\mbox{-}\mathbf{Alg}}

\newcommand{\AlgFunc}[2]{\mathsf{A}_{#1}#2}
\newcommand{\MFunc}[1]{\mathsf{M}#1}

\newcommand{\partm}[1]{\left\lceil #1\right\rceil}
\newcommand{\MonGrP}{\MonGr^{\mathbf{P}}}
\newcommand{\SMonGrP}{\SMonGr^{\mathbf{P}}}
\newcommand{\OMonGrP}[1]{#1\mbox{-}\MonGrP}
\newcommand{\OSMonGrP}[1]{#1\mbox{-}\SMonGrP}
\newcommand{\FMonGrP}{\FMonGr^{\mathbf{P}}}
\newcommand{\glucond}[2]{\mathrm{GC}(#1,#2)}
\newcommand{\dpo}[3]{\stackrel{\tuple{#1,#2}}{\Longrightarrow}_{#3}}
\newcommand{\spo}[2]{\stackrel{#1}{\Longrightarrow}_{#2}}
\newcommand{\partr}[2]{\partm{#1,#2}}

\newcommand{\tmorph}[1]{\vec{#1}}
\newcommand{\amorph}[1]{\dot{#1}}
\newcommand{\ATM}[2]{\mathbf{ATM}(#1,#2)}
\newcommand{\sigplus}{\dotplus}
\newcommand{\DFunc}[1]{\mathsf{D}#1}
\newcommand{\UFunc}[1]{\mathsf{U}#1}

\newcommand{\cubenodes}[8]{
\node (LL) at (0,0) {$#1$}; \node (LF) at (4,-1) {$#2$};
\node (LB) at (3,1) {$#3$}; \node (LR) at (7,0) {$#4$};
\node (UL) at (0,3.2) {$#5$}; \node (UF) at (4,2.2) {$#6$};
\node (UB) at (3,4.2) {$#7$}; \node (UR) at (7,3.2) {$#8$};}

\begin{document}
\title{Algebraic Monograph Transformations}

\author{
Thierry  Boy de la Tour 
}

\date{Univ. Grenoble Alpes, CNRS, Grenoble INP, LIG \\ 38000 Grenoble,
  France \\
{\small \texttt{thierry.boy-de-la-tour at imag.fr} }}

\maketitle 

\begin{abstract}
  Monographs are graph-like structures with directed edges of
  unlimited length that are freely adjacent to each other. The
  standard nodes are represented as edges of length zero. They can be
  drawn in a way consistent with standard graphs and many others, like
  E-graphs or $\infty$-graphs. The category of monographs share many
  properties with the categories of graph structures (algebras of
  monadic many-sorted signatures), except that there is no terminal
  monograph.  It is universal in the sense that its slice categories
  (or categories of typed monographs) are equivalent to the categories
  of graph structures.  Type monographs thus emerge as a natural way
  of specifying graph structures. A detailed analysis of single and
  double pushout transformations of monographs is provided, and a
  notion of attributed typed monographs generalizing typed attributed
  E-graphs is analyzed w.r.t. attribute-preserving transformations.
\end{abstract}



\section{Introduction}\label{sec-intro}

Many different notions of graphs are used in mathematics and computer
science: simple graphs, directed graphs, multigraphs, hypergraphs,
etc. One favourite notion in the context of logic and rewriting is
that also known as \emph{quivers}, i.e., structures of the form
$\tuple{N,E,s,t}$ where $N,E$ are sets and $s,t$ are functions from
$E$ (edges) to $N$ (nodes), identifying the source and target tips of
every edge (or arrow). One reason for this is that the category of
quivers is isomorphic to the category of algebras
of the many-sorted signature with two sorts
$\texttt{nodes}$ and $\texttt{edges}$ and two operator names
$\texttt{src}$ and $\texttt{tgt}$ of type
$\texttt{edges}\rightarrow \texttt{nodes}$. In conformity with this
tradition, by \emph{graph} we mean quiver throughout this paper.

In order to conveniently represent elaborate data structures it is
often necessary to enrich the structure of graphs with attributes:
nodes or edges may be labelled with elements from a fixed set, or with
values taken in some algebra, or with sets of values as in \cite{BdlTE20c},
etc. An interesting example can be found in \cite{EhrigEPT06} with the
notion of E-graphs, since the attributes are also considered as
nodes. More precisely, an E-graph is an algebra whose signature can be
represented by the following graph:

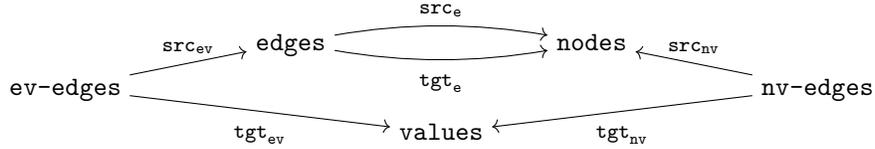
\begin{center}
  \begin{tikzpicture}[xscale=2,yscale=0.6,text height=1.5ex,text depth=.25ex]
    \node (EG) at (-1,2) {$\texttt{edges}$};
    \node (NG) at (1,2) {$\texttt{nodes}$};
    \node (EE) at (-2.5,1) {$\texttt{ev-edges}$};     
    \node (EN) at (2.5,1) {$\texttt{nv-edges}$};
    \node (NV) at (0,0) {$\texttt{values}$};
    \path[->,bend left] (EG) edge node[above, font=\footnotesize] {$\texttt{src}_{\texttt{e}}$} (NG);
    \path[->,bend right] (EG) edge node[below, font=\footnotesize] {$\texttt{tgt}_{\texttt{e}}$} (NG);
    \path[->] (EE) edge node[above, font=\footnotesize] {$\texttt{src}_{\texttt{ev}}$} (EG);
    \path[->] (EN) edge node[above, font=\footnotesize] {$\texttt{src}_{\texttt{nv}}$} (NG);
    \path[->] (EE) edge node[below, font=\footnotesize] {$\texttt{tgt}_{\texttt{ev}}$} (NV);
    \path[->] (EN) edge node[below, font=\footnotesize] {$\texttt{tgt}_{\texttt{nv}}$} (NV);
  \end{tikzpicture}
\end{center}

The names given to the sorts and operators help to understand the
structure of E-graphs: the $\texttt{edges}$ relate the
$\texttt{nodes}$ among themselves, the $\texttt{nv-edges}$ relate the
$\texttt{nodes}$ to the $\texttt{values}$, and the $\texttt{ev-edges}$
relate the $\texttt{edges}$ to the $\texttt{values}$. Hence the sort
\texttt{values} holds attributes that are also nodes. But then we see
that in E-graphs the \texttt{ev-edges} are adjacent to \texttt{edges}. This is
non standard, but we may still accept such structures as some form of
graph, if only because we understand how they can be drawn.

Hence the way of generalizing the notion of graphs seems to involve a
generalization of the signature of graphs considered as algebras. This
path has been followed by Michael Löwe in \cite{Lowe93}, where a
\emph{graph structure} is defined as a monadic many-sorted
signature. Indeed in the examples above, and in many examples provided
in \cite{Lowe93}, all operators have arity 1 and can therefore
be considered as edges from their domain to their range sort. Is this
the reason why they are called graph structures? But the example above
shows that E-graphs are very different from the graph
that represent their signature. Besides, it is not convenient that our
understanding of such structures should be based on syntax, i.e., on
the particular names given to sorts and operators in the signature.

Furthermore, it is difficult to see how the algebras of some very
simple monadic signatures can be interpreted as graphs of any
form. Take for instance the signature of graphs and reverse the target
function to $\texttt{tgt}:\texttt{nodes}\rightarrow
\texttt{edges}$. Then there is a symmetry between the sorts
\texttt{nodes} and \texttt{edges}, which means that in an algebra of
this signature nodes and edges would be objects of the same nature. Is
this still a graph? Can we draw it?  Worse still, if the two sorts are
collapsed into one, does it mean that a node/edge can be adjacent to
itself?

We may address these problems by restricting graph structures to some
class of monadic signatures whose algebras are guaranteed to behave in
an orthodox way, say by exhibiting clearly separated edges and
nodes. But this could be prone to arbitrariness, and it would still
present another drawback: that the notion of graph structure does not
easily give rise to a category.  Indeed, it is difficult to define
morphisms between algebras of different signatures, if only because
they can have any number of carrier sets.

The approach adopted here is rather to reject any \emph{structural}
distinction between nodes and edges, hence to adopt a unified view of
nodes as edges of length 0, and standard edges as edges of length 2
since they are adjacent to two nodes. This unified view logically
allows edges to be adjacent to any edges and not just to nodes, thus
generalizing the \texttt{ev-edges} of E-graphs, and even to edges that
are adjacent to themselves. Finally, there is no reason to restrict
the length of edges to 0 or 2, and we will find good reasons (in
Section~\ref{sec-graphstruct}) for allowing edges of infinite, ordinal
length. The necessary notions and notations are introduced in
Section~\ref{sec-defs}. The structure of \emph{monograph} (together
with morphisms) is defined in Section~\ref{sec-mono}, yielding a
bestiary of categories of monographs according to some of their
characteristics. The properties of these categories w.r.t. the
existence of limits and co-limits are analyzed in
Section~\ref{sec-cat}.

We then see in Section~\ref{sec-draw} how monographs can be accurately
represented by drawings, provided of course that they have finitely
many edges and that these have finite length. In particular, such
drawings correspond to the standard way of drawing a graph for those
monographs that can be identified with standard graphs, and similarly
for E-graphs.

Section~\ref{sec-graphstruct} is devoted to the comparison between
monographs and graph structures, and the corresponding algebras (that
we may call \emph{graph structured algebras}). We show a property of
universality of monographs, in the sense that all graph structured
algebras can be represented (though usually not in a canonical way) as
\emph{typed monographs}, i.e., as morphisms of monographs.

The notion of graph structure has been introduced in \cite{Lowe93} in
order to obtain categories of partial homomorphisms in which
techniques of algebraic graph rewriting could be carried out. The
correspondence with monographs established in
Section~\ref{sec-graphstruct} calls for a similar development of
partial morphisms of monographs in Section~\ref{sec-partial}. The
single and double pushout methods of rewriting monographs can then be
defined, analyzed and compared in Section~\ref{sec-dpo}.

The notion of E-graph has been introduced in \cite{EhrigEPT06} in
order to obtain well-behaved categories (w.r.t. graph rewriting) of
\emph{attributed graphs}, and hence to propose suitable
representations of real-life data structures. This is achieved by
enriching E-graphs with a data type algebra, and by identifying nodes
of sort \texttt{value} with the elements of this algebra. We pursue a
similar approach in Section~\ref{sec-atm} with the notion of
\emph{attributed typed monograph} by identifying elements of an
algebra with edges, and obtain similarly well-behaved categories. Due
to the universality of monographs we see that any $\Sig$-algebra can
be represented as an attributed typed monograph.

We conclude in Section~\ref{sec-concl}. Note that parts of
Sections~\ref{sec-cat} to \ref{sec-graphstruct} have been published in
\cite{BdlT21b}.

\section{Basic Definitions and Notations}\label{sec-defs}

\subsection{Sets}

For any sets $A$, $B$, relation $R\subseteq A\times B$ and subset
$X\subseteq A$,
let $\relim{R}{X} \defeq \setof{y\in B}{x\in X \wedge \tuple{x,y}\in
  R}$. For any $x\in A$, by abuse of notation we write $\relim{R}{x}$
for $\relim{R}{\set{x}}$. If $R$ is functional we write $R(x)$ for the
unique element of $\relim{R}{x}$, and if $S\subseteq C\times D$ is
also functional and $\relim{R}{A}\subseteq C$ let $S\circ R\defeq
\setof{\tuple{x,S(R(x))}}{x\in A}$.

A \emph{function} $f:A\rightarrow B$ is a triple $\tuple{A,R,B}$ where
$R\subseteq A\times B$ is a functional relation. We write
$\relim{f}{X}$ and $f(x)$ for $\relim{R}{X}$ and $R(x)$ respectively.
For any $Y\supseteq\relim{f}{X}$, let
$\restr{f}{X}{Y}\defeq\tuple{X,R\cap (X\times Y), Y}$ and
$\restr{f}{X}{}\defeq\restr{f}{X}{B}$. A function $g=\tuple{C,S,D}$
may be composed on the left with $f$ if $B=C$, and then
$g\circ f \defeq \tuple{A,S\circ R, D}$. If $\relim{R}{A}\subseteq C$
we may write $g\circ R$ or $S\circ f$ for $S\circ R$.

Sets and functions form the category $\Sets$ with identities
$\Id{A}\defeq\tuple{A,\setof{\tuple{x,x}}{x\in A},A}$. In $\Sets$ we
use the standard product $\tuple{A\times B, \projl, \projr}$ and
coproduct $\tuple{A+ B, \mu_1, \mu_2}$ of pairs of sets
$\tuple{A,B}$. The elements $p\in A\times B$ are pairs of elements of
$A$ and $B$, i.e., $p=\tuple{\projl(p),\projr(p)}$.  For functions
$f:C\rightarrow A$ and $g:C\rightarrow B$ we write
$\pair{f}{g}:C\rightarrow A\times B$ for the unique function such that
$\projl\circ\pair{f}{g} = f$ and $\projr\circ\pair{f}{g} = g$, i.e.,
$\pair{f}{g}(z) \defeq \tuple{f(z),g(z)}$ for all $z\in C$. The elements of
$A+B$ are pairs $\mu_1(x)\defeq\tuple{x,0}$ or $\mu_2(y)\defeq\tuple{y,1}$ for
all $x\in A$ and $y\in
B$, so that $A'\subseteq A$ and $B'\subseteq B$ entail
$A'+B'=\relim{\mu_1}{A'}\cup \relim{\mu_2}{B'}$.

An \emph{ordinal} is a set $\alpha$ such that every element of
$\alpha$ is a subset of $\alpha$, and such that the restriction of the
membership relation $\in$ to $\alpha$ is a strict well-ordering of
$\alpha$ (a total order where every non empty subset of $\alpha$ has a
minimal element). Every member of an ordinal is an ordinal, and we
write $\lambda<\alpha$ for $\lambda\in\alpha$. For any two ordinals
$\alpha$, $\beta$ we have either $\alpha<\beta$, $\alpha=\beta$ or
$\alpha>\beta$ (see e.g. \cite{Suppes72}). Every ordinal $\alpha$ has
a successor $\alpha\cup\set{\alpha}$, denoted $\alpha+1$. Natural
numbers $n$ are identified with finite ordinals, so that
$n=\set{0,1,\dotsc,n-1}$ and $\omega\defeq\set{0,1,\dotsc}$ is the
smallest infinite ordinal.

\subsection{Sequences}

For any set $E$ and ordinal $\lambda$, an \emph{$E$-sequence $s$ of
  length} $\lambda$ is an element of $E^{\lambda}$, i.e., a function
$s:\lambda \rightarrow E$. Let $\empstr$ be the only element of $E^0$
(thus leaving $E$ implicit), and for any $e\in E$ let
$\repet{e}{\lambda}$ be the only element of $\set{e}^{\lambda}$.  For
any $s\in E^{\lambda}$ and $\iota<\lambda$, the image of $\iota$ by
$s$ is written $s_{\iota}$. If $\lambda$ is finite and non zero then
$s$ can be described as $s=s_0\dotsb s_{\lambda-1}$.  For any $x\in E$
we write $\occin{x}{s}$ and say that $x$ \emph{occurs in} $s$ if there
exists $\iota<\lambda$ such that $s_{\iota}=x$.  For any ordinal
$\alpha$, let
$\seqto{E}{\alpha} \defeq \bigcup_{\lambda<\alpha}E^{\lambda}$; this
is a disjoint union. For any $s\in\seqto{E}{\alpha}$ let $\card{s}$ be
the length of $s$, i.e., the unique $\lambda<\alpha$ such that
$s\in E^{\lambda}$.

For any set $F$ and function $f:E\rightarrow F$, let
$\seqto{f}{\alpha}:\seqto{E}{\alpha}\rightarrow \seqto{F}{\alpha}$ be
the function defined by $\seqto{f}{\alpha}(s)\defeq f\circ s$ for all
$s\in\seqto{E}{\alpha}$. We have
$\seqto{\Id{E}}{\alpha} = \Id{\seqto{E}{\alpha}}$ and
$\seqto{(g\circ f)}{\alpha} = \seqto{g}{\alpha} \circ
\seqto{f}{\alpha}$ for all $g:F\rightarrow G$. Since
$s\in E^{\lambda}$ entails $f\circ s\in F^{\lambda}$, then
$\card{\seqto{f}{\alpha}(s)} = \card{s}$.

If $s$ and $s'$ are respectively $E$- and $F$-sequences of length
$\lambda$, then they are both functions with domain $\lambda$ hence
there is a function $\pair{s}{s'}$ of domain $\lambda$. Thus
$\pair{s}{s'}$ is an $(E\times F)$-sequence of length $\lambda$, and
then
$\seqto{\projl}{\alpha}(\pair{s}{s'}) = \projl\circ\pair{s}{s'} = s$
and similarly $\seqto{\projr}{\alpha}(\pair{s}{s'}) = s'$ for all
$\alpha>\lambda$.
If $f:E\rightarrow F$ and $g:E\rightarrow G$ then $\pair{f}{g}:
E\rightarrow F\times G$, hence for all $s\in\seqto{E}{\alpha}$ of
length $\lambda<\alpha$ we have
$\seqto{\pair{f}{g}}{\alpha}(s) =\pair{f}{g}\circ s = \pair{f\circ
  s}{g\circ s} = \pair{\seqto{f}{\alpha}(s)}{ \seqto{g}{\alpha}(s)}$
is an $(F\times G)$-sequence of length $\lambda$.

For $s\in\seqto{E}{\omega}$ and $(A_e)_{e\in E}$ an $E$-indexed family
of sets, let $A_s\defeq\prod_{\iota<\card{s}}A_{s_{\iota}}$. In
particular we take $A_{\empstr}\defeq 1$ as a terminal object in
$\Sets$. For $(B_e)_{e\in E}$ an $E$-indexed family of sets and
$(f_e:A_e\rightarrow B_e)_{e\in E}$ an $E$-indexed family of
functions, let
$f_s\defeq\prod_{\iota<\card{s}}f_{s_{\iota}}:A_s\rightarrow B_s$.

\subsection{Signatures and Algebras}

A \emph{signature} is a function\footnote{For the sake of simplicity we do
not allow the overloading of operator names as in
\cite{SannellaTarlecki12}. These names will turn out to be
irrelevant anyway.}
$\Sig:\Op\rightarrow \seqto{S}{\omega}$, such that
$\Sig(o)\neq\empstr$ for all $o\in\Op$. The elements of $\Op$ are
called \emph{operator names} and those of $S$ \emph{sorts}. The
\emph{arity} of an operator name $o\in\Op$ is the finite ordinal
$n\defeq \card{\Sig(o)}-1$, its \emph{range} is
$\Rng{o}\defeq \Sig(o)_{n}$ (the last element of the $S$-sequence
$\Sig(o)$) and its \emph{domain} is
$\Dom{o}\defeq \restr{\Sig(o)}{n}{}$ (the rest of the sequence). $o$
is \emph{monadic} if $n=1$. The signature $\Sig$ is \emph{finite} if
$\Op$ and $S$ are finite, it is a \emph{graph structure} if all its
operator names are monadic.
  
A \emph{$\Sig$-algebra} $\Alg$ is a pair
$\tuple{(\carrier{\Alg}{s})_{s\in S}, (\interp{o}{\Alg})_{o\in \Op}}$
where $(\carrier{\Alg}{s})_{s\in S}$ is an $S$-indexed family of sets
and
$(\interp{o}{\Alg}: \carrier{\Alg}{\Dom{o}} \rightarrow
\carrier{\Alg}{\Rng{o}})_{o\in\Op}$ is an $\Op$-indexed family of
functions.  A \emph{$\Sig$-homomorphism} $h$ from $\Alg$ to a
$\Sig$-algebra $\Blg$ is an $S$-indexed family of functions
$(h_s: \carrier{\Alg}{s}\rightarrow \carrier{\Blg}{s})_{s\in S}$ such
that
\[\interp{o}{\Blg}\circ h_{\Dom{o}} =
  h_{\Rng{o}}\circ \interp{o}{\Alg}\] for all $o\in \varOmega$. Let
$\id{\Alg}\defeq (\Id{\carrier{\Alg}{s}})_{s\in S}$ and for any
$\Sig$-homomorphism $k:\Blg\rightarrow \Clg$, the $\Sig$-homomorphism
$k\circ h:\Alg\rightarrow\Clg$ is defined by
$(k\circ h)_s\defeq k_s\circ h_s$ for all $s\in S$. Let
$\AlgCat{\Sig}$ be the category of $\Sig$-algebras and
$\Sig$-homomorphisms.

\subsection{Categories}

We assume familiarity with the notions of functors, limits, colimits
and their preservation and reflection by functors, see
\cite{Herrlich-Strecker79}.  Isomorphism between objects in a category
is denoted by $\iso$ and equivalence between categories by
$\equivCat$.

For any object $T$ of $\aCat$, the \emph{slice category}
$\sliceCat{\aCat}{T}$ has as objects the morphisms of codomain $T$ of
$\aCat$, as morphisms from object $a:A\rightarrow T$ to object
$b:B\rightarrow T$ the morphisms $f:A\rightarrow B$ of $\aCat$ such
that $b\circ f = a$, and the composition of morphisms in
$\sliceCat{\aCat}{T}$ is defined as the composition of the underlying
morphisms in $\aCat$ (see \cite{EhrigEPT06} or \cite[Definition
4.19]{Herrlich-Strecker79}).

\section{Monographs and their Morphisms}\label{sec-mono}

{\sloppy
  \begin{definition}[monographs, edges, ordinal for $A$] A set
    $A$ is a \emph{monograph} if there exists a set $E$ (whose
    elements are called \emph{edges} of $A$) and an ordinal $\alpha$
    (said to be an ordinal \emph{for} $A$) such that
    $\tuple{E,A,\seqto{E}{\alpha}}$ is a function.
  \end{definition}}

A monograph is therefore a functional relation, which means that its
set of edges is uniquely determined. On the contrary, there are always
infinitely many ordinals for a monograph. As running example we
consider the monograph $A=\set{\tuple{x,x\,y\,x},\tuple{y,y\,x\,y}}$
then its set of edges is $E=\set{x,y}$. Since $A(x)$ and $A(y)$ are
elements of $E^{3}\subseteq \seqto{E}{4}$, then
$\tuple{E,A, \seqto{E}{4}}$ is a function. Hence 4 is an ordinal for
$A$, and so are all the ordinals greater than 4.

It is easy to see that for any set of monographs there exists a
common ordinal for all its members.

\begin{definition}[length $\card{x}$, edge $x_{\iota}$, trace
  $\trace{A}$, $O$-monographs]
  For any monograph $A$ with set of edges $E$, the \emph{length} of an
  edge $x\in E$ is the length $\card{A(x)}$, also written $\card{x}$
  if there is no ambiguity.  Similarly, for any $\iota<\card{x}$ we
  may write $x_{\iota}$ for $A(x)_{\iota}$. The \emph{trace of $A$} is
  the set $\trace{A}\defeq \setof{\card{x}}{x\in E}$.  For any set $O$
  of ordinals, $A$ is an \emph{$O$-monograph} if
  $\trace{A}\subseteq O$.
\end{definition}

Since any ordinal is a set of ordinals, we see that an ordinal $\alpha$ is
for a monograph iff this is an $\alpha$-monograph. Hence all
edges of a monograph have finite length iff it is an
$\omega$-monograph. 

\begin{definition}[adjacency, nodes $\nodes{A}$, standard
  monographs]\sloppy
  For any monograph $A$ and edges $x,y$ of $A$, $x$ is \emph{adjacent
    to} $y$ if $\occin{y}{A(x)}$. A \emph{node} is an edge of length
  0, and the set of nodes of $A$ is written $\nodes{A}$. $A$ is
  \emph{standard} if $\occin{y}{A(x)}$ entails $y\in\nodes{A}$, i.e.,
  all edges are sequences of nodes.
\end{definition}

The running example $A$ has no nodes and is therefore not standard.
Since $A(x)=x\,y\,x$ then $x$ is adjacent to $y$ and to itself. Similarly,
$A(y)=y\,x\,y$ yields that $y$ is adjacent to $x$ and to itself. In this
case the adjacency relation is symmetric, but this is not generally
the case, e.g., a node is never adjacent to any edge, while edges may
be adjacent to nodes.  

\begin{definition}[morphisms of monographs]
  A \emph{morphism} $f$ from monograph $A$ to monograph $B$ with
  respective sets of edges $E$ and $F$, denoted $f:A\rightarrow B$, is
  a function $f: E\rightarrow F$ such that
  $\seqto{f}{\alpha}\circ A = B\circ f$, where $\alpha$ is
  any ordinal for $A$.
\end{definition}

Building on the running example, we consider the permutation $f=(x\ y)$
of $E$ (in cycle notation), we see that
$\seqto{f}{4}\circ A (x)= \seqto{f}{4}(x\,y\,x) = y\,x\,y = A(y) = A\circ
f(x)$ and similarly that
$\seqto{f}{4}\circ A (y)= \seqto{f}{4}(y\,x\,y) = x\,y\,x = A(x) = A\circ
f(y)$, hence $\seqto{f}{4}\circ A= A\circ f$ and $f$ is therefore a
morphism from $A$ to $A$. Since $f\circ f= \Id{E}$ is obviously the
identity morphism $\id{A}$ then $f$ is an isomorphism.

Note that the terms of the equation
$\seqto{f}{\alpha}\circ A = B\circ f$ are functional relations and not
functions. One essential feature is that this equation holds for all
ordinals $\alpha$ for $A$ iff it holds for one. Thus if we are given a
morphism then we know that the equation holds \emph{for all} big
enough $\alpha$'s, and if we want to prove that a function is a
morphism then we need only prove that \emph{there exists} a big
enough $\alpha$ such that the equation holds.

This equation is of course equivalent to
$\seqto{f}{\alpha}\circ A(x) = B\circ f(x)$ for all $x\in E$. The
terms of this last equation are $F$-sequences that should therefore
have the same length:
\[\card{x} = \card{A(x)} = \card{\seqto{f}{\alpha}\circ A(x)} =
  \card{B\circ f(x)} = \card{f(x)},\] i.e., the length of edges are
preserved by morphisms. Hence $\trace{A}\subseteq \trace{B}$, and the
equality holds if $f$ is surjective. This means that if $B$ is an
$O$-monograph then so is $A$, and that every ordinal for $B$ is an
ordinal for $A$. This also means that the images of nodes can only be
nodes:
\[\relim{\invrel{f}}{\nodes{B}}=\setof{x\in E}{\card{f(x)}=0}=
  \setof{x\in E}{\card{x}=0} = \nodes{A}.\]

The sequences $\seqto{f}{\alpha}\circ A(x)$ and $B\circ f(x)$ should
also have the same elements
\begin{align*}
 & (\seqto{f}{\alpha}\circ A(x))_{\iota} =
                                          (f\circ (A(x)))_{\iota} = f(A(x)_{\iota}) = f(x_{\iota})\\
  \text{and } & (B\circ f(x))_{\iota} = B(f(x))_{\iota} = f(x)_{\iota}
\end{align*}
 for all
$\iota<\card{x}$. Thus $f:E\rightarrow F$ is a morphism iff
\[\card{f(x)} = \card{x}\text{ and }f(x_{\iota}) = f(x)_{\iota}\text{
    for all }
x\in E\text{ and all }\iota<\card{x}.\]

Assuming that $f:A\rightarrow B$ is a morphism and that $B$ is
standard, we have $f(x_{\iota})=f(x)_{\iota}\in \nodes{B}$ thus
$x_{\iota}\in \relim{\invrel{f}}{\nodes{B}} = \nodes{A}$ for all
$x\in E$ and $\iota<\card{x}$, hence $A$ is also standard.

Given morphisms $f:A\rightarrow B$ and $g:B\rightarrow C$, we see
that $g\circ f$ is a morphism from $A$ to $C$ by letting $\alpha$ be
an ordinal for $B$, so that
\[\seqto{(g\circ f)}{\alpha} \circ A = \seqto{g}{\alpha} \circ
\seqto{f}{\alpha} \circ A = \seqto{g}{\alpha}\circ B\circ f = C\circ
g\circ f.\] 

{\sloppy \begin{definition}[categories of monographs, functor
    $\edges{}$]\label{def-catmono} Let $\MonGr$ be the category of
    monographs and their morphisms. Let $\SMonGr$ be its full
    subcategory of standard monographs. For any set $O$ of ordinals,
    let $\OMonGr{O}$ (resp. $\OSMonGr{O}$) be the full subcategory of
    $O$-monographs (resp. standard $O$-monographs).  Let $\FMonGr$ be
    the full subcategory of finite $\omega$-monographs.

    Let $\edges{}$ be the forgetful functor from $\MonGr$ to $\Sets$,
    i.e., for every monograph $A$ let $\edges{A}$ be the set of edges
    of $A$, and for every morphism $f:A\rightarrow B$ let
    $\edges{f}:\edges{A}\rightarrow\edges{B}$ be the underlying
    function, usually denoted $f$.
\end{definition}}

There is an obvious similitude between standard $\set{0,2}$-monographs
and graphs. It is actually easy to define a functor
$\MFunc{}:\Graphs\rightarrow \OSMonGr{\set{0,2}}$ by mapping any graph
$G=\tuple{N,E,s,t}$ to the monograph $\MFunc{G}$ whose set of edges is
the coproduct $N+E$, and that maps every edge $e\in E$ to the sequence
of nodes $s(e)\,t(e)$ (and of course every node $x\in N$ to
$\empstr$). Similarly graph morphisms are transformed into morphisms
of monographs through a coproduct of functions. It is easy to see that
$\MFunc{}$ is an equivalence of categories. 

It is customary in Algebraic Graph Transformation to call \emph{typed
  graphs} the objects of $\sliceCat{\Graphs}{G}$, where $G$ is a graph
called \emph{type graph}, see e.g. \cite{EhrigEPT06}.  We will extend
this terminology to monographs and refer to the objects of
$\sliceCat{\MonGr}{T}$ as the \emph{monographs typed by} $T$ and $T$
as a \emph{type monograph}.

\section{Limits and Colimits}\label{sec-cat}

The colimits of monographs follow the standard constructions of
colimits in $\Sets$ and $\Graphs$.

\begin{lemma}\label{lm-coprod}
  Every pair $\tuple{A,B}$ of monographs has a coproduct
  $\tuple{A+B,\mu_1,\mu_2}$ such that $\trace{A+B} = \trace{A}\cup
  \trace{B}$ and if $A$ and $B$ are finite (resp. standard) then so is
  $A+B$.
\end{lemma}
\begin{proof}
  Let $\alpha$ be an ordinal for $A$ and $B$, and
  $\tuple{\edges{A}+\edges{B},\mu_1,\mu_2}$ be the coproduct of
  $\tuple{\edges{A},\edges{B}}$ in $\Sets$. Since every element of
  $\edges{A}+\edges{B}$ is either a $\mu_1(x)$ or a $\mu_2(y)$ for
  some $x\in\edges{A}$, $y\in\edges{B}$, we can define a monograph $C$
  by taking $\edges{C}\defeq \edges{A}+\edges{B}$ with
  $C(\mu_1(x))\defeq \seqto{\mu_1}{\alpha} \circ A(x)$ and
  $C(\mu_2(y))\defeq \seqto{\mu_2}{\alpha} \circ B(y)$ for all
  $x\in\edges{A}$, $y\in\edges{B}$, so that $\mu_1:A\rightarrow C$ and
  $\mu_2:B\rightarrow C$ are morphisms. It is obvious that
  $\trace{C} = \trace{A}\cup \trace{B}$ and if $A$ and $B$ are finite
  (resp. standard) then so is $C$.

  \begin{center}
    \begin{tikzpicture}[scale=1]
      \node (EA) at (0,2) {$\edges{A}$}; \node (EB) at (0,0) {$\edges{B}$};
      \node (EC) at (0,1) {$\edges{A}+\edges{B}$}; \node (ED) at (2,1) {$\edges{D}$};
      \path[->] (EA) edge node[above, font=\footnotesize] {$f$} (ED);
      \path[->] (EA) edge node[left, font=\footnotesize] {$\mu_1$} (EC);
      \path[->] (EB) edge node[below, font=\footnotesize] {$g$} (ED);
      \path[->] (EB) edge node[left, font=\footnotesize] {$\mu_2$} (EC); 
      \path[->,dashed] (EC) edge node[above left, font=\footnotesize] {$h$} (ED); 
      \node (A) at (4,2) {${A}$}; \node (B) at (4,0) {${B}$};
      \node (C) at (4,1) {$C$}; \node (D) at (6,1) {${D}$};
      \path[->] (A) edge node[above, font=\footnotesize] {$f$} (D);
      \path[->] (A) edge node[left, font=\footnotesize] {$\mu_1$} (C);
      \path[->] (B) edge node[below, font=\footnotesize] {$g$} (D);
      \path[->] (B) edge node[left, font=\footnotesize] {$\mu_2$} (C); 
      \path[->,dashed] (C) edge node[above left, font=\footnotesize] {$h$} (D); 
    \end{tikzpicture}
  \end{center}

  Let $f:A\rightarrow D$ and $g:B\rightarrow D$, 
  there exists a unique function $h$ from
  $\edges{A}+\edges{B}=\edges{C}$ to $\edges{D}$ such that
  $f=h\circ\mu_1$ and $g=h\circ\mu_2$, hence \[\seqto{h}{\alpha}
    \circ C(\mu_1(x)) = \seqto{(h\circ \mu_1)}{\alpha}\circ A(x) =
    \seqto{f}{\alpha}\circ A(x) = D\circ f(x) = D\circ h(\mu_1(x))\]
  for all $x\in\edges{A}$, and similarly $\seqto{h}{\alpha}
    \circ C(\mu_2(y)) = D\circ h(\mu_2(y))$ for all $y\in\edges{B}$,
    hence $\seqto{h}{\alpha} \circ C = D\circ h$, i.e.,
    $h:C\rightarrow D$ is a morphism.
\end{proof}

\begin{lemma}\label{lm-coeq}
  Every pair of parallel morphisms $f,g:A\rightarrow B$ has a
  coequalizer $\tuple{Q,c}$ such that $\trace{Q}=\trace{B}$ and if $B$
  is finite (resp. standard) then so is $Q$.
\end{lemma}
\begin{proof}
  Let $\alpha$ be an ordinal for $B$ and $\sim$ be the smallest
  equivalence relation on $\edges{B}$ that contains
  $R=\setof{\tuple{f(x), g(x)}}{x\in \edges{A}}$ and
  $c:\edges{B}\rightarrow \edges{B}/{\sim}$ be the canonical
  surjection, so that $c\circ f = c\circ g$. We thus have for all
  $x\in\edges{A}$ that
  \[\seqto{c}{\alpha}\circ B\circ f(x) = \seqto{(c\circ
      f)}{\alpha}\circ A(x) = \seqto{(c\circ g)}{\alpha}\circ A(x) =
    \seqto{c}{\alpha}\circ B\circ g(x).\]
  For all $y,y'\in\edges{B}$ such that $c(y)=c(y')$, i.e., $y\sim y'$, there is a finite
  sequence $y_0,\ldots,y_n$ of elements of $\edges{B}$ such that
  $y_0=y$, $y_n=y'$ and $y_i\mathrel{R} y_{i+1}$ or
  $y_{i+1}\mathrel{R} y_i$ for all $0\leq i <n$, hence
  $\seqto{c}{\alpha}\circ B(y_i) = \seqto{c}{\alpha}\circ B(y_{i+1})$,
  and therefore $\seqto{c}{\alpha}\circ B(y) = \seqto{c}{\alpha}\circ
  B(y')$.

  We can now define a monograph $Q$ by taking
  $\edges{Q}=\edges{B}/{\sim}$ with
  $Q(c(y))\defeq \seqto{c}{\alpha}\circ B(y)$, so that
  $c:B\rightarrow Q$ is a morphism. Since $c$ is surjective then
  $\trace{Q} = \trace{B}$ and if $B$ is finite (resp. standard) then
  so is $Q$.

  \begin{center}
    \begin{tikzpicture}[xscale=1.8,yscale=1.3]
      \node (EA) at (0,0) {$\edges{A}$}; \node (EB) at (1,0) {$\edges{B}$};
      \node (EQ) at (2,0) {$\edges{B}/{\sim}$}; \node (ED) at (2,1) {$\edges{D}$};
      \path[->] (EA) edge [bend left] node[above, font=\footnotesize] {$f$} (EB);
      \path[->] (EA) edge [bend right] node[below, font=\footnotesize] {$g$} (EB);
      \path[->] (EB) edge node[below, font=\footnotesize] {$c$} (EQ);
      \path[->] (EB) edge node[above, font=\footnotesize] {$d$} (ED); 
      \path[->,dashed] (EQ) edge node[right, font=\footnotesize] {$h$} (ED); 
      \node (A) at (3,0) {${A}$}; \node (B) at (4,0) {${B}$};
      \node (Q) at (5,0) {$Q$}; \node (D) at (5,1) {${D}$};
      \path[->] (A) edge [bend left] node[above, font=\footnotesize] {$f$} (B);
      \path[->] (A) edge [bend right] node[below, font=\footnotesize] {$g$} (B);
      \path[->] (B) edge node[below, font=\footnotesize] {$c$} (Q);
      \path[->] (B) edge node[above, font=\footnotesize] {$d$} (D); 
      \path[->,dashed] (Q) edge node[right, font=\footnotesize] {$h$} (D); 
    \end{tikzpicture}
  \end{center}
  Let $d:B\rightarrow D$ such that $d\circ f = d\circ g$, there exists a
  unique function $h$ from $\edges{Q}$ to $\edges{D}$ such that
  $d=h\circ c$, and $h:Q\rightarrow D$ is a morphism since for all $y\in\edges{B}$ ,
  \[D\circ h(c(y)) = D\circ d(y) = \seqto{d}{\alpha}\circ B(y) =
    \seqto{h}{\alpha}\circ\seqto{c}{\alpha}\circ B(y) =
    \seqto{h}{\alpha}\circ Q(c(y)).\] 
\end{proof}

\begin{corollary}\label{cr-episurj}
  The epimorphisms in $\MonGr$ are the surjective morphisms.
\end{corollary}
\begin{proof}
  Assume $f:A\rightarrow B$ is an epimorphism. Let
  $\tuple{B+B, \mu_1,\mu_2}$ be a coproduct of $\tuple{B,B}$ and
  $\tuple{Q,c}$ be the coequalizer of
  $\mu_1\circ f,\mu_2\circ f:A \rightarrow B+B$ constructed in the
  proof of Lemma~\ref{lm-coeq}, then
  $c\circ\mu_1\circ f = c\circ\mu_2\circ f$, hence
  $c\circ\mu_1 = c\circ\mu_2$. For all $y\in\edges{B}$ we thus have
  $\mu_1(y)\sim \mu_2(y)$, and since $\mu_1(y)\neq \mu_2(y)$ then
  $\mu_1(y)$ must be related by $R$ to some element of
  $\edges{(B+B)}$, hence there is an $x\in\edges{A}$ such that
  $\mu_1(y)=\mu_1\circ f(x)$, thus $y=f(x)$ since $\mu_1$ is
  injective; this proves that $f$ is surjective. The converse is
  obvious.
\end{proof}

A well-known consequence of Lemmas~\ref{lm-coprod}, \ref{lm-coeq} and
that $\ensvide$ is the initial monograph is that all finite diagrams
have colimits. 

{\sloppy
\begin{theorem}\label{th-po} The categories of
  Definition~\ref{def-catmono} are finitely co-complete. 
\end{theorem}}

We next investigate the limits in categories of monographs. Products of
monographs are more difficult to build than products of graphs. This
is due to the fact that edges of identical length may be adjacent to
edges of different lengths.

\begin{lemma}\label{lm-prod}
  Every pair $\tuple{A,B}$ of monographs has a product
  $\tuple{A\times B,\projl',\projr'}$ such that $A\times B$ is finite
  whenever $A$ and $B$ are finite.
\end{lemma}
\begin{proof}
  Let $\alpha$ be an ordinal for $A$ and $B$, let
  $\tuple{\edges{A}\times\edges{B},\projl,\projr}$ be the product of
  $\tuple{\edges{A},\edges{B}}$ in $\Sets$, we consider the set of
  subsets $H$ of
  $\setof{\tuple{x,y}\in\edges{A}\times \edges{B}}{\card{x}=\card{y}}$
  such that $\tuple{x,y}\in H$ entails
  $\tuple{x_{\iota},y_{\iota}}\in H$ for all $\iota< \card{x}$. This
  set contains $\ensvide$ and is closed under union, hence it has a
  greatest element $\edges{P}$, and we let
  $P\tuple{x,y}\defeq\pair{A(x)}{B(y)}$ for all
  $\tuple{x,y}\in\edges{P}$; this is obviously an
  $\edges{P}$-sequence, hence $P$ is a monograph. Let
  $\projl'\defeq\restr{\projl}{\edges{P}}{}$ and
  $\projr'\defeq\restr{\projr}{\edges{P}}{}$, we
  have
  \[\seqto{\projl'}{\alpha}\circ P\tuple{x,y} = A(x) = A\circ
  \projl'\tuple{x,y}\]
  for all $\tuple{x,y}\in\edges{P}$, hence $\projl':P\rightarrow A$
  and similarly $\projr':P\rightarrow B$ are morphisms.

  \begin{center}
    \begin{tikzpicture}[xscale=1,yscale=1.2]
      \node (EA) at (2,2) {$\edges{A}$}; \node (EB) at (2,0) {$\edges{B}$};
      \node (EP) at (2,1) {$\edges{A}\times\edges{B}$}; \node (EC) at (0,1) {$\edges{C}$};
      \path[->] (EC) edge node[above, font=\footnotesize] {$f$} (EA);
      \path[->] (EP) edge node[right, font=\footnotesize] {$\projl$} (EA);
      \path[->] (EC) edge node[below, font=\footnotesize] {$g$} (EB);
      \path[->] (EP) edge node[right, font=\footnotesize] {$\projr$} (EB); 
      \path[->,dashed] (EC) edge node[above right,
      font=\footnotesize,near start] {$\pair{f}{g}$} (EP); 
      \node (A) at (6,2) {${A}$}; \node (B) at (6,0) {${B}$};
      \node (P) at (6,1) {$P$}; \node (C) at (4,1) {${C}$};
      \path[->] (C) edge node[above, font=\footnotesize] {$f$} (A);
      \path[->] (P) edge node[right, font=\footnotesize] {$\projl'$} (A);
      \path[->] (C) edge node[below, font=\footnotesize] {$g$} (B);
      \path[->] (P) edge node[right, font=\footnotesize] {$\projr'$} (B); 
      \path[->,dashed] (C) edge node[above right, font=\footnotesize] {$h$} (P); 
    \end{tikzpicture}
  \end{center}

  Let $f:C\rightarrow A$ and $g:C\rightarrow B$, then
  $\pair{f}{g}: \edges{C}\rightarrow \edges{A}\times\edges{B}$ and for
  all $z\in\edges{C}$ we have $\card{f(z)} = \card{z}=\card{g(z)}$
  hence
  $\relim{\pair{f}{g}}{\edges{C}} \subseteq
  \setof{\tuple{x,y}\in\edges{A}\times
    \edges{B}}{\card{x}=\card{y}}$. Assume that $\tuple{x,y}\in
  \relim{\pair{f}{g}}{\edges{C}}$, then there exists a $z\in\edges{C}$
  such that $x=f(z)$ and $y=g(z)$, hence $\card{x}=\card{y}$,
  $f(z_{\iota}) =   f(z)_{\iota}=x_i$ and $g(z_{\iota}) =g(z)_{\iota} =
  y_{\iota}$ for all $\iota<\card{x}$, hence
  $\tuple{x_{\iota},y_{\iota}}\in
  \relim{\pair{f}{g}}{\edges{C}}$. Thus
  $\relim{\pair{f}{g}}{\edges{C}}\subseteq \edges{P}$ and we let
  $h\defeq\restr{\pair{f}{g}}{\edges{C}}{\edges{P}}$, then $h$ is the
  unique function such that $\projl'\circ h=f$ and $\projr'\circ
  h=g$, and $h:C\rightarrow P$ is a morphism since for all $z\in \edges{C}$,
  \begin{eqnarray*}
        P\circ h(z) &=& P\tuple{f(z),g(z)}\\
                 &=& \pair{A\circ f(z)}{B\circ g(z)}\\
                 &=& \pair{\seqto{f}{\alpha}\circ
                  C(z)}{\seqto{g}{\alpha}\circ C(z)}\\
                &=& \seqto{h}{\alpha}\circ C(z). 
  \end{eqnarray*}
\end{proof}

We therefore see that $\edges{(A\times B)}$ is only a subset of
$\edges{A} \times\edges{B}$.

\begin{lemma}\label{lm-equalizer}
  Every pair of parallel morphisms $f,g:A\rightarrow B$ has an
  equalizer $\tuple{E,e}$ such that $E$ is finite whenever $A$ is finite.
\end{lemma}
\begin{proof}
  Let $\alpha$ be an ordinal for $A$,
  $\edges{E}\defeq\setof{x\in\edges{A}}{f(x)=g(x)}$,
  $e:\edges{E}\hookrightarrow \edges{A}$ be the canonical injection and
  $E(x)\defeq A(x)$ for all $x\in\edges{E}$. Since
  \[\seqto{f}{\alpha}\circ A(x) = B\circ f(x) = B\circ g(x) =
    \seqto{g}{\alpha}\circ A(x)\]
  then $E(x)$ is an $\edges{E}$-sequence, hence $E$ is a
  monograph. Besides $\seqto{e}{\alpha}\circ E(x) = A(x) = A\circ
  e(x)$, hence $e:E\rightarrow A$ is a morphism such that $f\circ e =
  g\circ e$.

  \begin{center}
    \begin{tikzpicture}[xscale=1.8,yscale=1.3]
      \node (EA) at (1,0) {$\edges{A}$}; \node (EB) at (2,0) {$\edges{B}$};
      \node (EE) at (0,0) {$\edges{E}$}; \node (ED) at (0,1) {$\edges{D}$};
      \path[->] (EA) edge [bend left] node[above, font=\footnotesize] {$f$} (EB);
      \path[->] (EA) edge [bend right] node[below, font=\footnotesize] {$g$} (EB);
      \path[{Hooks[right]}->] (EE) edge node[below, font=\footnotesize] {$e$} (EA);
      \path[->] (ED) edge node[above, font=\footnotesize] {$d$} (EA); 
      \path[->,dashed] (ED) edge node[left, font=\footnotesize] {$h$} (EE); 
      \node (A) at (4,0) {${A}$}; \node (B) at (5,0) {${B}$};
      \node (E) at (3,0) {${E}$}; \node (D) at (3,1) {${D}$};
      \path[->] (A) edge [bend left] node[above, font=\footnotesize] {$f$} (B);
      \path[->] (A) edge [bend right] node[below, font=\footnotesize] {$g$} (B);
      \path[{Hooks[right]}->] (E) edge node[below, font=\footnotesize] {$e$} (A);
      \path[->] (D) edge node[above, font=\footnotesize] {$d$} (A); 
      \path[->,dashed] (D) edge node[left, font=\footnotesize] {$h$} (E); 
    \end{tikzpicture}
  \end{center}

  For any $d:D\rightarrow A$ such that $f\circ d = g\circ d$, we have
  $d(y)\in\edges{E}$ for all $y\in\edges{D}$, hence
  $h\defeq \restr{d}{\edges{D}}{\edges{E}}$ is the unique function
  such that $d = e\circ h$. We have
  \[\seqto{e}{\alpha}\circ \seqto{h}{\alpha}\circ D =
    \seqto{d}{\alpha}\circ D = A\circ d = A\circ e\circ h =
    \seqto{e}{\alpha}\circ E\circ h\]
  and $\seqto{e}{\alpha}:\seqto{(\edges{E})}{\alpha}\hookrightarrow
  \seqto{(\edges{A})}{\alpha}$ is the canonical injection, hence
  $\seqto{h}{\alpha}\circ D = E\circ h$ and $h:D\rightarrow E$ is a
  morphism. 
\end{proof}
\begin{corollary}\label{cr-monoinj}
  The monomorphisms in $\MonGr$ are the injective morphisms.
\end{corollary}
\begin{proof}
  Assume $f:A\rightarrow B$ is a monomorphism. Let
  $\tuple{A\times A,\projl,\projr}$ be a product of $\tuple{A,A}$ and
  $\tuple{E,e}$ be the equalizer of
  $f\circ\projl,f\circ\projr:A\times A\rightarrow B$ constructed in
  the proof of Lemma~\ref{lm-equalizer}, then
  $f\circ \projl\circ e = f\circ \projr\circ e$, hence
  $\projl\circ e = \projr\circ e$. For all $x,y\in \edges{A}$, if
  $f(x)=f(y)$ then $f\circ\projl\tuple{x,y} = f\circ\projr\tuple{x,y}$
  hence $\tuple{x,y}\in\edges{E}$ and therefore
  $x= \projl\circ e\tuple{x,y} = \projr\circ e\tuple{x,y} = y$, hence
  $f$ is injective. The converse is obvious.
\end{proof}

A well-known consequence of Lemmas~\ref{lm-prod} and
\ref{lm-equalizer} is that all non-empty finite diagrams in $\MonGr$
have limits. Since a limit of $O$-monographs (resp. standard
monographs) is an $O$-monograph (resp. standard), this holds for all
categories of Definition~\ref{def-catmono}.  In particular they all
have pullbacks.

We shall now investigate the limits of the empty diagram in these
categories, i.e., their possible terminal objects.

\begin{definition}
  For any set of ordinals $O$, let \[\mathrm{T}_O = \left\{
    \begin{array}{ll}
      \setof{\tuple{\lambda, \repet{0}{\lambda}}}{\lambda\in O} & \text{if }
                                                           0\in O\\
      \ensvide & \text{otherwise.}
    \end{array}\right.\] 
\end{definition}

If $0\in O$ then $0$ is a node of $\mathrm{T}_O$ and obviously
$\edges{\mathrm{T}_O}=\trace{\mathrm{T}_O}=O$. Hence in all cases
$\mathrm{T}_O$ is a standard $O$-monograph.

\begin{lemma}\label{lm-terminalOS}
  $\mathrm{T}_O$ is terminal in $\OSMonGr{O}$.
\end{lemma}
\begin{proof}
  If $0\not\in O$ then $\ensvide=\mathrm{T}_O$ is the only standard
  $O$-monograph, hence it is terminal. Otherwise let $A$ be any
  standard $O$-monograph, $\alpha$ an ordinal for $A$ and
  $\ell:\edges{A}\rightarrow O$ be the function that maps every edge
  $x\in\edges{A}$ to its length $\card{x}$. Since $A$ is standard
  then $(\seqto{\ell}{\alpha}\circ A(x))_{\iota} = \card{A(x)_{\iota}} =
  0$ for all $\iota<\card{x}$, hence $\seqto{\ell}{\alpha}\circ A(x) =
  \repet{0}{\card{x}} = \mathrm{T}_O\circ \ell(x)$, so that $\ell:A\rightarrow
  \mathrm{T}_O$ is a morphism. Since morphisms preserve the length of
  edges and there is exactly one edge of each length in
  $\mathrm{T}_O$, then $\ell$ is unique.
\end{proof}

We now use the fact that every ordinal is a set of ordinals.

\begin{lemma}\label{lm-ordmono}
  For any monograph $T$ and morphism $f:\mathrm{T}_{\alpha}\rightarrow
  T$, any ordinal for $T$ is equal to or greater than $\alpha$.
\end{lemma}
\begin{proof}
  Let $\beta$ be an ordinal for $T$, then by the existence of $f$ we
  have $\alpha = \trace{\mathrm{T}_{\alpha}}\subseteq
  \trace{T}\subseteq \beta$, hence $\alpha\leq\beta$.
\end{proof}

\begin{lemma}\label{lm-notermo}
  $\MonGr$, $\SMonGr$ and $\FMonGr$ have no terminal object.
\end{lemma}
\begin{proof}
  Suppose that $T$ is a terminal monograph, then there is an ordinal
  $\beta$ for $T$ and there is a morphism from $\mathrm{T}_{\beta +1}$
  to $T$; by Lemma \ref{lm-ordmono} this implies that
  $\beta+1\leq\beta$, a contradiction. This still holds if $T$ is
  standard since $\mathrm{T}_{\beta +1}$ is standard. And it also
  holds if $T$ is a finite $\omega$-monograph, since then $\beta$ can
  be chosen finite, and then $\mathrm{T}_{\beta +1}$ is also a finite
  $\omega$-monograph.
\end{proof}

Since terminal objects are limits of empty diagrams obviously these
categories are not finitely complete.

\begin{theorem}
  $\OSMonGr{O}$ is finitely complete for every set of ordinals
  $O$. The categories   $\MonGr$, $\SMonGr$ and $\FMonGr$ are not
  finitely complete.
\end{theorem}
\begin{proof}
  By Lemmas~\ref{lm-prod}, \ref{lm-equalizer}, \ref{lm-terminalOS}
  and~\ref{lm-notermo}. 
\end{proof}

The category $\Graphs$ is also known to be adhesive, a property of
pushouts and pullbacks that has important consequences on algebraic
transformations (see \cite{LackS05}) and that we shall therefore investigate.

\begin{definition}[van Kampen squares, adhesive categories]\label{def-adhesive}
  A pushout square $\tuple{A,B,C,D}$ is a \emph{van Kampen square} if
  for any commutative cube
  
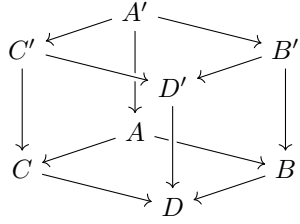
\begin{center}
    \begin{tikzpicture}[scale=0.5]
  \cubenodes{C}{D}{A}{B}{C'}{D'}{A'}{B'};
  \path[->] (LL) edge  (LF);
  \path[<-] (LF) edge (LR);
  \path[<-] (LR) edge (LB);
  \path[->] (LB) edge (LL);
  \path[<-] (LL) edge (UL); 
  \path[-] (LF) edge [draw=white, line width=3pt] (UF); 
  \path[<-] (LF) edge (UF); 
  \path[<-] (LR) edge (UR); 
  \path[<-] (LB) edge (UB); 
  \path[-] (UL) edge [draw=white, line width=3pt] (UF);
  \path[->] (UL) edge (UF);
  \path[<-] (UF) edge (UR);
  \path[<-] (UR) edge (UB);
  \path[->] (UB) edge (UL);
  \end{tikzpicture}
  \end{center}
  where the back faces $\tuple{A',A,B',B}$ and $\tuple{A',A,C',C}$ are
  pullbacks, it is the case that the top face $\tuple{A',B',C',D'}$ is
  a pushout iff the front faces $\tuple{B',B,D',D}$ and
  $\tuple{C',C,D',D}$ are both pullbacks.

  A category \emph{has pushouts along monomorphisms} if all sources
  $\tuple{A,f,g}$ have pushouts whenever $f$ or $g$ is a monomorphism.

  A category is \emph{adhesive} if it has pullbacks, pushouts along
  monomorphisms and all such pushouts are van Kampen squares.
\end{definition}

As in the proof that $\Graphs$ is adhesive, we will use the fact that
the category $\Sets$ is adhesive.

\begin{lemma}\label{lm-Erefliso}
  $\edges{}$ reflects isomorphisms.
\end{lemma}
\begin{proof}
  Let $f:A\rightarrow B$ such that ${f}$ is bijective, then it
  has an inverse $\invf{f}:\edges{B}\rightarrow\edges{A}$. For all
  $y\in\edges{B}$ and all $\iota<\card{y}$, let $x=\invf{f}(y)$, we
  have
  \[\invf{f}(y_{\iota}) = \invf{f}(f(x)_{\iota}) = \invf{f}(f(x_{\iota}))
    = x_{\iota} = \invf{f}(y)_{\iota}\] hence $\invf{f}:B\rightarrow A$ is a
  morphism, and $f$ is therefore an isomorphism.
\end{proof}

A side consequence is that $\MonGr$ is balanced, i.e., if $f$ is both
a monomorphism and an epimorphism, then by
Corollaries~\ref{cr-episurj} and \ref{cr-monoinj} ${f}$ is bijective,
hence is an isomorphism. More important is that we can use
\cite[Theorem 24.7]{Herrlich-Strecker79}, i.e., that a faithful and
isomorphism reflecting functor from a category that has some limits or
colimits and preserves them, also reflects them.

\begin{lemma}\label{lm-Epresreflcolimits}
  $\edges{}$ preserves and reflects finite colimits.
\end{lemma}
\begin{proof}
  It is easy to see from the proofs of Lemmas~\ref{lm-coprod} and
  \ref{lm-coeq} that $\edges{}$ preserves both coproducts and
  coequalizers, so that $\edges{}$ preserves all finite co-limits and
  hence also reflects them.
\end{proof}

This is particularly true for pushouts.  The situation for pullbacks
is more complicated since $\edges{}$ does not preserve products.

\begin{lemma}\label{lm-Epresreflpb}
  $\edges{}$ preserves and reflects pullbacks.
\end{lemma}
\begin{proof}
  We first prove that $\edges{}$ preserves pullbacks.
  Let $f:A\rightarrow C$, $g:B\rightarrow C$ and $\alpha$ be an ordinal
  for $A$ and $B$, we assume w.l.o.g. a canonical pullback
  $\tuple{E,h,k}$ of $\tuple{f,g,C}$, i.e., let
  $\tuple{A\times B,\projl',\projr'}$ be the product of $\tuple{A,B}$
  and $\tuple{E,e}$ be the equalizer of
  $\tuple{f\circ \projl', g\circ \projr'}$ with $h=\projl'\circ e$
  and $k=\projr'\circ e$. Let
  $\tuple{\edges{A}\times\edges{B},\projl,\projr}$ be the product of
  $\tuple{\edges{A},\edges{B}}$ in $\Sets$, we have by the proof of
  Lemma~\ref{lm-prod} that
  $\edges{(A\times B)}\subseteq\edges{A}\times \edges{B}$,
  $\projl'=\restr{\projl}{\edges{(A\times B)}}{}$ and
  $\projr'=\restr{\projr}{\edges{(A\times B)}}{}$.

  Let $H\defeq \setof{\tuple{x,y}\in\edges{A}\times \edges{B}}{f(x)=g(y)}$ and
  $j:H\hookrightarrow \edges{A}\times\edges{B}$ be the canonical
  injection. By canonical construction
  $\tuple{H, \projl\circ j,\projr\circ j}$ is a pullback of
  $\tuple{f,g,\edges{C}}$ in $\Sets$; we next prove that it is the
  image by $\edges{}$ of the pullback $\tuple{E,h,k}$ of
  $\tuple{f,g,C}$ in $\MonGr$.
  
  \begin{center}
    \begin{tikzpicture}[scale=1.2]
      \node (A) at (2,2) {$A$}; \node (B) at (0,0) {$B$};
      \node (C) at (2,0) {$C$}; \node (Ep) at (0,2) {$E$};
      \node (P) at (1,1) {$A\times B$};
      \path[->] (A) edge node[right, font=\footnotesize] {$f$} (C);
      \path[->] (B) edge node[below, font=\footnotesize] {$g$} (C); 
      \path[->] (Ep) edge node[above, font=\footnotesize] {$h$} (A); 
      \path[->] (Ep) edge node[left, font=\footnotesize] {$k$} (B);
      \path[->] (Ep) edge node[above right, font=\footnotesize] {$\!\!e$} (P);
      \path[->] (P) edge node[below, near end, font=\footnotesize] {$\projl'$} (A);
      \path[->] (P) edge node[right,font=\footnotesize] {$\projr'$} (B); 
      \node (EA) at (6,2) {$\edges{A}$}; \node (EB) at (4,0) {$\edges{B}$};
      \node (EC) at (6,0) {$\edges{C}$}; \node (E) at (4,2) {$H$};
      \node (EP) at (5,1) {$\edges{A}\times \edges{B}$};
      \path[->] (EA) edge node[right, font=\footnotesize] {$f$} (EC);
      \path[->] (EB) edge node[below, font=\footnotesize] {$g$} (EC); 
      \path[->] (E) edge node[above, font=\footnotesize] {$\projl\circ j$} (EA); 
      \path[->] (E) edge node[left, font=\footnotesize] {$\projr\circ j$} (EB);
      \path[{Hooks[right]}->] (E) edge node[above, near end,font=\footnotesize] {$j$} (EP);
      \path[->] (EP) edge node[below,near end, font=\footnotesize] {$\ \projl$} (EA);
      \path[->] (EP) edge node[right, font=\footnotesize] {$\projr$} (EB); 
    \end{tikzpicture}
  \end{center}

  By the construction of $E$ in Lemma~\ref{lm-equalizer} we have
  $\edges{E}=\setof{\tuple{x,y}\in\edges{(A\times
      B)}}{f(x)=g(y)}\subseteq H$ and
  $e:\edges{E}\hookrightarrow \edges{(A\times B)}$ is the canonical
  injection.  For all $\tuple{x,y}\in H$ we have
  $\card{x}=\card{f(x)} =\card{g(y)}= \card{y}$, and 
  for all $\iota<\card{x}$ we have
  $f(x_{\iota}) = f(x)_{\iota} =g(y)_{\iota} =g(y_{\iota})$ so that
  $\tuple{x_{\iota},y_{\iota}}\in H$ and therefore
  $H\subseteq \edges{(A\times B)}$ by the construction of $A\times B$
  in Lemma~\ref{lm-prod}. We thus have $H=\edges{E}$ hence
  $\projl\circ j = \projl'\circ e=h$ and
  $\projr\circ j = \projr'\circ e=k$, so that $\edges{}$ preserves
  pullbacks and hence as above $\edges{}$ also reflects them.
\end{proof}

{\sloppy
\begin{theorem}\label{th-adhesive}
  The categories of Definition~\ref{def-catmono} are adhesive.
\end{theorem}}
\begin{proof}
  The existence of pullbacks and pushouts is already established.  In
  any of these categories a commutative cube built on a pushout along
  a monomorphism as bottom face and with pullbacks as back faces, has
  an underlying cube in $\Sets$ that has the same properties by
  Corollary~\ref{cr-monoinj}, Lemmas~\ref{lm-Epresreflcolimits} and
  \ref{lm-Epresreflpb}. Since $\Sets$ is an adhesive category (see
  \cite{LackS05}) the underlying bottom face is a van Kampen square,
  hence such is the bottom face of the initial cube by
  Lemmas~\ref{lm-Epresreflcolimits} and \ref{lm-Epresreflpb}.
\end{proof}

\section{Drawing Monographs}\label{sec-draw}

Obviously we may endeavour to draw a monograph $A$ only if $\edges{A}$
is finite and if its edges have finite lengths, i.e., if $A$ is a
finite $\omega$-monograph. If we require that any monograph
$\MFunc{G}$ should be drawn as the graph $G$, then a node should be
represented by a bullet
\begin{tikzpicture}[point/.style={circle,inner sep=0pt,minimum size=3pt,fill=black}]
  \node (E) [point] at (0,0) {};
\end{tikzpicture}
and an edge of length 2 by an arrow
\raisebox{0.5ex}{\begin{tikzpicture}[> /.tip={Stealth[width=5pt, length=4pt]}]
  \draw[->,bend left] (0,0) to (1,0);
\end{tikzpicture}} joining its two adjacent nodes. But generally the
adjacent edges may not be nodes and there might be more than 2 of
them, hence we adopt the following convention: an edge $e$ of length
at least 2 is represented as a sequence of connected arrows with an
increasing number of tips 
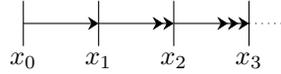
\begin{center}
  \begin{tikzpicture}[> /.tip={Stealth[width=5pt, length=4pt]}]
  \draw[->] (0,0) to (1,0);
  \draw[->>] (1,0) to (2,0);
  \draw[->>>] (2,0) to (3,0);
  \draw[dotted] (3,0) to (3.5,0);
  \draw (0,0.3) to (0,-0.3); \node at (0,-0.5){$x_0$}; 
  \draw (1,0.3) to (1,-0.3); \node at (1,-0.5){$x_1$}; 
  \draw (2,0.3) to (2,-0.3); \node at (2,-0.5){$x_2$}; 
  \draw (3,0.3) to (3,-0.3); \node at (3,-0.5){$x_3$}; 
\end{tikzpicture}
\end{center}
(where $A(e)=x_0x_1x_2x_3\dotsb$) and such that any arrow should enter
$x_i$ at the same angle as the next arrow leaves $x_i$.  For the sake
of clarity we represent symmetric adjacencies by a pair of crossings
rather than a single one, e.g., if $A(e)=xe'y$ and $A(e')=xey$, where
$x$ and $y$ are nodes, the drawing may be
\begin{center}
  \begin{tikzpicture}[point/.style={circle,inner sep=0pt,minimum
      size=3pt,fill=black},> /.tip={Stealth[width=5pt, length=4pt]}]
    \node (D) [point] at (0,0) {};
    \node (A) [point] at (3,0) {};
    \path[out=45,in=135] (D) edge (1,0);
    \path[->,out=-45,in=-135] (D) edge (1,0);
    \path[out=45,in=135] (1,0) edge (2,0); 
    \path[->, out=-45,in=-135] (1,0) edge (2,0);
    \path[->>,out=45,in=135] (2,0) edge (A); 
    \path[->>, out=-45,in=-135] (2,0) edge (A);
  \end{tikzpicture}\ \ \raisebox{1.5ex}{but not}\ \ 
  \begin{tikzpicture}[point/.style={circle,inner sep=0pt,minimum
      size=3pt,fill=black},> /.tip={Stealth[width=5pt, length=4pt]}] 
    \node (D) [point] at (0,0) {};
    \node (A) [point] at (2,0) {};
    \path[->,out=45,in=135] (D) edge (1,0);
    \path[->,out=-45,in=-135] (D) edge (1,0);
    \path[->>,out=45,in=135] (1,0) edge (A); 
    \path[->>, out=-45,in=-135] (1,0) edge (A);
  \end{tikzpicture}
\end{center}

It is sometimes necessary to name the edges in a drawing. We may then
adopt the convention sometimes used for drawing diagrams in a
category: the bullets are replaced by the names of the corresponding
nodes, and arrows are interrupted to write their name at a place free
from crossing, as in
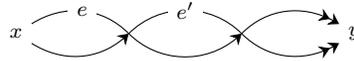
\begin{center}
  \begin{tikzpicture}[scale=1.5,> /.tip={Stealth[width=5pt, length=4pt]}]
    \node (D) at (0,0) {\footnotesize $x$};
    \node (A) at (3,0) {\footnotesize $y$};
    \path[out=32,in=135] (D) edge node[fill=white, font=\footnotesize] {$e$} (1,0);
    \path[->,out=-32,in=-135] (D) edge (1,0);
    \path[out=45,in=135] (1,0) edge node[fill=white, font=\footnotesize] {$e'$} (2,0); 
    \path[->, out=-45,in=-135] (1,0) edge (2,0);
    \path[->>,out=45,in=148] (2,0) edge (A); 
    \path[->>, out=-45,in=-148] (2,0) edge (A);
\end{tikzpicture}
\end{center}
Note that no confusion is possible between the names of nodes and
those of other edges, e.g., in
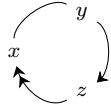
\begin{center}
  \begin{tikzpicture}[scale=0.6,> /.tip={Stealth[width=5pt, length=4pt]}]
    \node (X) at (180:1) {\footnotesize $x$};
    \node (Y) at (60:1) {\footnotesize $y$};
    \node (Z) at (-60:1) {\footnotesize $z$};
    \draw[out=90,in=150,distance=0.4cm] (X) to (Y); 
    \draw[->,out=-30,in=30,distance=0.4cm] (Y) to (Z); 
    \draw[->>,out=-150,in=-80,distance=0.5cm] (Z) to (X); 
  \end{tikzpicture}
\end{center}
it is clear that $x$ and $z$ are nodes since arrow tips point to them,
and that $y$ is the name of an edge of length 3.

As is the case of graphs, monographs may not be planar and drawing
them may require crossing edges that are not adjacent; in this case no
arrow tip is present at the crossing and no confusion is possible
with the adjacency crossings. However, it may seem preferable in such
cases to erase one arrow in the proximity of the other, as in 
\raisebox{-1ex}{\begin{tikzpicture}
    \path (-0.2,-0.2) edge (0.2,0.2);
    \path (0.2,-0.2) edge [draw=white, line width=3pt] (-0.2,0.2); 
    \path (0.2,-0.2) edge (-0.2,0.2); 
\end{tikzpicture}}.

There remains to represent the edges of length 1. Since $A(e)=x$ is
standardly written $A:e\mapsto x$, the edge $e$ will be drawn as
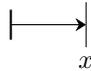
\begin{center}
  \begin{tikzpicture}[> /.tip={Stealth[width=5pt, length=4pt]}]
    \draw[thick] (0,0.2) -- (0,-0.2);
    \draw[->] (0,0) to (1,0);
    \draw (1,0.3) to (1,-0.3); \node at (1,-0.5){$x$}; 
  \end{tikzpicture}
\end{center}
In order to avoid confusion there should be
only one arrow out of the thick dash, e.g., if $A(e)=e'$ and
$A(e')=ex$ where $x$ is a node, the drawing may be
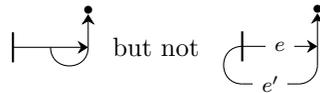
\begin{center}
  \begin{tikzpicture}[point/.style={circle,inner sep=0pt,minimum
      size=3pt,fill=black},> /.tip={Stealth[width=5pt, length=4pt]}] 
    \node (A) [point] at (1,0.5) {};
    \draw[thick] (0,0.2) -- (0,-0.2);
    \draw[->] (0,0) to (1,0);
    \path[out=-90,in=-90,distance=0.33cm] (0.5,0) edge (1,0);
    \path[->] (1,0) edge (A);
\end{tikzpicture}\ \ \raisebox{1.5ex}{but not}\ \ 
\raisebox{-2.5ex}{\begin{tikzpicture}[point/.style={circle,inner sep=0pt,minimum
    size=3pt,fill=black},> /.tip={Stealth[width=5pt, length=4pt]}] 
    \node (A) [point] at (1,0.5) {};
    \draw[thick] (0,0.2) -- (0,-0.2);
    \path[->] (0,0) edge node[fill=white, font=\footnotesize] {$e$} (1,0);
    \path[out=180,in=180,distance=0.33cm] (0,0) edge (0,-0.5);
    \path (0,-0.5) edge node[fill=white, font=\footnotesize] {$e'$} (0.75,-0.5);
    \draw[out=0,in=-90] (0.75,-0.5) edge  (1,-0.25);
    \path[->] (1,-0.25) edge (A);
\end{tikzpicture}}
\end{center}
since this last drawing may be interpreted as the monograph $A(e')=x$
and $A(e)=e'e'$, that is not isomorphic to the intended
monograph.

Other conventions may be more appropriate depending on the context or
on specific monographs. Consider for instance a monograph with one
node $x$ and two edges $\repet{x}{3}$ and $\repet{x}{4}$. The
concentration of many arrow tips on a single bullet would make things
confused unless it is sufficiently large. One possibility is to
replace the bullet by a circle and treat it as a standard edge without
tips. This monograph could then be drawn as
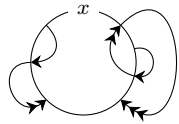
\begin{center}
  \begin{tikzpicture}[scale=0.7,> /.tip={Stealth[width=5pt, length=4pt]}]
    \draw (0,0) circle[radius=1cm];
    \node[fill=white,font=\footnotesize] at (0,1) {$x$};
    \path[->,out=-45,in=0,distance=0.5cm] (135:1) edge (180:1);
    \path[->>,out=180,in=-135,distance=0.8cm] (180:1) edge (-135:1);
    \path[->,out=15,in=-15,distance=0.5cm] (15:1) edge (-15:1);
    \path[->>,out=165,in=-135,distance=0.5cm] (-15:1) edge (45:1);
    \path[->>>,out=45,in=-45,distance=2cm] (45:1) edge (-45:1);
  \end{tikzpicture}
\end{center}

These conventions are designed so that it is only possible to read a drawing of any
finite $\omega$-monograph $A$ as the monograph $A$ itself if all edges
are named in the drawing, or as some monograph isomorphic to $A$
otherwise. This would not be possible if a monograph $A$ was a
function rather than a functional relation, since then its codomain
$\seqto{(\edges{A})}{\alpha}$ would not be pictured. It would of course be
possible to add the ordinal $\alpha$ to the drawing, but then would it
still qualify as a drawing?

Note that the drawing of a graph or of a standard
$\set{0,2}$-monograph can be read either as a graph $G$ or as a
monograph $A$, and then $\MFunc{G}\iso A$.

One particularity of monographs is that edges can be adjacent to
themselves, as in
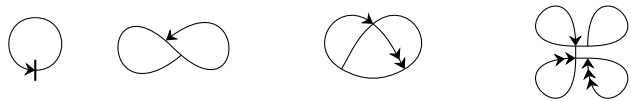
\begin{center}
  \begin{tabular}{c@{\hspace*{1em}}c@{\hspace*{2em}}cc}
\raisebox{3.7ex}{\begin{tikzpicture}[scale=0.7,> /.tip={Stealth[width=5pt, length=4pt]}]
  \draw[out=0,in=0,distance=0.67cm] (0,0) to (0,1);
  \draw[->,out=180,in=180,distance=0.67cm] (0,1) to (0,0);
  \draw[thick] (0,0.2) -- (0,-0.2);
\end{tikzpicture}} &
\begin{tikzpicture}[> /.tip={Stealth[width=5pt, length=4pt]}]
  \draw[->] (0.1,-0.1) .. controls (-1,-1) and (-1,1) .. (0,0) .. controls
  (1,-1) and (1,1) .. (-0.1,0.1);
\end{tikzpicture} &
\raisebox{3.7ex}{\begin{tikzpicture}[scale=0.6,> /.tip={Stealth[width=5pt, length=4pt]}]
  \draw[out=60,in=220] (-0.7,0) to (0,1);
  \draw[out=40,in=30,distance=1cm] (0,1) to (0.7,0);
  \draw[out=210,in=-30] (0.7,0) to (-0.7,0);
  \draw[->,out=150,in=140,distance=1cm] (-0.7,0) to (0,1);
  \draw[->>,out=-40,in=120] (0,1) to (0.7,0);
\end{tikzpicture}} &
\raisebox{-2.2ex}{
  \begin{tikzpicture}[scale=0.8,> /.tip={Stealth[width=5pt, length=4pt]}]
    \draw[out=90,in=0,distance=1.5cm] (0.1,0.1) to (0.1,0.1);
    \draw (0.1,0.1) to (-0.1,0.1);
    \draw[->,out=180,in=90,distance=1.5cm] (-0.1,0.1) to (-0.1,0.1);
    \draw (-0.1,0.1) to (-0.1,-0.1);
    \draw[->>,out=-90,in=180,distance=1.5cm] (-0.1,-0.1) to (-0.1,-0.1);
    \draw (-0.1,-0.1) to (0.1,-0.1);
    \draw[->>>,out=0,in=-90,distance=1.5cm] (0.1,-0.1) to (0.1,-0.1);
  \end{tikzpicture}}
\end{tabular}
\end{center}

We may also draw typed monographs, then
every edge $e\in\edges{A}$ has a type $a(e)$ that can be written at
the proximity of $e$. For instance, a monograph typed by
$T=\set{\tuple{u,\,v},\,\tuple{v,\,u}}$ is drawn with labels $u$ and
$v$ as in

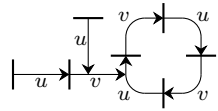
\begin{center}
\begin{tikzpicture}[type/.style={font=\footnotesize},scale=0.5,>
  /.tip={Stealth[width=5pt, length=4pt]}] 
      \draw[thick] (0.6,0) -- (1.4,0);
      \draw[thick] (-0.6,0) -- (-1.4,0);
      \draw[thick] (0,0.6) -- (0,1.4);
      \draw[thick] (0,-0.6) -- (0,-1.4);
      \draw[thick] (-2.5,-0.1) -- (-2.5,-0.9);
      \draw[thick] (-4,-0.1) -- (-4,-0.9);
      \draw[thick] (-2.4,1) -- (-1.6,1);
      \draw (1,0) to (1,-0.5);
      \draw[out=-90,in=0] (1,-0.5) to (0.5,-1);
      \draw[->] (0.5,-1) to (0,-1);
      \draw (0,-1) to (-0.5,-1); 
      \draw[out=180,in=-90] (-0.5,-1) to (-1,-0.5); 
      \draw[->] (-1,-0.5) to (-1,0); 
      \draw (-1,0) to (-1,0.5);
      \draw[out=90,in=180] (-1,0.5) to (-0.5,1);
      \draw[->] (-0.5,1) to (0,1);
      \draw (0,1) to (0.5,1); 
      \draw[out=0,in=90] (0.5,1) to (1,0.5); 
      \draw[->] (1,0.5) to (1,0); 
      \draw[->] (-2.5,-0.5) to (-1,-0.5);
      \draw[->] (-4,-0.5) to (-2.5,-0.5);
      \draw[->] (-2,1) to (-2,-0.5);
      \node[type] at (1.05,1.05) {$u$};
      \node[type] at (-1.05,-1.05) {$u$};
      \node[type] at (-1.05,1.05) {$v$};
      \node[type] at (1.05,-1.05) {$v$};
      \node[type] at (-1.8,-0.7) {$v$};
      \node[type] at (-3.25,-0.7) {$u$};
      \node[type] at (-2.2,0.4) {$u$};
\end{tikzpicture}
\end{center}

Of course, knowing that $a$ is a morphism sometimes allows to deduce
the type of an edge, possibly from the types of adjacent edges. In the
present case, indicating a single type would have been enough to
deduce all the others.

In particular applications it may be convenient to adopt completely
different ways of drawing (typed) monographs.

\begin{example}
  In \cite{Plump99} term graphs are defined from structures
  $\tuple{V,E,lab,att}$ where $V$ is a set of \emph{nodes}, $E$ a set
  of \emph{hyperedges}, $att:E\rightarrow \seqto{V}{\omega}$ defines
  the adjacencies and $lab:E\rightarrow\Op$ such that $\card{att(e)}$
  is 1 plus the arity of $lab(e)$ for all $e\in E$ (for the sake of
  simplicity, we consider only ground terms of a signature
  $\Sig:\Op\rightarrow \seqto{S}{\omega}$ such that
  $\Op\cap S=\ensvide$). The first element of the sequence $att(e)$ is
  considered as the \emph{result node} of $e$ and the others as its
  \emph{argument nodes}, so that $e$ determines \emph{paths} from its
  result node to all its argument nodes. \emph{Term graphs} are those
  structures such that paths do not cycle, every node is reachable
  from a root node and is the result node of a unique hyperedge. This
  definition is given for unsorted signatures but can easily be
  generalized, as we do now.

  We consider the type monograph $\mathrm{T}_{\Sig}$ defined by
  $\edges{\mathrm{T}_{\Sig}}\defeq S\cup\Op$, and
  \begin{align*}
    \mathrm{T}_{\Sig}(s) &\defeq\empstr \text{ for all }s\in S,\\
    \mathrm{T}_{\Sig}(o) &\defeq \Sig(o) \text{ for all }o\in \Op.
  \end{align*}
  Note that $\mathrm{T}_{\Sig}$ is a standard $\omega$-monograph, and
  indeed that any standard $\omega$-monograph has this form for a
  suitable $\Sig$.
  
  Any typed monograph $a:A\rightarrow \mathrm{T}_{\Sig}$ corresponds
  to a structure $\tuple{V,E,lab,att}$ where $V=\nodes{A}$,
  $E=\edges{A}\setminus \nodes{A}$, $lab(e)=a(e)$ and $att(e)=A(e)$
  for all $e\in E$. The only difference (due to our definition of
  signatures) is that the result node of $e$ is now the last node of
  the sequence $A(e)$.

  We now consider the signature $\Sig$ with two sorts \texttt{s},
  $\texttt{s}'$, a binary function symbol \texttt{f} with
  $\Sig(\texttt{f})= \texttt{s}'\, \texttt{s}'\, \texttt{s}$ and a
  constant symbol $\texttt{c}$ with $\Sig(\texttt{c})= \texttt{s}'$.
  We represent the term graph $\texttt{f}(\texttt{c}, \texttt{c})$,
  where the two occurrences of \texttt{c} are shared, as a typed
  monograph $a:A\rightarrow \mathrm{T}_{\Sig}$. We need two edges $e$,
  $e'$ and their result nodes $x$, $x'$, the first for $\texttt{f}$
  and the second for $\texttt{c}$. Thus $A$ is defined by
  \[\edges{A}=\set{x,\,x',\, e,\, e'},\
    A(x) = A(x') = \empstr,\ A(e) = x'\,x'\,x \text{ and } A(e') =
    x'.\] The typing morphism $a:A\rightarrow \mathrm{T}_{\Sig}$ is
  given by
  \[a(x)=\texttt{s},\ a(x')=\texttt{s}',\ a(e)=\texttt{f} \text{ and }
    a(e')=\texttt{c}.\] We give below the standard drawing of the
  monograph $A$ typed by $a$ and the (clearly preferable) standard
  depiction of the corresponding term graph.

  
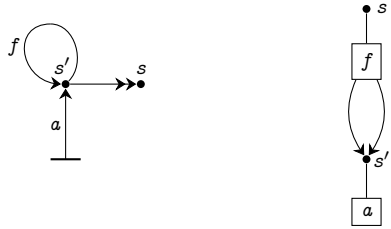
\begin{center}
    \begin{tikzpicture}[type/.style={font=\footnotesize},
      point/.style={circle,inner sep=0pt,minimum size=3pt,fill=black},>
      /.tip={Stealth[width=5pt, length=4pt]}]
      \node[point] (A) at (0,1) {}; \node[point] (B) at (1,1) {};
      \draw[thick] (0.2,0) -- (-0.2,0);
      \draw[->] (0,0) to (A);
      \draw[-, out=45,in=0,distance=0.33cm] (A) to (-0.2,1.8);
      \draw[->, out=180,in=180,distance=0.55cm] (-0.2,1.8) to (A);
      \path[->>] (A) edge (B);
      \node[type] at (-0.05,1.25) {$\texttt{s}'$};
      \node[type] at (1,1.2) {$\texttt{s}$};
      \node[type] at (-0.15,0.45) {$\texttt{a}$};
      \node[type] at (-0.7,1.5) {$\texttt{f}$};
      \node[point] (C) at (4,2){}; \node[point] (D) at (4,0){};
      \node[draw](E) at (4,1.3){\footnotesize $\texttt{f}$};
      \node[draw](F) at (4,-0.7){\footnotesize $\texttt{a}$};
      \path[-] (C) edge (E); \path[-] (D) edge (F);
      \path[->, bend right] (E) edge (D);
      \path[->, bend left] (E) edge (D);
      \node[type] at (4.2,2) {$\texttt{s}$};
      \node[type] at (4.2,0) {$\texttt{s}'$};
      
    \end{tikzpicture}
  \end{center}
\end{example}

\section{Graph Structures and Typed Monographs}\label{sec-graphstruct}

The procedure of reading the drawing of a graph as a
$\GraphSig$-algebra $\Glg$, where $\GraphSig$ is the signature of
graphs given in Section~\ref{sec-intro}, is rather simple: every
bullet is interpreted as an element of $\Glg_{\texttt{nodes}}$, every
arrow as an element of $\Glg_{\texttt{edges}}$ and the images of this
element by the functions $\texttt{src}^\Glg$ and $\texttt{tgt}^\Glg$
are defined according to geometric proximity in the drawing. A
procedure for reading E-graphs would be similar, except that bullets
may be interpreted either as $\texttt{nodes}$ or
$\texttt{values}$, and this typing information should
therefore be indicated in the drawing.

Since the drawing of a graph is nothing else than the drawing of a
standard $\set{0,2}$-monograph, we may skip the drawing step and
directly transform a standard $\set{0,2}$-monograph $A$ as a
$\GraphSig$-algebra $\Glg$. Then
\[\Glg_{\texttt{nodes}}=\nodes{A},\ 
\Glg_{\texttt{edges}}=\setof{x\in\edges{A}}{\card{x}=2},\ 
\texttt{src}^\Glg(x)=x_0\text{ and }
\texttt{tgt}^\Glg(x)=x_1\] for all $x\in
\Glg_{\texttt{edges}}$. Thus every node of
$A$ is typed by \texttt{nodes} and all other edges are typed by
\texttt{edges}. This typing is obviously a morphism from
$A$ to the monograph $\set{\tuple{\texttt{nodes},\, \empstr},\,
  \tuple{\texttt{edges} ,\, \texttt{nodes}\,
    \texttt{nodes}}}$ that is isomorphic to the terminal object of
$\OSMonGr{\set{0,2}}$ (see Lemma~\ref{lm-terminalOS}).

More generally, for any given graph structure $\Msig$ we may ask which
monographs, equipped with a suitable morphism to a type monograph $T$,
can be interpreted in this way as $\Msig$-algebras. As above, the
edges of $T$ should be the sorts of $\Msig$. But this is not
sufficient since there is no canonical way of linking adjacencies in
$T$ (such as $\texttt{edges}_0=\texttt{nodes}$ and
$\texttt{edges}_1=\texttt{nodes}$) with the operator names of $\Msig$
(such as \texttt{src} and \texttt{tgt}). We will therefore use a
notion of morphism between signatures in order to rename operators,
and we also rename sorts in order to account for functoriality in $T$.

\begin{definition}[categories $\SigCat$, $\MonSig$, $\SigCats$]
  A \emph{morphism} $r$ from $\Sig:\Op\rightarrow \seqto{S}{\omega}$
  to $\Sig':\Op'\rightarrow \seqto{S'}{\omega}$ is a pair
  $\tuple{\opn{r},\sorts{r}}$ of functions
  $\opn{r}:\Op\rightarrow \Op'$ and $\sorts{r}:S\rightarrow S'$ such
  that \[\seqto{\sorts{r}}{\omega}\circ \Sig = \Sig'\circ \opn{r}.\]
  For any morphism $r':\Sig'\rightarrow \Sig''$ let
  $r'\circ r \defeq \tuple{\opn{r'}\circ
    \opn{r},\sorts{r'}\circ\sorts{r}}:\Sig \rightarrow \Sig''$,
  $\id{\Sig}\defeq \tuple{\Id{\Op},\Id{S}}$, and $\SigCat$ be the
  category of signatures and their morphisms. Let $\MonSig$ be the
  full subcategory of graph structures.

  Let $\SigCats$ be the subcategory of $\SigCat$ restricted to
  morphisms of the form $\tuple{\opn{r},j}$ where $j$ is a canonical
  injection. We write ${\dotiso}$ for the isomorphism relation between
  objects in $\SigCats$.
\end{definition}

The question is therefore to elucidate the link between $T$ and
$\Msig$. As explained above, the edges of $T$ correspond to the sorts
of $\Msig$. We also see that every adjacency in $T$ corresponds to an
operator name in $\Msig$, e.g., an edge $e$ of length 2 adjacent to
$e_0$ and $e_1$ (i.e. such that $T(e)=e_0\,e_1$) corresponds to two
operator names, say $\texttt{src}_e$ and $\texttt{tgt}_e$, of domain
sort $e$ and range sort $e_0$ and $e_1$ respectively. Since edges may
have length greater than 2, we create canonical operator names of the
form $\opname{e}{\iota}$ for the $\iota^{\mathrm{th}}$ adjacency of
the edge $e$ for every $\iota<\card{e}$ (hence we favor
$\opname{e}{0}$ and $\opname{e}{1}$ over $\texttt{src}_e$ and
$\texttt{tgt}_e$).

\begin{definition}[functor $\SigFunc{}: \MonGr\rightarrow \MonSig$]\label{def-Sfunc}
  To every monograph $T$ we associate the set of operator names
  $\Funcs{T}\defeq \setof{\opname{e}{\iota}}{e\in\edges{T} \text{ and
    } \iota<\card{e}}$ and the graph structure
  $\SigFunc{T}:\Funcs{T}\rightarrow \seqto{(\edges{T})}{\omega}$
  defined by $\SigFunc{T}(\opname{e}{\iota}) \defeq e\, e_{\iota}$ for
  all $\opname{e}{\iota}\in \Funcs{T}$, i.e., we let
  $\Dom{\opname{e}{\iota}}\defeq e$ and
    $\Rng{\opname{e}{\iota}}\defeq e_{\iota}$.

  To every morphism $f:T\rightarrow T'$ in $\MonGr$ we associate the
  morphism $\SigFunc{f} :\SigFunc{T}\rightarrow \SigFunc{T'}$ defined
  by: $\opn{(\SigFunc{f})}$ is the function that maps every operator
  name $\opname{e}{\iota}\in\Funcs{T}$ to the operator name
  $\opname{f(e)}{\iota}\in\Funcs{T'}$, and $\sorts{(\SigFunc{f})}$ is
  the function ${f}:\edges{T}\rightarrow \edges{T'}$.
\end{definition}

We see that $\SigFunc{f}$ is indeed a morphism of graph structures:
\[ \seqto{\sorts{(\SigFunc{f})}}{\omega}\circ \SigFunc{T} (\opname{e}{\iota})
  = f(e)\, f(e_{\iota})= f(e)\, f(e)_{\iota} = \SigFunc{T'}(\opname{f(e)}{\iota})
  = \SigFunc{T'}\circ \opn{(\SigFunc{f})} (\opname{e}{\iota})\]
for all $\opname{e}{\iota}\in \Funcs{T}$, and it is obvious that
$\SigFunc{}$ is a faithful functor.

The next lemma is central as it shows that no graph structure is
omitted by the functor $\SigFunc{}$ if we allow sort-preserving
isomorphisms of graph structures.  We assume the Axiom of Choice
through its equivalent formulation known as the Numeration Theorem
\cite{Suppes72}.

\begin{lemma}\label{lm-Sig2Mono}
  For every graph structure $\Msig$ there exists a monograph $T$
  such that $\SigFunc{T}\dotiso\Msig$.
\end{lemma}
\begin{proof}
  Let $\Msig:\Op\rightarrow \seqto{S}{\omega}$ and for every
  sort $s\in S$ let $\Op_s$ be the set of operator names $o\in\Op$ whose
  domain sort is $s$, i.e., $\Op_s\defeq \relim{\invf{\mathrm{Dom}}}{s}$.  By
  the Numeration Theorem there exists an ordinal $\lambda_s$
  equipollent to $\Op_s$, i.e., such that there exists a bijection
  $\nu_s:\lambda_s\rightarrow \Op_s$. Let $T$ be the monograph such
  that $\edges{T}\defeq S$ and
  $T(s)_{\iota}\defeq \Rng{\nu_s(\iota)}$ for all
  $\iota<\lambda_s$, so that $T(s)$ is an $S$-sequence of length
  $\lambda_s$.
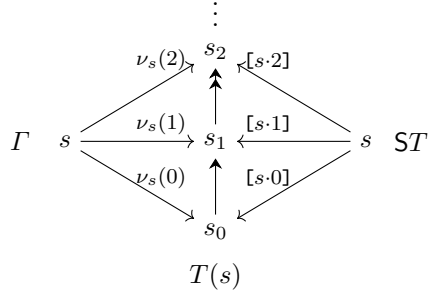
  \begin{center}
    \begin{tikzpicture}[xscale=2,yscale=1.2]
      \node (S1) at (0,0) {$s_1$}; \node (S0) at (0,-1) {$s_0$}; \node (S2) at (0,1) {$s_2$};
      \node (SG) at (-1,0) {$s$}; \node (ST) at (1,0) {$s$};
      \node at (-1.3,0) {$\Msig$}; \node at (1.3,0) {$\SigFunc{T}$};
      \node at (0, 1.5){$\vdots$}; \node at (0,-1.5) {$T(s)$};
      \node at (0.35,-0.4) {\footnotesize $\opname{s}{0}$};
      \node at (0.35,0.2) {\footnotesize $\opname{s}{1}$};
      \node at (0.35,0.9) {\footnotesize $\opname{s}{2}$};
      \node at (-0.36,-0.4) {\footnotesize $\nu_s(0)$};
      \node at (-0.36,0.2) {\footnotesize $\nu_s(1)$};
      \node at (-0.36,0.9) {\footnotesize $\nu_s(2)$};
      \path[->] (SG) edge (S0);
      \path[->] (ST) edge (S0);
      \path[->] (SG) edge (S1);
      \path[->] (ST) edge (S1);
      \path[->] (SG) edge (S2);
      \path[->] (ST) edge (S2);
      \draw[-{Stealth[width=5pt,length=4pt]}] (S0) to (S1);
      \draw[-{Stealth[width=5pt,length=4pt]Stealth[width=5pt,length=4pt]}] (S1) to (S2);
    \end{tikzpicture}
  \end{center}

  We now consider the function $\opn{r}:\Funcs{T}\rightarrow \Op$
  defined by $\opn{r}(\opname{s}{\iota})\defeq \nu_s(\iota)$. This
  function is surjective since for all $o\in \Op$, by taking
  $s=\Dom{o}$ and $\iota=\invf{\nu_s}(o)$ we get
  $\iota<\lambda_s=\card{s}$ hence $\opname{s}{\iota}\in\Funcs{T}$ and
  obviously $\opn{r}(\opname{s}{\iota}) = o$. It is also injective
  since $\opn{r}(\opname{s}{\iota}) = \opn{r}(\opname{s'}{\iota'})$
  entails $s=\Dom{\nu_s(\iota)} = \Dom{\nu_{s'}(\iota')} = s'$
  hence $\iota=\iota'$ and therefore
  $\opname{s}{\iota} = \opname{s'}{\iota'}$. Finally, we see that
  \[\seqto{\Id{S}}{\omega}\circ \SigFunc{T}(\opname{s}{\iota}) =
    s\,s_{\iota} = \Dom{\nu_s(\iota)}\,\Rng{\nu_s(\iota)} =
    \Msig(\nu_s(\iota)) = \Msig\circ\opn{r}(\opname{s}{\iota})\]
  for all $\opname{s}{\iota}\in\Funcs{T}$, hence
  $\tuple{\opn{r},\Id{S}}:\SigFunc{T}\rightarrow \Msig$ is an
  isomorphism, so that $\SigFunc{T}\dotiso\Msig$.
\end{proof}

The reason why monographs require edges of ordinal length now becomes
apparent: the length of an edge $s$ is the cardinality of $\Op_s$, i.e.,
the number of operator names whose domain sort is $s$, and no
restriction on this cardinality is ascribed to graph structures. The
bijections $\nu_s$ provide linear orderings of the sets
$\Op_s$. Since $T(s)$ depends on $\nu_s$ the monograph $T$ such that
$\SigFunc{T}\dotiso\Msig$ may not be unique, even though $\SigFunc{}$
is injective on objects, as we now show.

\begin{theorem}
  $\SigFunc{}$ is an isomorphism-dense embedding of $\MonGr$ into $\MonSig$.
\end{theorem}
\begin{proof}
  It is trivial by Lemma~\ref{lm-Sig2Mono} that $\SigFunc{}$ is
  isomorphism-dense since $\SigFunc{T}\dotiso\Msig$ entails
  $\SigFunc{T}\iso\Msig$.  Assume that $\SigFunc{T}=\SigFunc{T'}$ then
  $\edges{T}=\edges{T'}$ and $\Funcs{T}=\Funcs{T'}$, hence
  $\card{T(e)}=\card{T'(e)}$ for all $e\in\edges{T}$, and
  $T(e)_{\iota} = (\SigFunc{T}(\opname{e}{\iota}))_1 =
  (\SigFunc{T'}(\opname{e}{\iota}))_1 = T'(e)_{\iota}$ for all
  $\iota<\card{e}$, thus $T=T'$.
\end{proof}

It is therefore clear that if $\SigFunc{}$ were full it would be an
equivalence of categories, but this is not the case as we now
illustrate on graphs.

\begin{example}\label{ex-graphs}
  We consider the graphs structure $\GraphSig$. We have
  $\Op_{\texttt{nodes}}=\ensvide$ and
  $\Op_{\texttt{edges}} = \set{\texttt{src}, \texttt{tgt}}$, hence
  $\lambda_{\texttt{edges}} =2$. Let
  $\nu_{\texttt{edges}}:2\rightarrow \Op_{\texttt{edges}}$ be the bijection defined by
  $\nu_{\texttt{edges}}: 0\mapsto \texttt{src}, 1\mapsto \texttt{tgt}$, the
  corresponding monograph is
  $\mathrm{T_g}\defeq\set{\tuple{\texttt{nodes},\empstr},\tuple{\texttt{edges},
      \texttt{nodes}\,\texttt{nodes}}}$, and we easily check that
  $\SigFunc{\mathrm{T_g}}\dotiso \GraphSig$. However, the only
  automorphism of $\mathrm{T_g}$ is $\id{\mathrm{T_g}}$, while
  $\GraphSig$ has a non trivial automorphism
  $m=\tuple{(\texttt{src}\ \texttt{tgt}), \Id{\set{\texttt{nodes},
        \texttt{edges}}}}$ (in cycle notation), hence $\SigFunc{}$ is
  not surjective on morphisms.
\end{example}

This automorphism reflects the fact that a graph structure does not
define an order between its operator names. Directing edges as arrows
from $\texttt{src}$ to $\texttt{tgt}$ or the other way round is a
matter of convention that is reflected in the choice of $\nu_{\texttt{edges}}$ in
Example~\ref{ex-graphs}. This contrasts with monographs where edges
are inherently directed by ordinals, and also with the structure of
graphs where the source function comes first. In the translation from $\MonGr$
to $\MonSig$ the direction of edges are necessarily lost, hence these
categories are not equivalent.

\begin{example}\label{ex-Egraphs}
  The signature $\EGraphSig$ of E-graphs from \cite{EhrigEPT06} has
  five sorts $\texttt{edges}$, $\texttt{nv-edges}$,
  $\texttt{ev-edges}$, $\texttt{nodes}$, $\texttt{values}$ and six
  operator names $\texttt{src}_{\texttt{e}}$,
  $\texttt{tgt}_{\texttt{e}}$, $\texttt{src}_{\texttt{nv}}$,
  $\texttt{tgt}_{\texttt{nv}}$, $\texttt{src}_{\texttt{ev}}$,
  $\texttt{tgt}_{\texttt{ev}}$ whose domain and range sorts are
  defined as in Section~\ref{sec-intro}. We have
  $\Op_{\texttt{nodes}} = \Op_{\texttt{values}} =\ensvide$,
  $\Op_{\texttt{edges}}=\set{\texttt{src}_{\texttt{e}},
    \texttt{tgt}_{\texttt{e}}}$,
  $\Op_{\texttt{nv-edges}}=\set{\texttt{src}_{\texttt{nv}},
    \texttt{tgt}_{\texttt{nv}}}$ and
  $\Op_{\texttt{ev-edges}}=\set{\texttt{src}_{\texttt{ev}},
    \texttt{tgt}_{\texttt{ev}}}$. There are four possible monographs
  $T$ such that $\SigFunc{T}\dotiso \EGraphSig$ given by
  \[\begin{array}{ll}
      T(\texttt{nodes}) = T(\texttt{values}) =
      \empstr, &
                T(\texttt{nv-edges}) =
                \texttt{nodes}\, \texttt{values}\text{ or }
                \texttt{values}\, \texttt{nodes},\\
      T(\texttt{edges}) =
      \texttt{nodes}\, \texttt{nodes}, &
                                       T(\texttt{ev-edges}) =
                                       \texttt{edges}\, \texttt{values}\text{ or }
                                       \texttt{values}\, \texttt{edges}.
    \end{array}\]
  These four monographs are depicted below.
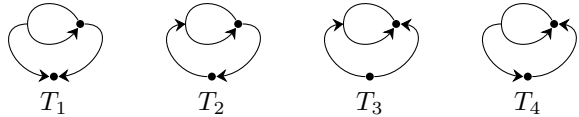
\[\begin{tikzpicture}[point/.style={circle,inner sep=0pt,minimum
      size=3pt,fill=black},scale=0.7,> /.tip={Stealth[width=5pt,
      length=4pt]}] 
    \node (N1) [point] at (0.5,1) {};
    \node (V1) [point] at (0,0) {};
    \path[-,out=120,in=90,distance=0.5cm] (N1) edge (-0.5,1);
    \path[->,out=-90,in=-120,distance=0.5cm] (-0.5,1) edge (N1);
    \path[->,out=180,in=180,distance=0.7cm] (-0.5,1) edge (V1); 
    \path[->, out=0,in=0,distance=0.7cm] (N1) edge (V1);
    \node (N2) [point] at (3.5,1) {};
    \node (V2) [point] at (3,0) {};
    \path[-,out=120,in=90,distance=0.5cm] (N2) edge (2.5,1);
    \path[->,out=-90,in=-120,distance=0.5cm] (2.5,1) edge (N2);
    \path[<-,out=180,in=180,distance=0.7cm] (2.5,1) edge (V2); 
    \path[->, out=0,in=0,distance=0.7cm] (N2) edge (V2);
    \node (N3) [point] at (6.5,1) {};
    \node (V3) [point] at (6,0) {};
    \path[-,out=120,in=90,distance=0.5cm] (N3) edge (5.5,1);
    \path[->,out=-90,in=-120,distance=0.5cm] (5.5,1) edge (N3);
    \path[<-,out=180,in=180,distance=0.7cm] (5.5,1) edge (V3); 
    \path[<-, out=0,in=0,distance=0.7cm] (N3) edge (V3);
    \node (N4) [point] at (9.5,1) {};
    \node (V4) [point] at (9,0) {};
    \path[-,out=120,in=90,distance=0.5cm] (N4) edge (8.5,1);
    \path[->,out=-90,in=-120,distance=0.5cm] (8.5,1) edge (N4);
    \path[->,out=180,in=180,distance=0.7cm] (8.5,1) edge (V4); 
    \path[<-, out=0,in=0,distance=0.7cm] (N4) edge (V4);
    \node at (0,-0.5) {$T_1$}; 
    \node at (3,-0.5) {$T_2$}; 
    \node at (6,-0.5) {$T_3$}; 
    \node at (9,-0.5) {$T_4$}; 
  \end{tikzpicture}\]
The type indicated by the syntax (and consistent with the drawings of
E-graphs in \cite{EhrigEPT06}) is of course $T_1$.
\end{example}

The restrictions of $\SigFunc{}$ to the categories of
Definition~\ref{def-catmono} are isomorphism-dense embeddings into
full subcategories of $\MonSig$ that are easy to define. The
$O$-monographs correspond to graph structures $\Msig:\Op\rightarrow
\seqto{S}{\omega}$ such that $\card{\Op_s}\in O$ for all $s\in S$, and
the standard monographs to $\Op_{\Rng{o}}=\ensvide$ for all
$o\in\Op$. The finite monographs correspond to finite $S$, hence
$\FMonGr$ corresponds to finite signatures.

We can now describe precisely how a monograph $A$ typed by $T$ through
$a:A\rightarrow T$ can be read as an $\SigFunc{T}$-algebra $\Alg$. As
mentioned above, every edge $x$ of $A$ is typed by $a(x)\in\edges{T}$
and should therefore be interpreted as an element of $\Alg_{a(x)}$,
hence $\Alg_{a(x)}$ is the set of all edges $x\in\edges{A}$ that are
typed by $a(x)$. Then, for every $\iota<\card{x}=\card{a(x)}$, the
$\iota^{\mathrm{th}}$ adjacent edge $x_{\iota}$ of $x$ is the image of
$x$ by the $\iota^{\mathrm{th}}$ operator name for this type of edge,
that is $\opname{a(x)}{\iota}$. Note that the sort of this image is
$a(x_{\iota}) = a(x)_{\iota}$ that is precisely the range sort of the
operator name $\opname{a(x)}{\iota}$ in $\SigFunc{T}$ (see
Definition~\ref{def-Sfunc}), so that $\Alg$ is indeed an
$\SigFunc{T}$-algebra. This leads to the following definition.

\begin{definition}[functor $\AlgFunc{T}{}:\sliceM{T}\rightarrow
  \AlgCat{\SigFunc{T}}$]\label{def-algfunc}
  Given a monograph $T$, we define the function $\AlgFunc{T}{}$ that
  maps every object $a:A\rightarrow T$ of $\sliceM{T}$ to the
  $\SigFunc{T}$-algebra $\Alg = \AlgFunc{T}{a}$ defined by
  \begin{itemize}
  \item $\carrier{\Alg}{e}\defeq \relim{\invrel{a}}{e}$ for all
    $e\in \edges{T}$, and
  \item $\interp{\opname{e}{\iota}}{\Alg}(x) \defeq x_{\iota}$
    for all $\opname{e}{\iota}\in \Funcs{T}$ and $x\in \carrier{\Alg}{e}$.
  \end{itemize}
  Besides, $\AlgFunc{T}{}$ also maps every morphism
  $f:a\rightarrow b$, where $b:B\rightarrow T$, to the
  $\SigFunc{T}$-homomorphism $\AlgFunc{T}{f}$ from $\Alg$ to
  $\Blg=\AlgFunc{T}{b}$ defined by
  \[(\AlgFunc{T}{f})_e \defeq\restr{f}{\Alg_e}{\Blg_e} \text{ for all }
  e\in \edges{T}.\]
\end{definition}

\begin{figure}[t]
  \centering
  \begin{tikzpicture}[xscale=3,yscale=2]
    \node at (0.5,1){$e\ e_{\iota}$};
    \node at (0,1) {$\opname{e}{\iota}\,\mapsto$};
    \node at (-0.8,1) {$\SigFunc{T}$};
    \node (X) at (0,0) {$x$}; \node (A) at (1,0) {$x_{\iota}$};
    \node at (-0.8,0) {$\Alg$};
    \node at (1.6,0){$\edges{A}$};
    \node at (1.6,1){$\edges{T}$};
    \path[{Bar}->] (X) edge (A);
    \node at (0.5,0.1) {${\scriptstyle \interp{\opname{e}{\iota}}{\Alg}}$};
    \draw (0,0) circle (0.3cm and 0.2cm);
    \draw (1,0) circle (0.3cm and 0.2cm);
    \draw (0.5,0) circle (1cm and 0.4cm);
    \path[->] (0,0.2) edge node [left, font=\footnotesize]{$a$} (0.4,0.9);
    \path[->] (1,0.2) edge node [right, font=\footnotesize]{$a$} (0.6,0.9);
  \end{tikzpicture}
  \caption{The $\SigFunc{T}$-algebra $\Alg=\AlgFunc{T}{a}$ where
    $a:A\rightarrow T$}
  \label{fig-AlgFunc}
\end{figure}

The $\SigFunc{T}$-algebra $\Alg$ can be pictured as in
Figure~\ref{fig-AlgFunc}. The carrier sets $\Alg_e$ form a partition
of $\edges{A}$. Since $f:a\rightarrow b$ (not pictured) is a function
$f:\edges{A}\rightarrow\edges{B}$ such that $b\circ f = a$, then
$\relim{b\circ f}{\Alg_e} = \relim{a}{\relim{\invf{a}}{e}} \subseteq
\set{e}$ hence
$\relim{f}{\Alg_e}\subseteq \relim{\invf{b}}{e} = \Blg_e$, so that
$\restr{f}{\Alg_e}{\Blg_e}$ 
is well-defined. We also see that
$h=\AlgFunc{T}{f}$ is an $\SigFunc{T}$-homomorphism from
$\Alg$ to
$\Blg$ since for every operator name
$\opname{e}{\iota}\in\Funcs{T}$ we have
$\Dom{\opname{e}{\iota}}=e$, $\Rng{\opname{e}{\iota}}=e_{\iota}$ and
\[\interp{\opname{e}{\iota}}{\Blg}\circ h_e(x) =
  \interp{\opname{e}{\iota}}{\Blg}(f(x)) = f(x)_{\iota}= f(x_{\iota}) =
  f(\interp{\opname{e}{\iota}}{\Alg}(x)) = h_{e_{\iota}}\circ
  \interp{\opname{e}{\iota}}{\Alg}(x)\]
for all $x\in\Alg_e$. It is obvious from
Definition~\ref{def-algfunc} that $\AlgFunc{T}$ preserves identities
and composition of morphisms, hence that it is indeed a functor.

\begin{theorem}\label{thm-sliceiso}
  For every monograph $T$, $\AlgFunc{T}{}$ is an equivalence.
\end{theorem}
\begin{proof}
  Let $a:A\rightarrow T$ and $b:B\rightarrow T$ be objects of
  $\sliceM{T}$ and $\Alg\defeq\AlgFunc{T}a$,
  $\Blg\defeq\AlgFunc{T}b$. It is trivial that $\AlgFunc{T}$ is faithful.
  
  \emph{$\AlgFunc{T}$ is full.} For any
  $\SigFunc{T}$-homomorphism  $h:\Alg\rightarrow \Blg$, let
  $f:\edges{A}\rightarrow \edges{B}$ be the function defined by
  $f(x)\defeq h_{a(x)}(x)$ for all $x\in\edges{A}$. Let $e=a(x)$ so
  that $x\in\Alg_e$, since $h_e(x)\in\Blg_e = \relim{\invrel{b}}{e}$
  then $b\circ f(x)=b(h_e(x)) = e$, hence $b\circ f=a$ and
  $\card{f(x)} = \card{b(f(x))}=\card{a(x)} = \card{x}$. For all
  $\iota<\card{x}$ we have $a(x_{\iota})=a(x)_{\iota} = e_{\iota}$ and
  since $h$ is an $\SigFunc{T}$-homomorphism then
  \[f(x_{\iota}) = h_{e_{\iota}}(\interp{\opname{e}{\iota}}{\Alg}(x))
    = \interp{\opname{e}{\iota}}{\Blg}(h_e(x)) = f(x)_{\iota}\] hence
  $f:a\rightarrow b$ is a morphism. Since
  $(\AlgFunc{T}f)_e(x) = \restr{f}{\Alg_e}{\Blg_e}(x) = h_e(x)$ for
  all $e\in\edges{T}$ and all $x\in\Alg_e$, then $\AlgFunc{T}f=h$.

  \emph{$\AlgFunc{T}$ is isomorphism-dense.} For any
  $\SigFunc{T}$-algebra $\Clg$, let
  \[\edges{C}\defeq \bigcup_{e\in\edges{T}}\Clg_e\times\set{e}\ \text{
      and }\ (C\tuple{x,e})_{\iota}\defeq
    \tuple{\interp{\opname{e}{\iota}}{\Clg}(x),e_{\iota}}\] for all
  $\tuple{x,e}\in\edges{C}$ and $\iota<\card{e}$.  Since
  $\Rng{\opname{e}{\iota}}=e_{\iota}$ then
  $\interp{\opname{e}{\iota}}{\Clg}(x) \in\Clg_{e_{\iota}}$ hence
  $(C\tuple{x,e})_{\iota}\in \edges{C}$, so that $C$ is a monograph
  such that $\card{\tuple{x,e}}=\card{e}$. Let
  $c:\edges{C}\rightarrow\edges{T}$ be defined by
  $c\tuple{x,e}\defeq e$, we have
  \[c(\tuple{x,e}_{\iota}) = 
    c \tuple{\interp{\opname{e}{\iota}}{\Clg}(x),e_{\iota}} =
    e_{\iota}=(c\tuple{x,e})_{\iota},\] hence $c:C\rightarrow T$ is a
  morphism.   For all $e\in\edges{T}$ we have
  $(\AlgFunc{T}c)_e = \relim{\invf{c}}{e} = \Clg_e\times\set{e}$, and
  we let $h_e:\Clg_e\rightarrow (\AlgFunc{T}c)_e$ be defined by
  $h_e(x)\defeq \tuple{x,e}$ for all $x\in\Clg_e$. The functions $h_e$
  are bijective and $h\defeq (h_e)_{e\in\edges{T}}$ is an
  $\SigFunc{T}$-homomorphism since
  \[\interp{\opname{e}{\iota}}{\AlgFunc{T}c}\circ h_e(x)
    =\interp{\opname{e}{\iota}}{\AlgFunc{T}c}\tuple{x,e}
    =\tuple{x,e}_{\iota} =
    \tuple{\interp{\opname{e}{\iota}}{\Clg}(x),e_{\iota}} =
    h_{e_{\iota}}\circ \interp{\opname{e}{\iota}}{\Clg}(x),\]
  for all $\opname{e}{\iota}\in\Funcs{T}$ and $x\in\Clg_{e}$,
  hence $\Clg\iso\AlgFunc{T}c$.
\end{proof}

It is easy to see that for any two signatures $\Sig$ and $\Sig'$, if
$\Sig\iso\Sig'$ then $\AlgCat{\Sig}\iso \AlgCat{\Sig'}$. We conclude
that all graph structured algebras can be represented as typed
monographs.

\begin{corollary}\label{cr-gs2slicem}
  For every graph structure $\Msig$ there exists a monograph $T$
  such that $\AlgCat{\Msig} \equivCat \sliceM{T}$.
\end{corollary}
\begin{proof}
  By Lemma~\ref{lm-Sig2Mono} there exists $T$ such that
  $\Msig\iso\SigFunc{T}$, hence
  $\sliceM{T}\equivCat \AlgCat{\SigFunc{T}} \iso\AlgCat{\Msig}$.
\end{proof}

\newcommand{\dfdiag}[2]{\,\raisebox{-3ex}{\begin{tikzpicture}[xscale=0.8,yscale=0.2]
      \path[<-] (0,1) edge node[above,font=\footnotesize]{$#1$} (1,1);
      \path[<-] (0,0) edge node[below,font=\footnotesize]{$#2$} (1,0);
    \end{tikzpicture}}\ }
\newcommand{\tginf}{\mathrm{T_{\infty}}}

\begin{example}\label{ex-poly}
  Following \cite{Burroni1993}, an
  \emph{$\infty$-graph} $\Glg$ is given by a diagram of sets
  \[\Glg_0\dfdiag{s_0}{t_0}\Glg_1\dfdiag{s_1}{t_1} \cdots
    \dfdiag{s_{n-1}}{t_{n-1}} \Glg_n \dfdiag{s_n}{t_n}\Glg_{n+1}
    \dfdiag{s_{n+1}}{t_{n+1}} \cdots\]
  such that, for every $n\in\omega$, the following equations hold:
  \[s_n\circ s_{n+1}= s_n\circ t_{n+1},\ \ \ \ \  t_n\circ s_{n+1}= t_n\circ t_{n+1}.\]
  This means that every element $x$ of $\Glg_{n+2}$ is an edge whose
  source $x_0$ and target $x_1$ are edges of $\Glg_{n}$ that are
  parallel, i.e., that have same source $(x_0)_0=(x_1)_0$ and same
  target $(x_0)_1=(x_1)_1$. Graphically:
  
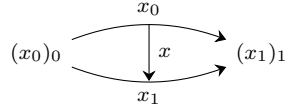
\begin{center}
    \begin{tikzpicture}[xscale=1.5,> /.tip={Stealth[width=5pt, length=4pt]}]
      \node (S) at (-1,0) {\footnotesize $(x_0)_0$}; \node (T) at (1,0) {\footnotesize $(x_1)_1$};
      \path[->,bend left] (S) edge node[above,font=\footnotesize]{$x_0$} (T);
      \path[->,bend right] (S) edge node[below,font=\footnotesize]{$x_1$} (T);
      \path[->] (0,0.38) edge node[right,font=\footnotesize]{$x$} (0,-0.37);
    \end{tikzpicture}
  \end{center}
  This is known as the \emph{globular condition}.  We consider the
  type monograph $\tginf$ defined by
  $\edges{\tginf}=\omega$,$\tginf(0) = \empstr$ and
  $\tginf(n+1) = n\,n$ for all $n\in\omega$.
  This is an infinite non-standard $\set{0,2}$-monograph that can be
  pictured as
  
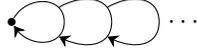
\[\begin{tikzpicture}[point/.style={circle,inner sep=0pt,minimum
        size=3pt,fill=black},scale=0.7,> /.tip={Stealth[width=5pt,
        length=4pt]}] 
      \node (N1) [point] at (0,0) {};
      \path[-,out=60,in=90,distance=0.5cm] (N1) edge (1,0);
      \path[->,out=-90,in=-60,distance=0.5cm] (1,0) edge (N1);
      \path[-,out=30,in=90,distance=0.5cm] (0.9,0.28) edge (1.9,0);
      \path[->,out=-90,in=-30,distance=0.5cm] (1.9,0) edge (0.9,-0.28);
      \path[-,out=30,in=90,distance=0.5cm] (1.82,0.28) edge (2.8,0);
      \path[->,out=-90,in=-30,distance=0.5cm] (2.8,0) edge (1.82,-0.28);
      \node at (3.3,0){$\cdots$};
    \end{tikzpicture}\]
  We express the globular condition on typed monographs
  $g:G\rightarrow \tginf$ as:
  \[ \text{for all } x\in \edges{G},\text{ if } g(x)\geq 2\text{ then
    } G(x_0) = G(x_1).\]

  We rapidly check that this is equivalent to the globular condition on
  the $\SigFunc{\tginf}$-algebra $\Glg = \AlgFunc{\tginf}g$. The set
  of sorts of $\SigFunc{\tginf}$ is $\omega$ and its operator names
  are $\opname{n+1}{0}$ and $\opname{n+1}{1}$ with domain sort $n+1$ and range
  sort $n$, for all $n\in\omega$. We let $s_n\defeq
  \interp{\opname{n+1}{0}}{\Glg}$ and $t_n\defeq
  \interp{\opname{n+1}{1}}{\Glg}$, that are functions from
  $\Glg_{n+1}$ to $\Glg_n$ as in the diagram of $\infty$-graphs.

  By Definition~\ref{def-algfunc} we have for all
  $x\in\Glg_{n+2}=\relim{\invf{g}}{n+2}$ and all $i,j\in 2$ that 
  \[\interp{\opname{n+1}{j}}{\Glg} \circ
    \interp{\opname{n+2}{i}}{\Glg}(x) =
    \interp{\opname{n+1}{j}}{\Glg}(x_i) = (x_i)_j\]
  hence
  \begin{align*}
    G(x_0) = G(x_1) &\text{ iff } (x_0)_0 = (x_1)_0 \text{ and } (x_0)_1=(x_1)_1\\
                    &\text{ iff } \interp{\opname{n+1}{0}}{\Glg} \circ
    \interp{\opname{n+2}{0}}{\Glg}(x) = \interp{\opname{n+1}{0}}{\Glg} \circ
    \interp{\opname{n+2}{1}}{\Glg}(x)\\&\ \ \text{ and } \interp{\opname{n+1}{1}}{\Glg} \circ
    \interp{\opname{n+2}{0}}{\Glg}(x) = \interp{\opname{n+1}{1}}{\Glg} \circ
    \interp{\opname{n+2}{1}}{\Glg}(x)\\
    &\text{ iff } s_n\circ s_{n+1}(x)= s_n\circ t_{n+1}(x)\text{ and } t_n\circ s_{n+1}(x)= t_n\circ t_{n+1}(x).
    \end{align*} 
\end{example}

\begin{example}\label{ex-hypergraphs}
  The signature $\mathrm{\Gamma_h}$ of \emph{hypergraphs} (see
  \cite[Example 3.4]{Lowe93}) is defined by the set of sorts
  $\mathrm{S_h}\defeq\set{\texttt{V}}\cup \setof{\texttt{H}_{n,m}}{n,m
    \in\omega}$ and for all $n,m\in\omega$ by $n$ operator names
  $\texttt{src}^{n,m}_i$ and $m$ operator names $\texttt{tgt}^{n,m}_j$
  with domain sort $\texttt{H}_{n,m}$ and range sort \texttt{V} for
  all $1\leq i\leq n$ and $1\leq j\leq m$. Hence there are $n+m$
  operator names of domain $\texttt{H}_{n,m}$, and $(n+m)!$ bijections
  from the ordinal $n+m$ to this set of operator names. But since they
  all have the same range sort \texttt{V}, the type monograph
  $\mathrm{T_h}$ does not depend on these bijections (one for every
  pair $\tuple{n,m}$). It is defined by
  $\edges{\mathrm{T_h}}\defeq \mathrm{S_h}$ and
  \begin{align*}
    \mathrm{T_h}(\texttt{V}) &= \empstr \\
    \mathrm{T_h}(\texttt{H}_{n,m}) &= \repet{\texttt{V}}{(n+m)} \text{ for all }n,m\in\omega.
  \end{align*}
  This is a standard $\omega$-monograph. It is easy to
  see that any standard $\omega$-monograph can by typed by
  $\mathrm{T_h}$, though not in a unique way. Every edge of length
  $l>0$ can be typed by any sort $\texttt{H}_{n,m}$ such that $n+m=l$,
  and every node can be typed by $\texttt{V}$ (or by $\texttt{H}_{0,0}$
  if it is not adjacent to any edge). To any such typing corresponds an
  $\SigFunc{\mathrm{T_h}}$-algebra by the equivalence
  $\AlgFunc{\mathrm{T_h}}$, and then to a hypergraph (a
  $\mathrm{\Gamma_h}$-algebra) since
  $\mathrm{\Gamma_h}\iso \SigFunc{\mathrm{T_h}}$.

  But to know which hypergraph $\mathcal{H}$ corresponds exactly to a
  typed monograph we need to be more specific, since there are
  infinitely many isomorphisms between $\mathrm{\Gamma_h}$ and
  $\SigFunc{\mathrm{T_h}}$. The natural isomorphism stems from the
  obvious orderings
  $\texttt{src}^{n,m}_1<\cdots <\texttt{src}^{n,m}_n <
  \texttt{tgt}^{n,m}_1 <\cdots <\texttt{tgt}^{n,m}_m$ for all
  $n,m\in\omega$. In this isomorphism the canonical operator name
  $\opname{\texttt{H}_{n,m}}{i}$ for all $i<n+m$ corresponds to
  $\texttt{src}^{n,m}_{i+1}$ if $i<n$, and to
  $\texttt{tgt}^{n,m}_{i+1-n}$ if $i\geq n$. Thus an edge $x$, say of
  length 3 typed by $\texttt{H}_{2,1}$, must be interpreted as an
  hyperedge $x\in\mathcal{H}_{\texttt{H}_{2,1}}$ with
  $\interp{(\texttt{src}^{2,1}_1)}{\mathcal{H}}(x)=x_0$,
  $\interp{(\texttt{src}^{2,1}_2)}{\mathcal{H}}(x)=x_1$,
  $\interp{(\texttt{tgt}^{2,1}_1)}{\mathcal{H}}(x)=x_2$ and $x_0,
  x_1, x_2\in\mathcal{H}_{\texttt{V}}$.
\end{example}

The results of this section apply in particular to typed graphs. It is
easy to see that $\SigFunc{}\circ\MFunc{}$ is an isomorphism-dense
embedding of $\Graphs$ into the full subcategory of graph structures
$\Msig: \Op\rightarrow \seqto{S}{\omega}$ such that for every operator
name $o\in \Op$ we have $\card{\Op_{\Dom{o}}}=2$ and
$\Op_{\Rng{o}}=\ensvide$. Hence for every such $\Msig$ there exists a
graph $G$ such that
$\sliceCat{\Graphs}{G} \equivCat \sliceM{\MFunc{G}} \equivCat
\AlgCat{\Msig}$. The type graph $G$ is determined only up to the
orientation of its edges.

\section{Submonographs and Partial Morphisms}\label{sec-partial}

Graph structures have been characterized in \cite{Lowe93} as the
signatures that allow the transformation of the corresponding algebras
by the single pushout method. This method is based on the construction
of pushouts in categories of partial homomorphisms, defined as
standard homomorphisms from subalgebras of their domain algebra, just
as partial functions are standard functions from subsets of their
domain (in the categorical theoretic sense of the word
\emph{domain}). The results of Section~\ref{sec-graphstruct} suggest
that a similar approach can be followed with monographs. We first need
a notion of submonograph, their (inverse) image by morphisms and
restrictions of morphisms to submonographs.

\begin{definition}[submonographs and their images, restricted morphisms]\label{def-submono}
  A monograph $A$ is a \emph{submonograph} of a monograph $M$ if
  $A\subseteq M$.  For any monograph $N$ and morphism
  $f:M\rightarrow N$, let
  $f(A)\defeq\setof{\tuple{f(x), N\circ f(x)}}{x\in \edges{A}}$. For
  any submonograph $C\subseteq N$, let
  $\invf{f}(C)\defeq \setof{\tuple{x,M(x)}}{x\in
    \relim{\invf{f}}{\edges{C}}}$.  If $f(A)\subseteq C$, let
  $\restr{f}{A}{C}:A\rightarrow C$ be the morphism whose underlying
  function is $\restr{f}{\edges{A}}{\edges{C}}$.
\end{definition}

In the sequel we will use the following obvious facts without explicit
reference. $f(A)$ and $\invf{f}(C)$ are submonographs of $N$ and
$M$ respectively.  If $A$ and $B$ are submonographs of
$M$ then so are $A\cup B$ and $A\cap B$. We have $f(A\cup B) =
f(A)\cup f(B)$ thus $A\subseteq B$ entails $f(A)\subseteq
f(B)$. If $C$ and $D$ are submonographs of
$N$ we have similarly $\invf{f}(C\cup D) = \invf{f}(C) \cup
\invf{f}(D)$ and $C\subseteq D$ entails $\invf{f}(C)\subseteq
\invf{f}(D)$. We also have $A\subseteq
\invf{f}(f(A))$ and $f(\invf{f}(C))= C\cap
f(M)$. For any $g:N\rightarrow P$ and submonograph $E$ of
$P$, $\invf{(g\circ f)}(E) =
\invf{f}(\invf{g}(E))$.  If
$\tuple{A+B,\mu_1,\mu_2}$ is the coproduct of $\tuple{A,B}$ and
$C$ is a submonograph of
$A+B$ then $C=\invf{\mu_1}(C)+\invf{\mu_2}(C)$.

We may now define the notion of partial morphisms of monographs, with
a special notation in order to distinguish them from standard
morphisms, and their composition.

\begin{definition}[categories of partial morphisms of
  monographs]\label{def-catp}
  \sloppy
  A \emph{partial morphism} $\partm{f}:A\rightarrow B$ is a morphism
  $f:A'\rightarrow B$ where $A'$ is a submonograph of $A$. $f$ is
  called the \emph{underlying morphism} of $\partm{f}$. If the domain
  of $f$ is not otherwise specified, we write
  $\partm{f}:A\hookleftarrow A'\rightarrow B$. If the domain $A'$ of $f$ is
  specified but not the domain of $\partm{f}$ then they are assumed to
  be identical, i.e., $\partm{f}:A'\hookleftarrow A'\rightarrow B$.  For
  any $\partm{g}:B\hookleftarrow B'\rightarrow C$ we define the composition
  of partial morphisms as
  \[\partm{g}\circ \partm{f}\ \defeq\ \partm{g\circ
      \restr{f}{\invf{f}(B')}{B'}} : A\hookleftarrow 
    \invf{f}(B') \rightarrow C.\]

  Let $\MonGrP$ be the category of monographs and partial
  morphisms. Let $\SMonGrP$ be its full subcategory of standard
  monographs. For any set $O$ of ordinals, let $\OMonGrP{O}$
  (resp. $\OSMonGrP{O}$) be its full subcategory of $O$-monographs
  (resp. standard $O$-monographs).  Let $\FMonGrP$ be its full
  subcategory of finite $\omega$-monographs.
\end{definition}

Note that
$\tuple{\invf{f}(B'),\, \restr{f}{\invf{f}(B')}{B'}:
  \invf{f}(B')\rightarrow B',\,j': \invf{f}(B')\hookrightarrow A'}$ is a
pullback of $\tuple{j:B'\hookrightarrow B,\,f:A'\rightarrow B,\, B}$ and
is therefore an inverse image (i.e., a pullback along a monomorphism,
see \cite{Herrlich-Strecker79}), and it is therefore easy to see that
composition of partial morphisms is associative, see
\cite{RobinsonR1988}. (Note however that $\MonGrP$ is not a category of
partial maps in the sense of \cite{RobinsonR1988}, since partial maps
are defined modulo isomorphic variations of $A'$.)

We now see how these inverse images allow to formulate a sufficient
condition ensuring that restrictions of coequalizers are again
coequalizers.

\begin{lemma}[coequalizer restriction]\label{lm-coeq-restr}
  Let $A'$ and $B'$ be submonographs of $A$ and $B$ respectively and
  $f,g:A\rightarrow B$ be parallel morphisms such that
  \[\invf{f}(B') = A' = \invf{g}(B'),\] if $\tuple{Q,c}$ is a
  coequalizer of $\tuple{f,g}$ then $\tuple{Q',c'}$ is a coequalizer
  of $\tuple{\restr{f}{A'}{B'}, \restr{g}{A'}{B'}}$, where
  $Q' =c(B')$, $c'=\restr{c}{B'}{Q'}$ and $\invf{c}(Q') = B'$.
\end{lemma}

\begin{center}
  \begin{tikzpicture}[xscale=2,yscale=2]
    \node (A) at (0,1) {${A}$}; \node (B) at (1,1) {${B}$}; \node (Q) at (2,1) {$Q$}; 
    \path[->] (A) edge [bend left] node[above, font=\footnotesize] {$f$} (B);
    \path[->] (A) edge [bend right] node[below, font=\footnotesize] {$g$} (B);
    \path[->] (B) edge node[above, font=\footnotesize] {$c$} (Q);
    \node (A') at (0,0) {${A'}$}; \node (B') at (1,0) {${B'}$}; \node (Q') at (2,0) {$Q'$}; 
    \path[->] (A') edge [bend left] node[above, font=\footnotesize] {$\restr{f}{A'}{B'}$} (B');
    \path[->] (A') edge [bend right] node[below, font=\footnotesize] {$\restr{g}{A'}{B'}$} (B');
    \path[->] (B') edge node[above, font=\footnotesize] {$c'$} (Q');
    \path [arrows={Hooks[right]->}] (A') edge (A); 
    \path [arrows={Hooks[right]->}] (B') edge (B); 
    \path [arrows={Hooks[right]->}] (Q') edge (Q); 
  \end{tikzpicture}
\end{center}

\begin{proof}
  We assume w.l.o.g. that $(Q,c)$ is the coequalizer of $\tuple{f,g}$
  constructed in Lemma~\ref{lm-coeq} with ${\sim}$ being the
  equivalence relation generated by
  $R=\setof{\tuple{f(x), g(x)}}{x\in \edges{A}}$, and we let $(Q',c')$
  be the coequalizer of $\tuple{\restr{f}{A'}{B'}, \restr{g}{A'}{B'}}$
  constructed similarly with the equivalence relation ${\approx}$
  generated by
  $R'=\setof{\tuple{\restr{f}{A'}{B'} (x), \restr{g}{A'}{B'}
      (x)}}{x\in \edges{A'}}$. By the properties of $f$ and $g$ we
  have that
  \[f(x)\in\edges{B'}\text{ iff } x\in \relim{\invf{f}}{\edges{B'}}
    \text{ iff } x\in\edges{A'}\text{ iff } x\in
    \relim{\invf{g}}{\edges{B'}} \text{ iff } g(x)\in\edges{B'}\] for
  all $x\in \edges{A}$, hence for all $y,y'\in\edges{B}$ we have that
  $y\mathrel{R'}y'$ iff $y\mathrel{R}y'$ and at least one of $y,y'$ is
  in $\edges{B'}$. By an easy induction we see that $y\approx y'$ iff
  $y\sim y'$ and $y'\in\edges{B'}$, hence the $\approx$-classes are
  the $\sim$-classes of the elements of $\edges{B'}$, i.e.,
  $\edges{Q'}=\relim{c}{\edges{B'}}$. It follows trivially that
  $Q' =c(B')$, $c'=\restr{c}{B'}{Q'}$ and $\invf{c}(Q') = B'$.
\end{proof}

It is then easy to obtain a similar result on pushouts.

\begin{lemma}[pushout restriction]\label{lm-pushout-restr}
  Let $A'$, $B'$, $C'$ be submonographs of $A$, $B$, $C$ respectively
  and $f:A\rightarrow B$, $g:A\rightarrow C$ be morphisms such that
  \[\invf{f}(B') = A' = \invf{g}(C'),\] if $\tuple{h,k,Q}$ is a pushout
  of $\tuple{A,f,g}$, let $Q'=h(B')\cup k(C')$, $\invf{h}(Q') = B'$ and $\invf{k}(Q') =
  C'$, then
  $\tuple{\restr{h}{B'}{Q'},\restr{k}{C'}{Q'}, Q'}$ is a pushout of
  $\tuple{A',\restr{f}{A'}{B'}, \restr{g}{A'}{C'}}$.
\end{lemma}

  \begin{center}
    \begin{tikzpicture}[xscale=0.7,yscale=0.7]
      \cubenodes{A'}{C'}{B'}{B'+C'}{A}{C}{B}{B+C}
      \node (Q) at (11,3.2) {$Q$}; \node (Q') at (11,0) {$Q'$}; 
      \path [arrows={Hooks[right]->}] (LL) edge (UL); 
      \path [arrows={Hooks[right]->}] (Q') edge (Q); 
      \path [arrows={Hooks[right]->}] (LB) edge (UB); 
      \path[->] (LL) edge node[above, font=\footnotesize] {$\restr{f}{A'}{B'}$} (LB);
      \path[->] (LR) edge node[fill=white, font=\footnotesize,near start] {$c'$} (Q'); 
      \path[->] (LB) edge node[above right, font=\footnotesize] {$\restr{h}{B'}{Q'}$} (Q');
      \path[->] (LF) edge node[below right, font=\footnotesize] {$\restr{k}{C'}{Q'}$} (Q');
      \path[->] (LL) edge node[below, font=\footnotesize] {$\restr{g}{A'}{C'}$} (LF);
      \path[->] (UL) edge node[above, font=\footnotesize] {$f$} (UB);
      \path[-] (UL) edge[draw=white, line width=3pt]  (UF);
      \path[->] (UL) edge node[below, font=\footnotesize] {$g$} (UF);
      \path[->] (UB) edge node[below, font=\footnotesize] {$\mu_1$} (UR);
      \path[->] (UF) edge node[above, font=\footnotesize] {$\mu_2$} (UR);
      \path[->] (LB) edge node[below, font=\footnotesize] {$\mu'_1$} (LR);
      \path[->] (LF) edge node[above, font=\footnotesize] {$\mu'_2$} (LR);
      \path[->] (UB) edge node[above right, font=\footnotesize] {$h$} (Q);
      \path[-] (LR) edge[draw=white, line width=3pt]  (UR);
      \path [arrows={Hooks[right]->}] (LR) edge (UR);
      \path[-] (UF) edge[draw=white, line width=3pt]  (Q);
      \path[->] (UF) edge node[below right, font=\footnotesize] {$k$} (Q);
      \path[->] (UR) edge node[fill=white, font=\footnotesize,near start] {$c$} (Q); 
      \path[-] (LF) edge[draw=white, line width=3pt]  (UF);
      \path [arrows={Hooks[right]->}] (LF) edge (UF);
      \end{tikzpicture}
  \end{center}

\begin{proof}
  We assume w.l.o.g. that $\tuple{h,k,Q}$ is obtained by the canonical
  construction of pushouts, i.e., that $h=c\circ\mu_1$ and
  $k=c\circ\mu_2$ where $\tuple{Q,c}$ is a coequalizer of
  $\tuple{\mu_1\circ f,\mu_2\circ g}$ and $\tuple{B+C,\mu_1,\mu_2}$ is
  the coproduct of $\tuple{B,C}$. Let $\tuple{B'+C',\mu'_1,\mu'_2}$ be
  the coproduct of $\tuple{B',C'}$, then obviously
  $B'+C'\subseteq B+C$, $\mu'_1 = \restr{\mu_1}{B'}{B'+C'}$ and
  $\mu'_2 = \restr{\mu_2}{C'}{B'+C'}$. Since
  \[\invf{(\mu_1\circ f)}(B'+C') = \invf{f}(B') = A' = \invf{g}(C') =
    \invf{(\mu_2\circ g)}(B'+C')\] then by Lemma~\ref{lm-coeq-restr}
  $\tuple{Q',c'}$ is a coequalizer of
  \[\tuple{\restr{(\mu_1\circ f)}{A'}{B'+C'}, \restr{(\mu_2\circ
        g)}{A'}{B'+C'}} = \tuple{\mu'_1\circ\restr{f}{A'}{B'},
      \mu'_2\circ\restr{g}{A'}{C'}}\] where $Q'=c(B'+C')$,
  $c'=\restr{c}{B'+C'}{Q'}$ and $\invf{c}(Q')=B'+C'$. We thus have
  $\invf{h}(Q') = \invf{(c\circ\mu_1)}(Q') = \invf{\mu_1}(B'+C') = B'$
  and similarly $\invf{k}(Q') = C'$. We also have
  $\restr{h}{B'}{Q'} = \restr{(c\circ\mu_1)}{B'}{Q'} = c'\circ \mu'_1$
  and
  $\restr{k}{C'}{Q'} = \restr{(c\circ\mu_2)}{C'}{Q'} = c'\circ
  \mu'_2$, hence $\tuple{\restr{h}{B'}{Q'},\restr{k}{C'}{Q'}, Q'}$ is
  the canonical pushout of
  $\tuple{A',\restr{f}{A'}{B'}, \restr{g}{A'}{B'}}$, and therefore
  $Q' = \restr{h}{B'}{Q'}(B')\cup \restr{k}{C'}{Q'}(C')= h(B')\cup
  k(C')$.
\end{proof}

We can now show that categories of partial morphisms of monographs
have pushouts. The following construction is inspired by
\cite[Construction 2.6, Theorem 2.7]{Lowe93} though the proof uses
pushout restriction.

\begin{theorem}\label{th-ppo}
  The categories of Definition~\ref{def-catp} have pushouts.
\end{theorem}
\begin{proof}
  Let $\partm{f}:A\hookleftarrow A_1\rightarrow B$ and
  $\partm{g}:A\hookleftarrow A_2\rightarrow C$. The set of submonographs
  $J\subseteq A_1\cap A_2$ such that $\invf{f}(f(J))=J$ and
  $\invf{g}(g(J))=J$ contains $\ensvide$ and is closed under
  union, hence has a greatest element denoted $I$. There is also
  a greatest submonograph $X\subseteq B$ such that
  $\invf{f}(X)\subseteq I$, that must therefore be greater than
  $f(I)$, i.e., we have $f(I)\subseteq X$ hence
  $\invf{f}(f(I))\subseteq \invf{f}(X)$ and this yields
  $\invf{f}(X)=I$. Similarly, there is a greatest submonograph
  $Y\subseteq C$ such that $\invf{g}(Y)\subseteq I$, so that $g(I)\subseteq Y$
  and $\invf{g}(Y) = I$.

  Let $f'=\restr{f}{I}{X}$, $g'=\restr{g}{I}{Y}$ and $\tuple{h,k,Q}$
  be a pushout of $\tuple{I,f',g'}$ in $\MonGr$, we claim that
  $\tuple{\partm{h},\partm{k},Q}$ is a pushout of
  $\tuple{A,\partm{f},\partm{g}}$ in $\MonGrP$, where obviously
  $\partm{h}:B\hookleftarrow X \rightarrow Q$ and
  $\partm{k}:C\hookleftarrow Y \rightarrow Q$. We first see that
  \[\partm{h}\circ\partm{f}
    = \partm{h\circ\restr{f}{\invf{f}(X)}{X}}= \partm{h\circ f'} 
    = \partm{k\circ g'} = \partm{k\circ\restr{g}{\invf{g}(Y)}{Y}}
    = \partm{k}\circ\partm{g}.\]

  We now consider any pair of partial morphisms $\partm{v}:B\hookleftarrow
  B'\rightarrow U$ and $\partm{w}:C\hookleftarrow C'\rightarrow U$ such
  that $\partm{v}\circ\partm{f} = \partm{w}\circ\partm{g}$, hence
  $v\circ \restr{f}{J}{B'} = w\circ \restr{g}{J}{C'}$ where
  $J\defeq\invf{f}(B') = \invf{g}(C')$. Since $f(J)=
  f(\invf{f}(B')) \subseteq B'$ then $J\subseteq \invf{f}(f(J))
  \subseteq \invf{f}(B') = J$, hence $\invf{f}(f(J))=J$ and similarly
  $\invf{g}(g(J)) = J$, so that $J\subseteq I$. This can be written
  $\invf{f}(B')\subseteq I$ and thus entails $B'\subseteq X$ and
  similarly $C'\subseteq Y$, hence $\invf{f'}(B') = J =
  \invf{g'}(C')$.

  We can therefore apply Lemma~\ref{lm-pushout-restr} and get that
  $\tuple{\restr{h}{B'}{Q'},\restr{k}{C'}{Q'}, Q'}$ is a pushout of
  $\tuple{J,\restr{f'}{J}{B'}, \restr{g'}{J}{C'}}$ where
  $Q'=h(B')\cup k(C')$, $\invf{h}(Q') = B'$ and $\invf{k}(Q') =
  C'$. Since $v\circ \restr{f'}{J}{B'} = v\circ \restr{f}{J}{B'} =
  w\circ \restr{g}{J}{C'} = w\circ \restr{g'}{J}{C'}$ there exists a
  unique $u:Q'\rightarrow U$ such that $u\circ \restr{h}{B'}{Q'}=v$
  and $w=u\circ \restr{k}{C'}{Q'}$. We thus have a partial morphism
  $\partm{u}:Q\hookleftarrow Q'\rightarrow U$ such that
  \[\partm{u}\circ\partm{h} = \partm{u\circ
      \restr{h}{\invf{h}(Q')}{Q'}} = \partm{u\circ
      \restr{h}{B'}{Q'}}= \partm{v} \] and similarly
  $\partm{u}\circ\partm{k} = \partm{w}$.

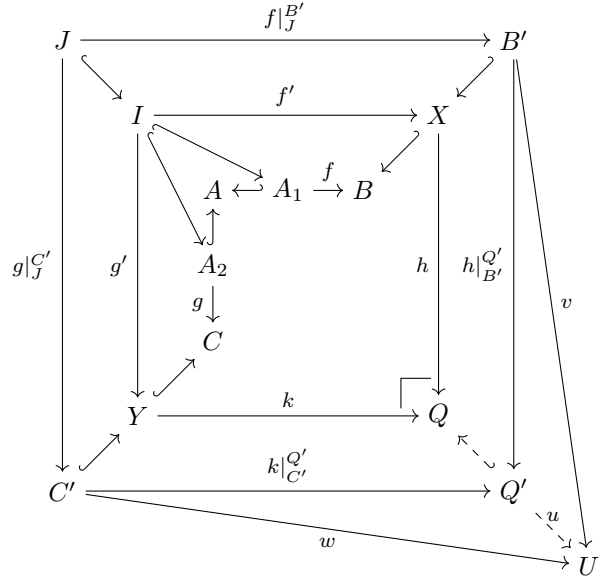
\begin{center}
    \begin{tikzpicture}[xscale=1,yscale=-1]
      \node (J) at (0,0) {$J$}; \node (I) at (1,1) {$I$}; \node (A) at (2,2) {$A$};
      \node (A1) at (3,2) {$A_1$}; \node (A2) at (2,3) {$A_2$};
      \node (B) at (4,2) {$B$}; \node (C) at (2,4) {$C$};
      \node (X) at (5,1) {$X$}; \node (Y) at (1,5) {$Y$}; \node (Q) at (5,5) {$Q$}; 
      \node (B') at (6,0) {$B'$}; \node (C') at (0,6) {$C'$}; \node (Q') at (6,6) {$Q'$};
      \node (U) at (7,7) {$U$}; 
      \path[->] (A1) edge node[above, font=\footnotesize] {$f$} (B); 
      \path[->] (A2) edge node[left, font=\footnotesize] {$g$} (C); 
      \path[->] (I) edge node[above, font=\footnotesize] {$f'$} (X); 
      \path[->] (I) edge node[left, font=\footnotesize] {$g'$} (Y);
      \path[->] (Y) edge node[above, font=\footnotesize] {$k$} (Q); 
      \path[->] (X) edge node[left, font=\footnotesize] {$h$} (Q); 
      \path[->] (C') edge node[above, font=\footnotesize] {$\restr{k}{C'}{Q'}$} (Q'); 
      \path[->] (B') edge node[left, font=\footnotesize] {$\restr{h}{B'}{Q'}$} (Q'); 
      \path[->] (C') edge node[below, font=\footnotesize] {$w$} (U); 
      \path[->] (B') edge node[right, font=\footnotesize] {$v$} (U); 
      \path[->] (J) edge node[above, font=\footnotesize] {$\restr{f}{J}{B'}$} (B'); 
      \path[->] (J) edge node[left, font=\footnotesize] {$\restr{g}{J}{C'}$} (C'); 
      \path[->,dashed] (Q') edge node[above, font=\footnotesize] {$u$} (U); 
      \path [arrows={Hooks[left]->}] (A1) edge (A);
      \path [arrows={Hooks[right]->}] (A2) edge (A);
      \path [arrows={Hooks[left]->}] (I) edge (A1);
      \path [arrows={Hooks[right]->}] (I) edge (A2);
      \path [arrows={Hooks[right]->}] (J) edge (I);
      \path [arrows={Hooks[right]->}] (Y) edge (C);
      \path [arrows={Hooks[left]->}] (X) edge (B);
      \path [arrows={Hooks[right]->}] (C') edge (Y);
      \path [arrows={Hooks[left]->}] (B') edge (X);
      \path [arrows={Hooks[left]->},dashed] (Q') edge (Q);
      \draw (4.5,4.9) to (4.5,4.5) to (4.9,4.5);
    \end{tikzpicture}
  \end{center}

  Suppose there is a $\partm{u'}: Q\hookleftarrow D\rightarrow U$ such that
  $\partm{u'}\circ\partm{h} = \partm{v}$ and $\partm{u'}\circ\partm{k}
  = \partm{w}$, then $u'\circ \restr{h}{\invf{h}(D)}{D} = v$ hence
  $\invf{h}(D)=B'$ and similarly $\invf{k}(D)=C'$. Since $D\subseteq Q
  = h(X)\cup k(Y)$ then
  \[D = (D\cap h(X))\cup (D\cap k(Y)) = h(\invf{h}(D)) \cup
    k(\invf{k}(D)) = h(B')\cup k(C') = Q'\]
  and we get $\partm{u'}=\partm{u}$ by the unicity of $u$.

  If $B$ and $C$ are finite (resp. standard, resp. $O$-monographs)
  then so are $X$ and $Y$, hence so is $Q$ by Theorem~\ref{th-po}.
\end{proof}

One important feature of this construction is illustrated below.

\begin{example}\label{ex-partialpo}
  Suppose there are edges $x$ of $A_1\cap A_2$ and
  $y\in\edges{A_2}\setminus \edges{A_1}$ such that $g(x)=g(y)$. If $x$
  is an edge of $I = \invf{g}(g(I))$ then so is $y$, which is
  impossible since $I\subseteq A_1\cap A_2$. Hence $x$ is not an edge
  of $I=\invf{f}(X)$ and therefore $f(x)\not\in \edges{X}$. Since $y$
  is not an edge of $I=\invf{g}(Y)$ then similarly
  $g(x)=g(y)\not\in\edges{Y}$. This means that even though $x$ has
  images by both $f$ and $g$, none of these has an image (by $h$ or
  $k$) in $Q$, i.e., they are ``deleted'' from the pushout.
\end{example}

The result of the present section can be replicated by replacing every
monograph, say $A$, by a typed monograph with a fixed type $T$, say
$a:A\rightarrow T$. But then expressions like $A\subseteq B$ are
replaced by $a\subseteq b$, which ought to be interpreted as
$A\subseteq B$ \emph{and} $a=\restr{b}{A}{}$, so that $\AlgFunc{T}{a}$
is then a subalgebra of $\AlgFunc{T}{b}$. In this way the results of
\cite{Lowe93} on categories of partial homomorphisms could be deduced
from Corollary~\ref{cr-gs2slicem}. They cannot be obtained directly
from Theorem~\ref{th-ppo}.

\section{Algebraic Transformations of Monographs}\label{sec-dpo}

Rule-based transformations of graphs are conceived as substitutions of
subgraphs (image of a left hand side of a rule) by subgraphs (image of
its right hand side). Substitutions are themselves designed as an
operation of deletion (of nodes or edges) followed by an operation of
addition. This last operation is conveniently represented as a
pushout, especially when edges are added between existing nodes
(otherwise a coproduct would be sufficient).

The operation of deletion is however more difficult to represent in
category theory, since there is no categorical notion of a
complement. This is a central and active issue in the field of
Algebraic Graph Transformation, and many definitions have been
proposed, see
\cite{CorradiniHHK06,CorradiniDEPR15,CorradiniDEPR19,BdlTE21a}. The
most common and natural one, known as the double pushout method
\cite{EhrigPS73,CorradiniMREHL97,journals/mscs/HabelMP01}, assumes the
operation of deletion as the inverse of the operation of addition.

More precisely, in the following pushout diagram
\begin{center}
  \begin{tikzpicture}[scale=1.5]
    \node (M) at (0,0) {$M$}; 
    \node (K) at (1,1) {$K$}; \node (L) at (0,1) {$L$};
    \node (D) at (1,0) {$D$}; 
    \path[->] (K) edge node[above, font=\footnotesize] {$l$} (L); 
    \path[->] (K) edge node[right, font=\footnotesize] {$k$} (D); 
    \path[->] (L) edge node[left, font=\footnotesize] {$m$} (M); 
    \path[->] (D) edge node[below, font=\footnotesize] {$f$} (M); 
    \draw (0.3,0.1) to (0.3,0.3) to (0.1,0.3);
  \end{tikzpicture}
\end{center}
we understand $M$ as the result of adding edges to $D$ as specified by
$l$ and $k$. Images of edges of $K$ are present in both $D$ and $L$,
and therefore also in $M$, without duplications (since $f\circ k =
m\circ l$). The edges that are added to $D$ are therefore the images
by $m$ of the edges of $L$ that do not occur in $l(K)$. We may then
inverse this operation and understand $D$ as the result of removing
these edges from $M$. The monograph $M$ and the morphisms $m$, $l$
then appear as the input of the operation, and the monograph $D$ and
morphisms $k$, $f$ as its output. The problem of course is that the
pushout operation is not generally bijective, hence it cannot always
be inverted. We first analyze the conditions of existence of $D$.

\begin{definition}[pushout complement, gluing condition]
  A \emph{pushout complement} of morphisms $l:K\rightarrow L$ and
  $m:L\rightarrow M$ is a monograph $D$ and a pair of morphisms
  $k:K\rightarrow D$ and $f:D\rightarrow M$ such that $\tuple{m,f,M}$
  is a pushout of $\tuple{K,l,k}$.

  The morphisms $l:K\rightarrow L$ and $m:L\rightarrow M$ satisfy the
  \emph{gluing condition} ($\glucond{l}{m}$ for short) if, for
  $L'=\edges{L}\setminus \relim{l}{\edges{K}}$,
  \begin{itemize}
  \item[{(1)}] for all $x,x'\in\edges{L}$, $m(x)=m(x')$ and $x\in
    L'$ entail $x=x'$, and
  \item[{(2)}] for all $e,e'\in\edges{M}$, $\occin{e}{M(e')}$ and
    $e\in \relim{m}{L'}$ entail $e'\in \relim{m}{L'}$.
  \end{itemize}
\end{definition}

The edges of $M$ that should be removed from $M$ to obtain $D$ are the
elements of $\relim{m}{L'}$. We may say that an edge $m(x)$ of $M$ is
\emph{marked for removal} if $x\in L'$ and \emph{marked for
  preservation} if $x\in \relim{l}{\edges{K}}$. Condition (1) of the
gluing condition states that the restriction of $m$ to
$\relim{\invf{m}}{\relim{m}{L'}}$ should be injective, or in other
words that an edge can be deleted if it is marked for removal once,
and not marked for preservation. Condition (2) states that an edge can
be deleted only if all the edges that are adjacent to it are also
deleted (otherwise these edges would be adjacent to a non existent
edge). It is obvious that this gluing condition reduces to the
standard one known on graphs, when applied to standard
$\set{0,2}$-monographs. We now prove that it characterizes the
existence of pushout complements (note that $l$ is not assumed to be
injective).

\begin{lemma}\label{lm-gc}
  The morphisms $l:K\rightarrow L$ and $m:L\rightarrow M$ have a
  pushout complement iff they satisfy the gluing condition.
\end{lemma}
\begin{proof}
  \emph{Necessary condition.} We assume w.l.o.g. that the pushout
  $\tuple{m,f,M}$ of $\tuple{K,l,k}$ is obtained by canonical
  construction, i.e., let $\tuple{L+D,\mu_1,\mu_2}$ be the coproduct
  of $\tuple{L,D}$, $\tuple{M,c}$ bet the coequalizer of
  $\tuple{\mu_1\circ l,\mu_2\circ k}$, $m=c\circ\mu_1$ and
  $f=c\circ\mu_2$. Thus $\edges{M}$ is the quotient of
  $\edges{L}+\edges{D}$ by the equivalence relation $\sim$
  generated by
  $R=\setof{\tuple{\mu_1\circ l(z),\mu_2\circ
      k(z)}}{z\in\edges{K}}$. Let
  $L'=\edges{L}\setminus \relim{l}{\edges{K}}$, we first prove
  {(1)} and then {(2)}.

  \begin{center}
    \begin{tikzpicture}
      \node (M) at (0,0) {$M$}; \node (LD) at (1,1) {$L+D$};
      \node (K) at (2,2) {$K$}; \node (L) at (0,2) {$L$};
      \node (D) at (2,0) {$D$};
      \path[->] (K) edge node[above, font=\footnotesize] {$l$} (L); 
      \path[->] (K) edge node[right, font=\footnotesize] {$k$} (D); 
      \path[->] (L) edge node[left, font=\footnotesize] {$m$} (M); 
      \path[->] (L) edge node[right,near start, font=\footnotesize] {$\mu_1$} (LD); 
      \path[->] (D) edge node[right,near end, font=\footnotesize] {$\mu_2$} (LD); 
      \path[->] (D) edge node[below, font=\footnotesize] {$f$} (M); 
      \path[->] (LD) edge node[above,near end, font=\footnotesize] {$c$} (M); 
    \end{tikzpicture}
  \end{center}

  For all $x,x'\in\edges{L}$, if $x\in L'$ then
  $x\not\in \relim{l}{\edges{K}}$, hence $\mu_1(x)$ is not related by
  $R$ to any element and is therefore alone in its
  $\sim$-class. Hence\footnote{Another consequence is that $\mu_1(x)$
    is not related by $\sim$ to any element of
    $\relim{\mu_2}{\edges{D}}$, hence that
    $m(x)\not\in \relim{f}{\edges{D}}$.\label{fn-1}} if $m(x)=m(x')$
  then $\mu_1(x)\sim \mu_1(x')$ and therefore $x=x'$.
  
  For all $e,e'\in \edges{M}$ such that $\occin{e}{M(e')}$ and
  $e\in \relim{m}{L'}$, let $x\in L'$ such that $e=m(x)$. Suppose that
  $e'=f(y')$ for some $y'\in\edges{D}$ then
  $M(e')=\seqto{f}{\alpha}\circ D(y')$ hence there is a
  $\occin{y}{D(y')}$ such that $e=f(y)$, hence
  $m(x)\in \relim{f}{\edges{D}}$ which is impossible by
  note~\ref{fn-1}. Since $M=f(D)\cup m(L)$
  there must be a
  $x'\in\edges{L}$ such that $e'=m(x')$. Suppose now that $x'=l(z)$
  for some $z\in\edges{K}$ then
  $e'=m(l(z)) = f(k(z))\in \relim{f}{\edges{D}}$, and we have seen
  this is impossible. Hence $x'\not\in \relim{l}{\edges{K}}$ and
  therefore $e'\in \relim{m}{L'}$.

  \emph{Sufficient condition.} We assume {(1)} and {(2)}, let $\alpha$
  be an ordinal for $M$,
  $\edges{D}\defeq \edges{M}\setminus \relim{m}{L'}$ and
  $D(e)\defeq M(e)$ for all $e\in\edges{D}$; by {(2)} this is an
  $\edges{D}$-sequence, hence $D$ is a submonograph of $M$ and the
  canonical injection $f:D\hookrightarrow M$ is a morphism. By {(1)}
  we have $\relim{m}{L'}\cap\relim{m\circ l}{\edges{K}} = \ensvide$,
  hence $\relim{m\circ l}{\edges{K}}\subseteq \edges{D}$ and we let
  $k\defeq \restr{(m\circ l)}{\edges{K}}{\edges{D}}$ so that
  $f\circ k = m\circ l$. We have
  \[\seqto{k}{\alpha} \circ K = \seqto{m}{\alpha}\circ
    \seqto{l}{\alpha}\circ K = \seqto{m}{\alpha}\circ L\circ l= M\circ
    m\circ l = D\circ k\] hence $k:K\rightarrow D$ is a morphism.

  
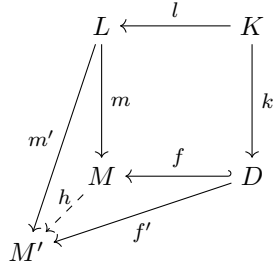
\begin{center}
    \begin{tikzpicture}
      \node (M) at (1,1) {$M$}; \node (M') at (0,0) {$M'$}; 
      \node (K) at (3,3) {$K$}; \node (L) at (1,3) {$L$};
      \node (D) at (3,1) {$D$};
      \path[->] (K) edge node[above, font=\footnotesize] {$l$} (L); 
      \path[->] (K) edge node[right, font=\footnotesize] {$k$} (D); 
      \path[->] (L) edge node[right, font=\footnotesize] {$m$} (M); 
      \path[arrows={Hooks[left]->}] (D) edge node[above, font=\footnotesize] {$f$} (M); 
      \path[->] (L) edge node[left, font=\footnotesize] {$m'$} (M'); 
      \path[->] (D) edge node[below, font=\footnotesize] {$f'$} (M');
      \path[->,dashed] (M) edge node[above, font=\footnotesize] {$h$} (M'); 
      \end{tikzpicture}
  \end{center}

  To prove that $\tuple{m,f,M}$ is a pushout of $\tuple{K,l,k}$, let
  $m':L\rightarrow M'$ and $f':D\rightarrow M'$ be morphisms such that
  $m'\circ l = f'\circ k$. Since
  $\edges{M}=\edges{D}\uplus \relim{m}{L'}$  we define
  $h:\edges{M}\rightarrow \edges{M'}$ as \[h(e)\defeq \left\{
      \begin{array}[c]{ll}
        f'(e) & \text{if } e\in\edges{D}\\
        m'(x) & \text{if } x\in L'\text{ and } e=m(x)
      \end{array}\right.\]
  since $x$ is unique by (1). For all $x\in\edges{L}$, if $x\in L'$
  then $h\circ m(x)= m'(x)$, otherwise there is a $z\in\edges{K}$ such
  that $x=l(z)$ and then
  \[h\circ m(x)= h\circ m\circ l(z) = h\circ f\circ k(z) = f'\circ k(z)
    = m'\circ l(z) = m'(x),\] hence $h\circ m = m'$. It is obvious
  that $h\circ f=f'$ and that these two equations uniquely determine
  $h$. Proving that $h:M\rightarrow M'$ is a morphism is straightforward.
\end{proof}

Note that $D$ is finite whenever $M$ is finite. This proves that this
gluing condition is also valid in $\FMonGr$, and it is obviously also
the case in $\SMonGr$, $\OMonGr{O}$ and $\OSMonGr{O}$ for every set
$O$ of ordinals. It therefore characterizes the existence of $D$, but
by no means its unicity.

It is well known (and easy to see) that in the category of sets one
may find pushout complements with non isomorphic sets $D$, this is
therefore also the case for monographs (since $\Sets\iso
\OMonGr{1}$). An analysis of the proof of Lemma~\ref{lm-gc} (necessary
condition) however yields that $\relim{f}{\edges{D}}$ is invariant.

\begin{corollary}\label{cr-detimD}
  If $D$, $k:K\rightarrow D$, $f:D\rightarrow M$ is a pushout
  complement of $l:K\rightarrow L$, $m:L\rightarrow M$ then
  $\relim{f}{\edges{D}} = \edges{M}\setminus\relim{m}{L'}$, where
  $L'=\edges{L}\setminus \relim{l}{\edges{K}}$.
\end{corollary}
\begin{proof}
  Since $\relim{m}{\edges{L}}\setminus\relim{(m\circ
  l)}{\edges{K}}\subseteq \relim{m}{L'}$ then
\[\relim{m}{\edges{L}}\setminus \relim{m}{L'} \subseteq \relim{(m\circ
    l)}{\edges{K}} = \relim{(f\circ k)}{\edges{K}} \subseteq
  \relim{f}{\edges{D}}.\] By property of pushouts we have
$\edges{M}=\relim{f}{\edges{D}}\cup \relim{m}{\edges{L}}$, and by
note~\ref{fn-1} we have
$\relim{m}{L'}\cap \relim{f}{\edges{D}}=\ensvide$, hence
\[\edges{M}\setminus\relim{m}{L'} = (\relim{f}{\edges{D}}\setminus
\relim{m}{L'}) \cup (\relim{m}{\edges{L}}\setminus \relim{m}{L'}) =
\relim{f}{\edges{D}}.\]
\end{proof}

One way of ensuring the unicity of $D$ (up to isomorphism) is to
assume that $l$ is injective: this is a well-known consequence of
Theorem~\ref{th-adhesive} (see \cite{LackS05}). However, an analysis
of the construction of $D$ in the proof of Lemma~\ref{lm-gc}
(sufficient condition) shows that we can always build $D$ as a
submonograph of $M$, hence we may as well assume that $f$ is a
canonical injection and avoid restrictions on $l$. We therefore adopt
a restricted notion of double pushout transformation compared to the
standard one.

\begin{definition}[span rules $\tuple{l,r}$, matching $m$, relation $\dpo{l}{r}{m}$]\label{def-dpo}
  A \emph{span rule} is a pair $\tuple{l,r}$ of morphisms
  $l:K\rightarrow L$, $r:K\rightarrow R$ with the same domain $K$. A
  \emph{matching} of $\tuple{l,r}$ in an object $M$ is a morphism
  $m:L\rightarrow M$. For any object $N$ we write
  $M\dpo{l}{r}{m} N$ if there exists a double-pushout diagram
  \begin{center}
    \begin{tikzpicture}[scale=1.5]
      \node (M) at (0,0) {$M$}; \node (R) at (2,1) {$R$};
      \node (K) at (1,1) {$K$}; \node (L) at (0,1) {$L$};
      \node (D) at (1,0) {$D$}; \node (N) at (2,0) {$N$};
      \path[->] (K) edge node[above, font=\footnotesize] {$l$} (L); 
      \path[->] (K) edge node[right, font=\footnotesize] {$k$} (D); 
      \path[->] (L) edge node[left, font=\footnotesize] {$m$} (M); 
      \path[->] (K) edge node[above, font=\footnotesize] {$r$} (R); 
      \path[->] (D) edge node[below, font=\footnotesize] {$g$} (N); 
      \path[arrows={Hooks[left]->}] (D) edge node[below, font=\footnotesize] {$f$} (M);
      \path[->] (R) edge node[right, font=\footnotesize] {$n$} (N);
      \draw (0.3,0.1) to (0.3,0.3) to (0.1,0.3);
      \draw (1.7,0.1) to (1.7,0.3) to (1.9,0.3);
    \end{tikzpicture}
  \end{center}
  where $f$ is a canonical injection.
\end{definition}

We easily see that the relation $\dpo{l}{r}{m}$ is 
deterministic up to isomorphism.

\begin{corollary}\label{th-detdpo}
  $M\dpo{l}{r}{m} N$ and $M\dpo{l}{r}{m} N'$ entail $N\iso N'$.
\end{corollary}
\begin{proof}
  We have two pushout complements $k:K\rightarrow D$,
  $f:D\hookrightarrow M$ and $k':K\rightarrow D'$,
  $f':D'\hookrightarrow M$ of $m$, $l$, hence by
  Corollary~\ref{cr-detimD}\[\edges{D} = \relim{f}{\edges{D}} =
    \edges{M}\setminus\relim{m}{L'} = \relim{f'}{\edges{D'}} =
    \edges{D'}\] hence $D=D'$, $f=f'$,
  $k = \restr{(f\circ k)}{K}{D} = \restr{(m\circ l)}{K}{D'} =
  \restr{(f'\circ k')}{K}{D'} = k'$, and therefore $N\iso N'$ by
  general property of pushouts.
\end{proof}

It is obvious by Theorem~\ref{th-po} and by the construction of $D$ in
Lemma~\ref{lm-gc} that, in the categories of
Definition~\ref{def-catmono}, there exists a $N$ such that
$M\dpo{l}{r}{m} N$ if and only if $l$ and $m$ satisfy the gluing
condition. This means in particular that an edge $e$ of $M$ may be
deleted only if it is explicitly marked for removal, i.e., if there
is an edge $x\in L'$ such that $m(x)=e$. All edges that are not marked
for removal are guaranteed to be preserved. This conservative
semantics for transformation rules is extremely safe but imposes a
discipline of programming that may be tedious.

As noted in Example~\ref{ex-partialpo}, pushout of partial morphisms
have a potential of removing edges. Since such pushouts always exist,
they can be used to define transformations that are not restricted by
the gluing condition. This is the idea of the single pushout method,
that was initiated in \cite{R84:tcs} and fully developed in
\cite{phd/dnb/Lowe91,Lowe93}.

\begin{definition}[partial rules $\partm{r}$, relation
  $\spo{\partm{r}}{m}$, rule $\partr{l}{r}$]\label{def-spo}
  A \emph{partial rule} is a partial morphism
  $\partm{r}:L\hookleftarrow K\rightarrow R$. A \emph{matching} of
  $\partm{r}$ in a monograph $M$ is a morphism $m:L\rightarrow M$. For
  any monograph $N$ we write $M\spo{\partm{r}}{m} N$ if there exist partial
  morphisms $\partm{g}$ and $\partm{n}$ such that
  $\tuple{\partm{n},\partm{g},N}$ is a pushout of
  $\tuple{L,\partm{r},\partm{m}}$.

  To any span rule $\tuple{l,r}$ where $l:K\rightarrow L$,
  $r:K\rightarrow R$ we associate a partial rule
  $\partr{l}{r} \defeq\partm{r'}:L\hookleftarrow l(K)\rightarrow R'$ such
  that $\tuple{q,r',R'}$ is a pushout of
  $\tuple{K,r,l'}$ where $l'\defeq \restr{l}{K}{l(K)}$.
  \begin{center}
    \begin{tikzpicture}[scale=1.5]
      \node (M) at (0,0) {$L$}; \node (R) at (2,1) {$R$};
      \node (K) at (1,1) {$K$}; \node (L) at (0,1) {$L$};
      \node (D) at (1,0) {$l(K)$}; \node (N) at (2,0) {$R'$};
      \node at (-1,1){$\tuple{l,r}$}; \node at (-1,0){$\partr{l}{r}$};
      \path[->] (K) edge node[above, font=\footnotesize] {$l$} (L); 
      \path[->] (K) edge node[left, font=\footnotesize] {$l'$} (D); 
      \path[->] (L) edge node[left, font=\footnotesize] {$\id{L}$} (M); 
      \path[->] (K) edge node[above, font=\footnotesize] {$r$} (R); 
      \path[->] (D) edge node[below, font=\footnotesize] {$r'$} (N); 
      \path [arrows={Hooks[left]->}] (D) edge (M); 
      \path[->] (R) edge node[right, font=\footnotesize] {$q$} (N);
      \draw (1.7,0.1) to (1.7,0.3) to (1.9,0.3);
    \end{tikzpicture}
  \end{center}
\end{definition}

The relation $\spo{\partm{r}}{m}$ is also deterministic up to
isomorphism since $N$ is obtained as a pushout.  Obviously a morphism
$m$ is a matching of $\tuple{l,r}$ in $M$ iff it is a matching of
$\partr{l}{r}$ in $M$. The partial rule $\partr{l}{r}$ is designed to
perform the same transformation as the span rule $\tuple{l,r}$. We
prove that this is indeed the case when the gluing condition holds.

\begin{theorem}
  For any span rule $\tuple{l,r}$, monographs $M$, $N$ and matching
  $m$ of $\tuple{l,r}$ in $M$, we have
  \[M\dpo{l}{r}{m} N\ \text{ iff }\ M\spo{\partr{l}{r}}{m} N \text{ and }\glucond{l}{m}.\]
\end{theorem}
\begin{proof}
  Let $R'$, $l'$, $q$ and $r'$ be as in Definition~\ref{def-spo}. We
  first compute the pushout of $\partr{l}{r}$ and $\partm{m}$
  according to the construction in Lemma~\ref{th-ppo}, by assuming the
  gluing condition $\glucond{l}{m}$ and that $D\subseteq M$,
  $k:K\rightarrow D$, $f:D\hookrightarrow M$ is a pushout
  complement of $l$, $m$.

  Let $I$ be the greatest submonograph of $l(K)\cap L$ such that
  $\invf{r'}(r'(I))=I$ and $\invf{m}(m(I))=I$. By $\glucond{l}{m}$ (1)
  we have for all $x\in\edges{L}$ that
  $m(x)\in \relim{m}{\relim{l}{\edges{K}}}$ entails
  $x\not\in L' = \edges{L}\setminus \relim{l}{\edges{K}}$, i.e.,
  $x\in\relim{l}{\edges{K}}$, hence $\invf{m}(m(l(K))) \subseteq l(K)$
  and since the reverse inclusion is always true we get
  $I=l(K)$. Hence the greatest monograph $X\subseteq R'$ such that
  $\invf{r'}(X)\subseteq I$ is $R'$.

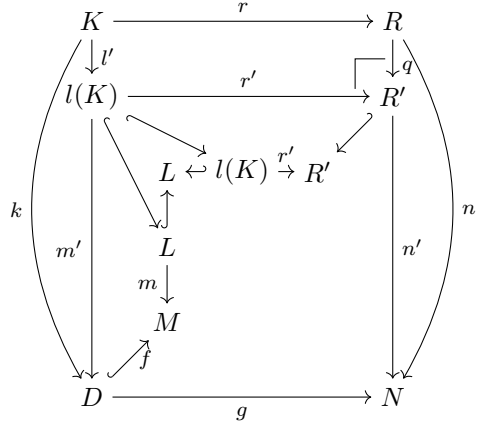
 \begin{center}
    \begin{tikzpicture}[xscale=1,yscale=-1]
      \node (K) at (1,0) {$K$}; \node (R) at (5,0) {$R$};
      \node (I) at (1,1) {$l(K)$}; \node (A) at (2,2) {$L$};
      \node (A1) at (3,2) {$l(K)$}; \node (A2) at (2,3) {$L$};
      \node (B) at (4,2) {$R'$}; \node (C) at (2,4) {$M$};
      \node (X) at (5,1) {$R'$}; \node (Y) at (1,5) {$D$}; \node (Q) at (5,5) {$N$}; 
      \path[->] (A1) edge node[above, font=\footnotesize] {$r'$} (B); 
      \path[->] (A2) edge node[left, font=\footnotesize] {$m$} (C); 
      \path[->] (I) edge node[above, font=\footnotesize] {$r'$} (X); 
      \path[->] (I) edge node[left, font=\footnotesize] {$m'$} (Y);
      \path[->] (Y) edge node[below, font=\footnotesize] {$g$} (Q); 
      \path[->] (X) edge node[right, font=\footnotesize] {$n'$} (Q);
      \path[->] (K) edge node[above, font=\footnotesize] {$r$} (R);
      \path[->] (K) edge node[right, font=\footnotesize] {$l'$} (I);
      \path[->] (R) edge node[right, near end, font=\footnotesize] {$q$} (X);
      \path[->] (R) edge [bend right] node[right,font=\footnotesize] {$n$} (Q);
      \path[->] (K) edge [bend left] node[left,font=\footnotesize] {$k$} (Y);
      \path [arrows={Hooks[left]->}] (A1) edge (A);
      \path [arrows={Hooks[right]->}] (A2) edge (A);
      \path [arrows={Hooks[left]->}] (I) edge (A1);
      \path [arrows={Hooks[right]->}] (I) edge (A2);
      \path [arrows={Hooks[right]->}] (Y) edge node[right, font=\footnotesize] {$f$} (C);
      \path [arrows={Hooks[left]->}] (X) edge (B);
      \draw (4.5,0.9) to (4.5,0.5) to (4.9,0.5);
    \end{tikzpicture}
  \end{center}

  Let $Y$ be the greatest submonograph of $M$ such that
  $\invf{m}(Y)\subseteq l(K)$, this entails
  $\relim{\invf{m}}{\edges{Y}}\cap L'=\ensvide$, hence $\edges{Y}\cap
  \relim{m}{L'}=\ensvide$ and by Corollary~\ref{cr-detimD} $Y\subseteq
  f(D)=D$. Conversely, for all
  $x\in\relim{\invf{m}}{\edges{D}} =
  \relim{\invf{m}}{\edges{M}\setminus\relim{m}{L'}}$ we have
  $m(x)\not\in\relim{m}{L'}$, hence by $\glucond{l}{m}$ (1)
  $x\not\in L'$ and thus $x\in\relim{l}{\edges{K}}$, so that
  $\invf{m}(D)\subseteq l(K)$. Hence $D\subseteq Y$ and we get
  $Y= D$.

  The pushout of $\partr{l}{r}$ and $\partm{m}$ is therefore obtained
  from the pushout of $r'$ and $m'\defeq\restr{m}{l(K)}{D}$.  Besides,
  we have
  $m'\circ l'=\restr{(m\circ l)}{K}{D} = \restr{(f\circ k)}{K}{D}=k$.

  \emph{Sufficient condition.} We assume $M\dpo{l}{r}{m} N$ and the
  diagram in Definition~\ref{def-dpo}. By Lemma~\ref{lm-gc} we have
  $\glucond{l}{m}$. By the above we get
  $(g\circ m')\circ l' = g\circ k = n\circ r$, and since
  $\tuple{q,r',R'}$ is a pushout of $\tuple{K,r,l'}$ then there exists
  a unique $n':R'\rightarrow N$ such that $n'\circ r'=g\circ m'$ and
  $n'\circ q = n$. Since $\tuple{n,g,N}$ is a pushout of
  $\tuple{K,r,k}$ then by pushout decomposition $\tuple{n',g,N}$ is a
  pushout of $\tuple{l(K),r',m'}$, hence $M\spo{\partr{l}{r}}{m} N$.

  \emph{Necessary condition.} By $\glucond{l}{m}$ and Lemma~\ref{lm-gc}
  we can build a pushout complement $D\subseteq M$,
  $k:K\rightarrow D$, $f:D\hookrightarrow M$ of $l$, $m$. By
  $M\spo{\partr{l}{r}}{m} N$ and the above there is a pushout
  $\tuple{n',g,N}$ of $\tuple{l(K),r',m'}$, hence by pushout
  composition $\tuple{N,n'\circ q,g}$ is a pushout of $\tuple{K,r,k}$,
  hence $M\dpo{l}{r}{m} N$.
\end{proof}

Note that any partial rule $\partm{r}:L\hookleftarrow K\rightarrow R$ can
be expressed as $\partm{r}=\partr{j}{r}$ where $j:K\hookrightarrow L$ is
the canonical injection. Thus, provided the gluing condition holds,
single and double pushout transformations are equivalent. Single
pushout transformations are more expressive since they also apply when
the gluing condition does not hold, as illustrated in the following
example.

\begin{example}
  We consider the following ``loop removing'' rule:
  
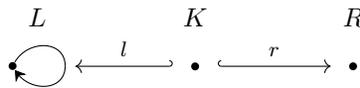
\begin{center}
    \begin{tikzpicture}
      [xscale=1.5,yscale=0.5]
      \node at (-1.4,1.3) {$L$}; \node at (0,1.3) {$K$}; \node at (1.4,1.3) {$R$};
      \node (L) at (-1.4,0) {
        \raisebox{-1pt}{\begin{tikzpicture}[point/.style={circle,inner sep=0pt,minimum
              size=3pt,fill=black},scale=0.7,> /.tip={Stealth[width=5pt,
              length=4pt]}] 
            \node (N1) [point] at (0,0) {};
            \path[-,out=60,in=90,distance=0.5cm] (N1) edge (1,0);
            \path[->,out=-90,in=-60,distance=0.5cm] (1,0) edge (N1);
          \end{tikzpicture}}};
      \node (K) at (0,0) {\begin{tikzpicture}[point/.style={circle,inner sep=0pt,minimum
            size=3pt,fill=black},scale=0.7,> /.tip={Stealth[width=5pt,
            length=4pt]}] 
          \node (N1) [point] at (0,0) {};
        \end{tikzpicture}};
      \node (R) at (1.4,0) {\raisebox{-1pt}{\begin{tikzpicture}[point/.style={circle,inner sep=0pt,minimum
            size=3pt,fill=black},scale=0.7,> /.tip={Stealth[width=5pt,
            length=4pt]}] 
          \node (N1) [point] at (0,0) {};
        \end{tikzpicture}}};
      \path [arrows={Hooks[left]->}] (-0.2,0) edge node[above, font=\footnotesize] {$l$} (L);
      \path [arrows={Hooks[right]->}] (0.2,0) edge node[above, font=\footnotesize] {$r$} (1.2,0);
    \end{tikzpicture}
  \end{center}
  and try to apply it to monograph $\tginf$ from
  Example~\ref{ex-poly}. There is a unique morphism
  $m:L\rightarrow \tginf$ but it does not satisfy the gluing
  condition. Indeed, we see that condition (2) is breached since
  $\occin{{1}}{\tginf({2})}$ and
  ${1}\in \relim{m}{L'}$ and yet
  $2\not\in\relim{m}{L'}$. Hence the only way to apply the rule
  to $\tginf$ is through a single pushout transformation.

  For this we first compute the rule $\partr{l}{r}$. Since $l$ is the
  canonical injection of $l(K)=K$ into $L$, then $r'=r$ (and $R'=R=K$)
  and hence $\partr{l}{r} = \partm{r}:L\hookleftarrow K\rightarrow R$. The
  monograph $D$ is the greatest one such that
  $D\subseteq \tginf$ and $\invf{m}(D)\subseteq l(K)$, hence
  obviously $D=\set{\tuple{0,\empstr}}$. Since $l(K)$ and $R$
  are both isomorphic to $D$ then so is the result of the
  transformation, i.e.,
  \[
    \raisebox{-2.8ex}{\begin{tikzpicture}
      \node (L) at (0,0) {
        \raisebox{-1pt}{\begin{tikzpicture}[point/.style={circle,inner sep=0pt,minimum
              size=3pt,fill=black},scale=0.7,> /.tip={Stealth[width=5pt,
              length=4pt]}] 
            \node (N1) [point] at (0,0) {};
            \path[-,out=60,in=90,distance=0.5cm] (N1) edge (1,0);
            \path[->,out=-90,in=-60,distance=0.5cm] (1,0) edge (N1);
            \path[-,out=30,in=90,distance=0.5cm] (0.9,0.28) edge (1.9,0);
            \path[->,out=-90,in=-30,distance=0.5cm] (1.9,0) edge (0.9,-0.28);
            \path[-,out=30,in=90,distance=0.5cm] (1.82,0.28) edge (2.8,0);
            \path[->,out=-90,in=-30,distance=0.5cm] (2.8,0) edge (1.82,-0.28);
            \node at (3.3,0){$\cdots$};
          \end{tikzpicture}}};
    \end{tikzpicture}}
    \ \spo{\partr{l}{r}}{m}\
    \raisebox{-0.6ex}{\begin{tikzpicture}
      \node (R) at (0,0) {\raisebox{-1pt}{\begin{tikzpicture}[point/.style={circle,inner sep=0pt,minimum
            size=3pt,fill=black},scale=0.7,> /.tip={Stealth[width=5pt,
            length=4pt]}] 
          \node (N1) [point] at (0,0) {};
        \end{tikzpicture}}};
    \end{tikzpicture}}
    \]
    Hence removing the edge $1$ from $\tginf$ silently
    removes the edges $n$ for all $n>1$.
\end{example}

We therefore see that single pushouts implement a semantics where
edges can be silently removed, but minimally so for a monograph to be
obtained. This may remove edges in a cascade, a feature
that does not appear on graphs.
Note that item (1) of the gluing condition may also be breached when
an edge is marked more than once for removal, in which case it is
deleted, but also when an edge is marked both for removal and for
preservation. Example~\ref{ex-partialpo} shows that in such cases the
edge is also removed. All edges marked for removal are guaranteed to
be deleted, and the other edges are preserved only if this does not
conflict with deletions. This semantics of transformation rules is
thus dual to the previous one, and should be more appealing to the
daring (or lazy) programmer.

\section{Attributed Typed Monographs}\label{sec-atm}

The notion of E-graph has been designed in \cite{EhrigEPT06} in order
to obtain an adhesive category of graphs with attributed nodes and
edges. This follows from a line of studies on Typed Attributed Graph
Transformations, see
\cite{Loewe-Korff-Wagner:93,DBLP:conf/gg/HeckelKT02,DBLP:journals/fuin/EhrigEPT06}. The
attributes are taken in a data type algebra and may be of different
sorts (booleans, integers, strings, etc.). In the case of E-graphs
only the nodes of sort \texttt{values} represent such attributes. But
they are also typed by E-graphs, and in the type E-graphs each node of
sort \texttt{values} represent a sort of the data type algebra. This
should recall the constructions of Section~\ref{sec-graphstruct} that
we now use in order to generalize the notion of typed attributed
graphs given in \cite{EhrigEPT06}. The idea is similarly to impose
that the edges typed by a sort of a data type algebra are the elements
of the corresponding carrier set.

\begin{definition}[categories $\ATM{T}{\Sig}$]
  For any monograph $T$ and signature
  $\Sig:\Op\rightarrow \seqto{S}{\omega}$, an \emph{attributed typed
    monograph} (ATM for short) over $T$, $\Sig$ is a pair
  $\tuple{a,\Alg}$ of an object $a:A\rightarrow T$ in $\sliceM{T}$ and
  a $\Sig$-algebra $\Alg$ such that $\Alg_s = (\AlgFunc{T}a)_s$ for
  all $s\in S\cap \edges{T}$.

  A \emph{morphism} $m$ from $\tuple{a,\Alg}$ to an ATM
  $\tuple{b,\Blg}$ over $T$, $\Sig$ is a pair
  $\tuple{\tmorph{m},\amorph{m}}$ of a morphism
  $\tmorph{m}:a\rightarrow b$ in $\sliceM{T}$ and a
  $\Sig$-homomorphism $\amorph{m}:\Alg\rightarrow\Blg$ such that
  $\amorph{m}_s = (\AlgFunc{T}\tmorph{m})_s$ for all
  $s\in S\cap \edges{T}$.

  Let $\id{\tuple{a,\Alg}}\defeq \tuple{\id{a},\id{\Alg}}$ and for any
  morphism $m':\tuple{b,\Blg}\rightarrow \tuple{c,\Clg}$ let $m'\circ
  m\defeq\tuple{\tmorph{m}'\circ\tmorph{m},
    \amorph{m}'\circ\amorph{m}}$ that is a morphism from
  $\tuple{a,\Alg}$ to $\tuple{c,\Clg}$. Let $\ATM{T}{\Sig}$ be the
  category of ATMs over $T$, $\Sig$ and their morphisms.
\end{definition}

The edges that are considered as attributes are not the nodes of a
specific sort as in E-graphs; they are characterized by the fact that
they are typed by an edge of $T$ that happens to be also a sort of the
data type signature $\Sig$, i.e., an element of $S$. This is
consistent with the typed attributed E-graphs of \cite{EhrigEPT06}.

We therefore see that the signatures $\SigFunc{T}$ and $\Sig$ share
sorts but we shall consider them as otherwise distinct, in particular
w.r.t. operator names. To account for this property we need the
following construction.

\begin{definition}[signature $\Sig\sigplus\Sig'$]
  Given two signatures $\Sig:\Op\rightarrow\seqto{S}{\omega}$ and
  $\Sig':\Op'\rightarrow\seqto{S'}{\omega}$, let
  $\tuple{\Op+\Op',\mu_1,\mu_2}$ be the coproduct of
  $\tuple{\Op,\Op'}$ in $\Sets$ and $j$,
  $j'$ be the canonical injections of $S$,
  $S'$ respectively into $S\cup
  S'$, let $\Sig\sigplus\Sig':\Op+\Op'\rightarrow \seqto{(S\cup
    S')}{\omega}$ be the unique function such that
  $(\Sig\sigplus\Sig')\circ\mu_1 = \seqto{j}{\omega}\circ
  \Sig$ and $(\Sig\sigplus\Sig')\circ\mu_2 = \seqto{j'}{\omega}\circ
  \Sig'$.

\begin{center}
  \begin{tikzpicture}[xscale=3,yscale=1.2]
    \node (SO) at (0,0) {$\Op+\Op'$}; \node (US) at (1,0) {$\seqto{(S\cup S')}{\omega}$};
    \node (O) at (0,1) {$\Op$}; \node (O') at (0,-1) {$\Op'$};
    \node (S) at (1,1) {$\seqto{S}{\omega}$}; \node (S') at (1,-1) {$\seqto{S'}{\omega}$};
    \path[->] (O) edge node [left, font=\footnotesize]{$\mu_1$} (SO);
    \path[->] (O') edge node [left, font=\footnotesize]{$\mu_2$} (SO);
    \path[->] (O) edge node [above, font=\footnotesize]{$\Sig$} (S);
    \path[->] (O') edge node [above, font=\footnotesize]{$\Sig'$} (S');
    \path[{Hooks[left]}->] (S) edge node [right, font=\footnotesize]{$\seqto{j}{\omega}$} (US);
    \path[{Hooks[right]}->] (S') edge node [right, font=\footnotesize]{$\seqto{j'}{\omega}$} (US);
    \path[->,dashed] (SO) edge node [above, font=\footnotesize]{$\Sig\sigplus\Sig'$} (US);
  \end{tikzpicture}
\end{center}

\end{definition}

We leave it to the reader to check that this construction defines a
coproduct in the category $\SigCats$ and therefore that
$\Sig_1\dotiso\Sig_2$ and $\Sig'_1\dotiso\Sig'_2$ entail
$\Sig_1\sigplus\Sig'_1\dotiso \Sig_2\sigplus\Sig'_2$.  For the sake of
simplicity we will assume in the sequel that $\SigFunc{T}$ and $\Sig$
have no operator name in common, thus assimilate $\Funcs{T}+\Op$ to
$\Funcs{T}\cup\Op$ and omit the canonical injections, so that
$\SigFunc{T}=\restr{(\SigFunc{T}\sigplus\Sig)}{\Funcs{T}}{\seqto{(\edges{T})}{\omega}}$
and $\Sig=\restr{(\SigFunc{T}\sigplus\Sig)}{\Op}{\seqto{S}{\omega}}$.

\begin{definition}[functor $\DFunc{}: \ATM{T}{\Sig}\rightarrow \AlgCat{(\SigFunc{T}\sigplus\Sig)}$]
  For every signature $\Sig:\Op\rightarrow\seqto{S}{\omega}$ and
  monograph $T$ such that $\Funcs{T}\cap\Op=\ensvide$, let
  $\Sig'\defeq \SigFunc{T}\sigplus\Sig$ and
  $\DFunc{}:\ATM{T}{\Sig}\rightarrow \AlgCat{\Sig'}$ be the functor
  defined as follows: for every object $\tuple{a,\Alg}$ of
  $\ATM{T}{\Sig}$ let $\DFunc{\tuple{a,\Alg}}$ be the $\Sig'$-algebra
  $\Alg'$ defined by
  \begin{itemize}
  \item $\Alg'_s\defeq \Alg_s$ for all $s\in S$ and
    $\Alg'_e\defeq (\AlgFunc{T}a)_e$ for all $e\in\edges{T}$,
  \item $\interp{o}{\Alg'}\defeq \interp{o}{\Alg}$ for all $o\in\Op$
    and
    $\interp{\opname{e}{\iota}}{\Alg'}\defeq
    \interp{\opname{e}{\iota}}{\AlgFunc{T}a}$ for all
    $\opname{e}{\iota}\in\Funcs{T}$.
  \end{itemize}
  For every morphism $m:\tuple{a,\Alg}\rightarrow \tuple{b,\Blg}$, let
  $(\DFunc{m})_s\defeq \amorph{m}_s$ for all $s\in S$ and
  $(\DFunc{m})_e\defeq (\AlgFunc{T}\tmorph{m})_e$ for all
  $e\in\edges{T}$.
\end{definition}

It is straightforward to check that $\DFunc{m}$ is a
$\Sig'$-homomorphism from $\DFunc{\tuple{a,\Alg}}$ to
$\DFunc{\tuple{b,\Blg}}$, and hence that $\DFunc{}$ is a functor.
  
\begin{theorem}\label{th-ATM2Alg}
  $\DFunc{}$ is an equivalence from $\ATM{T}{\Sig}$ to
  $\AlgCat{(\SigFunc{T}\sigplus\Sig)}$.
\end{theorem}
\begin{proof}
  It is easy to see that $\DFunc{}$ is full and faithful by the same
  property of $\AlgFunc{T}$.

  We prove that $\DFunc{}$ is isomorphism-dense.  For any
  $\Sig'$-algebra $\Blg'$, let $\Blg$ (resp. $\Clg$) be its
  restriction to $\Sig$ (resp. $\SigFunc{T}$).  Since $\AlgFunc{T}$ is
  isomorphism-dense by Theorem~\ref{thm-sliceiso}, there exist an
  object $a:A\rightarrow T$ in $\sliceM{T}$ and an
  $\SigFunc{T}$-isomorphism $h:\AlgFunc{T}a\rightarrow\Clg$.  We
  define simultaneously a set $\Alg_s$ and a function
  $k_s:\Alg_s\rightarrow\Blg_s$ for all $s\in S$ by taking
  $\Alg_s\defeq \Blg_s$ and $k_s\defeq \id{\Alg_s}$ if
  $s\in S\setminus\edges{T}$, and $\Alg_s\defeq (\AlgFunc{T}a)_s$ and
  $k_s\defeq h_s$ if $s\in S\cap\edges{T}$ (in this case we have
  $\Clg_s=\Blg'_s=\Blg_s$). We then define for every $o\in\Op$ the
  function
  $\interp{o}{\Alg}\defeq \invf{k_{\Rng{o}}}\circ
  \interp{o}{\Blg}\circ k_{\Dom{o}}:\Alg_{\Dom{o}}\rightarrow
  \Alg_{\Rng{o}}$, and the $\Sig$-algebra
  $\Alg\defeq\big((\Alg_s)_{s\in
    S},(\interp{o}{\Alg})_{o\in\Op}\big)$. By construction
  $\tuple{a,\Alg}$ is obviously an ATM over $T,\Sig$ and
  $k\defeq(k_s)_{s\in S}$ is a $\Sig$-isomorphism
  $k:\Alg\rightarrow\Blg$.

  \begin{center}
    \begin{tikzpicture}[scale=.8]
      \draw (-1.8,1.5) circle (1cm and 1.3cm);
      \draw (1.8,1.5) circle (1cm and 1.3cm);
      \draw (-1.8,0) circle (1cm and 1.3cm);
      \draw (1.8,0) circle (1cm and 1.3cm);
      \draw (-3.2,1.5) node {$\Alg$};
      \draw (-5.4,1.5) node {$\Sig$};
      \draw (3.2,1.5) node {$\Blg$};
      \draw (-3.4,0) node {$\AlgFunc{T}a$};
      \draw (-5.4,0) node {$\SigFunc{T}$};
      \draw (3.2,0) node {$\Clg$};
      \draw (-5.4,-1.7) node {$\Sig'$};
      \draw (-1.8,-1.7) node {$\Alg'$};
      \draw (1.8,-1.7) node {$\Blg'$};
      \path[->] (-0.8,01.5) edge node [above, font=\footnotesize]{$k$} (0.8,1.5);
      \path[->] (-0.8,0) edge node [above, font=\footnotesize]{$h$} (0.8,0);
    \end{tikzpicture}
  \end{center}

  Let $\Alg'\defeq\DFunc{\tuple{a,\Alg}}$, $h'_s\defeq
  k_s:\Alg'_s\rightarrow \Blg'_s$ for all $s\in S$ and $h'_e\defeq
  h_e:\Alg'_e\rightarrow \Blg'_e$ for all $e\in\edges{T}$, since
  $h_s=k_s$ for all $s\in S\cap\edges{T}$ then
  $h'\defeq(h'_s)_{s\in S\cup\edges{T}}$ is well-defined. It is then
  easy to see that $h':\Alg'\rightarrow \Blg'$ is a
  $\Sig'$-isomorphism, so that $\DFunc{\tuple{a,\Alg}}\iso\Blg'$.
\end{proof}

Theorem~\ref{th-ATM2Alg} generalizes\footnote{Our proof is also much
  shorter than the 6 pages taken by the corresponding result on
  attributed typed E-graphs. This is due partly to our use of
  $\AlgFunc{T}$ (Definition~\ref{def-algfunc}) and of
  Theorem~\ref{thm-sliceiso}, but also to the simplicity of monographs
  compared to the 5 sorts and 6 operator names of E-graphs.}
\cite[Theorem 11.3]{EhrigEPT06} that establishes an isomorphism
between the category of attributed E-graphs typed by an attributed
E-graph ${ATG}$ and the category of algebras of a signature denoted
$\mathrm{AGSIG}(ATG)$. In particular Theorem~11.3 of \cite
{EhrigEPT06} requires the hypothesis that $\mathrm{AGSIG}(ATG)$ should
be \emph{well-structured}, which means that if there is an operator
name of $\SigFunc{T}$ whose domain sort is $s$ then $s$ is not a sort
of the data type signature $\Sig$. Obviously this is equivalent to
requiring that only nodes of $T$ can be considered as sorts of $\Sig$
and is linked to the fact that only \texttt{values} nodes of E-graphs
are supposed to hold attributes. Since we are not restricted to
E-graphs there is no need to require that attributes should only be
nodes. This has an interesting consequence:

\begin{corollary}
  For every signatures $\Sig$, $\Sig'$ and graph structure $\Msig$
  such that $\Sig'=\Msig\sigplus\Sig$ there exists a monograph
  $T$ such that $\AlgCat{\Sig'}\equivCat\ATM{T}{\Sig}$.
\end{corollary}
\begin{proof}
  By Lemma~\ref{lm-Sig2Mono} there exists a monograph $T$ such that
  $\SigFunc{T}\dotiso \Msig$, hence
  $\SigFunc{T}\sigplus \Sig \dotiso \Msig\sigplus\Sig = \Sig'$ and
  therefore
  $\AlgCat{\Sig'}\iso \AlgCat{(\SigFunc{T}\sigplus\Sig)} \equivCat
  \ATM{T}{\Sig}$.
\end{proof}

Obviously, any signature $\Sig'$ can be decomposed as
$\Msig\sigplus\Sig$ by putting some of its monadic operators (and the
sorts involved in these) in $\Msig$ and all other operators in
$\Sig$. And then any $\Sig'$-algebra can be represented as an ATM over
$T,\Sig$, where $\SigFunc{T}\dotiso \Msig$. This opens the way to
applying graph transformations to these algebras, but this requires
some care since it is not generally possible to remove or add elements
to a $\Sig'$-algebra and obtain a $\Sig'$-algebra as a result.

The approach adopted in \cite[Definition 11.5]{EhrigEPT06} is to
restrict the morphisms used in span rules to a class of monomorphisms
that are extensions of $\Sig$-isomorphisms to
$(\Msig\sigplus\Sig)$-homomorphisms. It is then possible to show
\cite[Theorem 11.11]{EhrigEPT06} that categories of typed attributed
E-graphs are adhesive HLR categories (a notion that generalizes
Definition~\ref{def-adhesive}, see
\cite{DBLP:journals/fuin/EhrigPPH06}) w.r.t. this class of
monomorphisms.

A similar result holds on categories of ATMs. For the sake of
simplicity, and since rule-based graph transformations are unlikely to
modify attributes such as booleans, integers or strings (and if they
do they should probably not be considered as graph transformations),
we will only consider morphisms that leave the data type algebra
unchanged, element by element. This leaves the possibility to
transform the edges whose sort is in $\Msig$ but not in $\Sig$.

\begin{definition}[categories $\ATM{T}{\Alg}$, functor $\UFunc{}$, $f$
  stabilizes $\Alg$]\sloppy
  For any $\Sig$-algebra $\Alg$ let $\ATM{T}{\Alg}$ be the subcategory
  of $\ATM{T}{\Sig}$ restricted to objects $\tuple{a,\Alg}$ and
  morphisms $\tuple{f, \id{\Alg}}$.

    The \emph{forgetful functor} $\UFunc{}:\ATM{T}{\Alg}\rightarrow \Sets$
  is defined by $\UFunc{\tuple{a,{\Alg}}}  \defeq \edges{A}$, where
  $a:A\rightarrow T$ and $\UFunc{\tuple{f,\id{\Alg}}}  \defeq
  \edges{f}$ (usually denoted $f$).

  By abuse of notation we write $\Alg$ for the set
  $\bigcup_{s\in S\cap\edges{T}}\Alg_s$. A function $f$
  \emph{stabilizes} $\Alg$ if $\relim{\invf{f}}{x} = \set{x}$ for all
  $x\in\Alg$.
\end{definition}

The proof that the categories $\ATM{T}{\Alg}$ are adhesive 
will only be sketched below. The key point is the following lemma.

\begin{lemma}\label{lm-stab}
  For all objects $\tuple{a,\Alg}$, $\tuple{b,\Alg}$ of
  $\ATM{T}{\Alg}$ and morphism $f:a\rightarrow b$ of $\sliceM{T}$, we
  have \[\tuple{f,\id{\Alg}}: \tuple{a,\Alg} \rightarrow
  \tuple{b,\Alg}\text{ is a morphism in } \ATM{T}{\Alg}\ \text{ iff }\ f\text{ stabilizes }\Alg.\]
\end{lemma}
\begin{proof}
  For all $s\in S\cap\edges{T}$ we have
  $\Alg_s=(\AlgFunc{T}a)_s = \relim{\invf{a}}{s}$ and
  $\Alg_s= \relim{\invf{b}}{s}$. Since $b\circ f = a$ then
  $\relim{\invf{f}}{\Alg_s}=\relim{\invf{f}}{\relim{\invf{b}}{s}} =
  \relim{\invf{a}}{s}=\Alg_s$, hence
  $\relim{\invf{f}}{\Alg}=\Alg$. Thus $f$ stabilizes $\Alg$ iff
  $f(x)=x$ for all $x\in\Alg$ iff
  $(\AlgFunc{T}f)_s = \restr{f}{\Alg_s}{\Alg_s} = \Id{\Alg_s} =
  (\id{\Alg})_s$ for all $s\in S\cap\edges{T}$ iff
  $\tuple{f,\id{\Alg}}$ is a morphism in $\ATM{T}{\Alg}$.
\end{proof}

Hence the property of stabilization characterizes the difference
between morphisms in $\sliceM{T}$ and morphisms in
$\ATM{T}{\Alg}$. Besides, it is well-known how pushouts and pullbacks
in $\sliceM{T}$ can be constructed from those in $\MonGr$, and we have
seen that these can be constructed from those in $\Sets$.

But then it is quite obvious that in $\Sets$, starting from a span of
functions that stabilize $\Alg$, it is always possible to find as
pushout a cospan of functions that stabilize $\Alg$. Hence not only
does $\ATM{T}{\Alg}$ have pushouts, but these are preserved by the
functor $\UFunc{}$. A similar result holds for pullbacks, and a
construction similar to Corollary~\ref{cr-monoinj} yields that
$\UFunc{}$ also preserves monomorphisms. Finally, we see that
$\UFunc{}$ reflects isomorphisms since $\invf{f}$ stabilizes $\Alg$
whenever $f$ does. We conclude as in Theorem~\ref{th-adhesive}.

\begin{theorem}
  $\ATM{T}{\Alg}$ is adhesive.
\end{theorem}

This result does not mean that all edges that are not attributes can
be freely transformed. Their adjacencies to or from attributes may
impose constraints that only few morphisms are able to satisfy.

\begin{example}
  Let $\Sig$ be the signature with no operation name and one sort
  \texttt{s}, and $\Alg$ be the $\Sig$-algebra defined by
  $\Alg_{\texttt{s}}=\set{a,b}$. We consider the type monograph
  $T=\set{\tuple{e,\texttt{s}}, \tuple{\texttt{s},e}}$. A monograph
  typed by $T$ has any number (but at least one) of edges typed by $e$
  that must be adjacent either to $a$ or $b$, and two edges typed by
  $\texttt{s}$, namely $a$ and $b$, that must be adjacent to either
  the same edge $x$ typed by $e$, which yields two classes of
  monographs
  
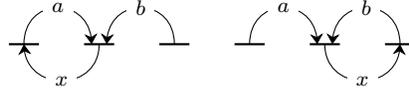
\begin{center}
    \begin{tikzpicture}[scale=0.5,> /.tip={Stealth[width=5pt, length=4pt]}]
      \draw[thick](0,0)--(0.8,0); \draw[thick](2,0)--(2.8,0); \draw[thick](4,0)--(4.8,0);
      \path[->,out=90,in=90,distance=1.3cm] (0.4,0) edge node[fill=white, font=\footnotesize] {$a$} (2.2,0); 
      \path[->,out=90,in=90,distance=1.3cm] (4.4,0) edge node[fill=white, font=\footnotesize] {$b$} (2.6,0); 
      \path[->,out=-90,in=-90,distance=1.3cm] (2.4,0) edge node[fill=white, font=\footnotesize] {$x$} (0.4,0); 
      \draw[thick](6,0)--(6.8,0); \draw[thick](8,0)--(8.8,0); \draw[thick](10,0)--(10.8,0);
      \path[->,out=90,in=90,distance=1.3cm] (6.4,0) edge node[fill=white, font=\footnotesize] {$a$} (8.2,0); 
      \path[->,out=90,in=90,distance=1.3cm] (10.4,0) edge node[fill=white, font=\footnotesize] {$b$} (8.6,0); 
      \path[->,out=-90,in=-90,distance=1.3cm] (8.4,0) edge node[fill=white, font=\footnotesize] {$x$} (10.4,0); 
    \end{tikzpicture}
  \end{center}
  (to which may be added any number of edges typed by $e$ and adjacent
  to either $a$ or $b$), or $a$ and $b$ are adjacent to $y$ and $z$
  respectively, and we get four more classes:
  
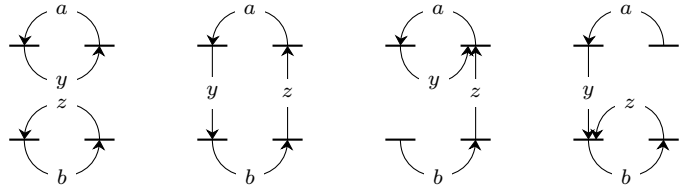
\begin{center}
    \begin{tikzpicture}[scale=0.5,> /.tip={Stealth[width=5pt, length=4pt]}]
      \draw[thick] (0.6,0) -- (1.4,0);
      \draw[thick] (-0.6,0) -- (-1.4,0);
      \draw[thick] (0.6,2.5) -- (1.4,2.5);
      \draw[thick] (-0.6,2.5) -- (-1.4,2.5);
      \path[->,out=90,in=90,distance=1.3cm] (1,2.5) edge node[fill=white, font=\footnotesize] {$a$} (-1,2.5); 
      \path[->,out=-90,in=-90,distance=1.3cm] (-1,2.5) edge node[fill=white, font=\footnotesize] {$y$} (1,2.5);       
      \path[->,out=-90,in=-90,distance=1.3cm] (-1,0) edge node[fill=white, font=\footnotesize] {$b$} (1,0);       
      \path[->,out=90,in=90,distance=1.3cm] (1,0) edge node[fill=white, font=\footnotesize] {$z$} (-1,0);
      \draw[thick] (5.6,0) -- (6.4,0);
      \draw[thick] (4.4,0) -- (3.6,0);
      \draw[thick] (5.6,2.5) -- (6.4,2.5);
      \draw[thick] (4.4,2.5) -- (3.6,2.5);
      \path[->,out=90,in=90,distance=1.3cm] (6,2.5) edge node[fill=white, font=\footnotesize] {$a$} (4,2.5); 
      \path[->] (4,2.5) edge node[fill=white, font=\footnotesize] {$y$} (4,0);       
      \path[->,out=-90,in=-90,distance=1.3cm] (4,0) edge node[fill=white, font=\footnotesize] {$b$} (6,0);       
      \path[->] (6,0) edge node[fill=white, font=\footnotesize] {$z$} (6,2.5) ;
      \draw[thick] (10.6,0) -- (11.4,0);
      \draw[thick] (9.4,0) -- (8.6,0);
      \draw[thick] (10.6,2.5) -- (11.4,2.5);
      \draw[thick] (9.4,2.5) -- (8.6,2.5);
      \path[->,out=90,in=90,distance=1.3cm] (11,2.5) edge node[fill=white, font=\footnotesize] {$a$} (9,2.5); 
      \path[->,out=-90,in=-90,distance=1.3cm] (9,2.5) edge node[fill=white, font=\footnotesize] {$y$} (10.8,2.5);       
      \path[->,out=-90,in=-90,distance=1.3cm] (9,0) edge node[fill=white, font=\footnotesize] {$b$} (11,0);       
      \path[->] (11,0) edge node[fill=white, font=\footnotesize] {$z$} (11,2.5) ;
      \draw[thick] (15.6,0) -- (16.4,0);
      \draw[thick] (14.4,0) -- (13.6,0);
      \draw[thick] (15.6,2.5) -- (16.4,2.5);
      \draw[thick] (14.4,2.5) -- (13.6,2.5);
      \path[->,out=90,in=90,distance=1.3cm] (16,2.5) edge node[fill=white, font=\footnotesize] {$a$} (14,2.5); 
      \path[->] (14,2.5) edge node[fill=white, font=\footnotesize] {$y$} (14,0);       
      \path[->,out=-90,in=-90,distance=1.3cm] (14,0) edge node[fill=white, font=\footnotesize] {$b$} (16,0);       
      \path[->,out=90,in=90,distance=1.3cm] (16,0) edge node[fill=white, font=\footnotesize] {$z$} (14.2,0);
    \end{tikzpicture}
  \end{center}
  The function $y,z\mapsto x$ is a morphism from these last two
  monographs to the two monographs above (respectively). There are no
  other morphisms between monographs from distinct classes. We
  therefore see that in the category $\ATM{T}{\Alg}$ it is possible to
  add or remove edges typed by $e$ to which $a$ or $b$ are not
  adjacent, but there is no way to remove the edges $y$ and $z$
  (because this would require a rule with a left morphism from an ATM
  without $y$ and $z$ to an ATM with $y$ and $z$, and there is no such
  morphism), though they are not attributes.

  Besides, we see that this category has no initial object, no
  terminal object, no products nor coproducts.
\end{example}

\section{Conclusion}\label{sec-concl}

Monographs generalize standard notions of directed graphs by allowing
edges of any length with free adjacencies. An edge of length zero
represents a node, and if it has greater length it can be adjacent to
any edge, including itself.  In ``monograph'' the prefix mono- is
justified by this unified view of nodes as edges and of edges with
unrestricted adjacencies that provide formal conciseness (morphisms
are functions characterized by a single equation); the suffix -graph
is justified by the correspondence (up to isomorphism) between finite
$\omega$-monographs and their drawings.

Monographs are universal with respect to graph structures and the
corresponding algebras, in the sense that monographs are equivalent to
graph structures extended with suitable ordering conventions on their
operator names, and that categories of typed monographs are equivalent
to the corresponding categories of algebras. Since many standard or
exotic notions of directed graphs can be represented as monadic
algebras, they can also be represented as typed monographs, but these
have two advantages over graph structures: they provide an orientation
of edges and they (consequently) dispense with operator names.

Algebraic transformations of monographs are similar to those of
standard graphs.  Typed monographs may therefore be simpler to handle
than graph structured algebras, as illustrated by the results of
Section~\ref{sec-atm}. The representation of oriented edges as
sequences seems more natural than their standard representation as
unstructured objects that have images by a bunch of functions. Thus
type monographs emerge as a natural way of specifying graph
structures.

\bibliographystyle{elsarticle-num}
{

}

\end{document}